\documentclass[prl,twocolumn,superscriptaddress,nofootinbib]{revtex4-1}
\usepackage[letterpaper, left=0.8in, right=0.8in, top=0.9in, bottom=0.9in]{geometry}
\usepackage{titlesec}

\usepackage{cmap} 
\usepackage[utf8]{inputenc}
\usepackage[english]{babel}
\usepackage[T1]{fontenc}
\usepackage{amsmath}
\usepackage[normalem]{ulem}
\usepackage{amsfonts,mathrsfs}
\usepackage{enumerate}
\usepackage[mathscr]{euscript} 
\usepackage{bm}
\usepackage{bbold}
\usepackage[table]{xcolor}
\definecolor{blueviolet}{rgb}{0.2, 0.2, 0.6}
\definecolor{webgreen}{rgb}{0,.5,0}
\definecolor{webbrown}{rgb}{.6,0,0}
\usepackage[pdftex,
	bookmarks=false,
	colorlinks=true, 
	urlcolor=webbrown, 
	linkcolor=blueviolet, 
	citecolor=webgreen,
	pdfstartpage=1,
	pdfstartview={FitH},  
	bookmarksopen=false
	]{hyperref}
\usepackage{tikz}
\usepackage{natbib}
\usepackage{mathtools}
\usepackage{graphicx}
\usepackage{csquotes}
\usepackage{braket}
\usepackage{dsfont}
\usepackage[all]{hypcap}

\usepackage{listings}
\allowdisplaybreaks
\DeclareMathOperator{\Expect}{\mathbb{E}}
\usepackage{dsfont}
\usepackage{changepage} 

\usepackage{multirow}
\usepackage{amssymb}

\usepackage{amsthm}

\usepackage{color}
\usepackage{nicefrac}

\DeclareFixedFont{\ttb}{T1}{txtt}{bx}{n}{9} 
\DeclareFixedFont{\ttm}{T1}{txtt}{m}{n}{9}  

\usepackage{color}
\definecolor{deepblue}{rgb}{0,0,0.5}
\definecolor{deepred}{rgb}{0.6,0,0}
\definecolor{deepgreen}{rgb}{0,0.5,0}

\newcommand\pythonstyle{\lstset{
language=Python,
basicstyle=\ttm,
morekeywords={self},              
keywordstyle=\ttb\color{deepblue},
emph={MyClass,__init__},          
emphstyle=\ttb\color{deepred},    
stringstyle=\color{deepgreen},
frame=tb,                         
showstringspaces=false
}}

\lstnewenvironment{python}[1][]
{
\pythonstyle
\lstset{#1}
}
{}

\newcommand\pythoninline[1]{{\pythonstyle\lstinline!#1!}}

\def\bra#1{\ensuremath{\mathinner{\langle{#1}|}}}
\def\ket#1{\ensuremath{\mathinner{|{#1}\rangle}}}

\newcommand{\norm}[1]{\left\lVert#1\right\rVert}

\newcommand{\expval}[1]{\langle #1\rangle}
\newcommand{\tr}{\text{tr}}
\newcommand{\Tr}{\text{tr}}

\newsavebox{\mstrut}
\newcommand{\bbra}[1]{%
    \sbox{\mstrut}{\(#1\)}%
    \mathinner{\left\langle\kern-0.5\ht\mstrut\left\langle{#1}\right|\mkern-2mu\right|}%
}
\newcommand{\kett}[1]{%
    \sbox{\mstrut}{\(#1\)}%
    \mathinner{\left|\mkern-2mu\left|{#1}\right\rangle\kern-0.5\ht\mstrut\right\rangle}%
}

\newtheorem{proposition}{Proposition}

\newtheorem{lemma}{Lemma}

\makeatletter
\newtheorem*{rep@proposition}{\rep@title}
\newcommand{\newrepproposition}[2]{%
\newenvironment{rep#1}[1]{%
 \def\rep@title{#2 \ref{##1}}%
 \begin{rep@proposition}}%
 {\end{rep@proposition}}}
\makeatother
\newrepproposition{proposition}{Proposition}

\makeatletter
\newtheorem*{rep@theorem}{\rep@title}
\newcommand{\newreptheorem}[2]{%
\newenvironment{rep#1}[1]{%
 \def\rep@title{#2 \ref{##1}}%
 \begin{rep@theorem}}%
 {\end{rep@theorem}}}
\makeatother
\newreptheorem{theorem}{Theorem}

\makeatletter
\newtheorem*{rep@definition}{\rep@title}
\newcommand{\newrepdefinition}[2]{%
\newenvironment{rep#1}[1]{%
 \def\rep@title{#2 \ref{##1}}%
 \begin{rep@definition}}%
 {\end{rep@definition}}}
\makeatother
\newrepproposition{definition}{Definition}

\newtheorem{theorem}{Theorem}
\newtheorem{corollary}{Corollary}
\newtheorem{definition}{Definition}

\newcommand{\equref}[1]{Eq.~\eqref{#1}}

\usepackage[noend]{algpseudocode}
\usepackage{algorithm,algorithmicx}

\algrenewcommand\alglinenumber[1]{\sf\scriptsize\color{black}{#1}}
\algrenewcommand\algorithmicrequire{\textbf{Input:}}
\algrenewcommand\algorithmicensure{\textbf{Output:}}

\usepackage[toc,page,header]{appendix}
\usepackage{minitoc}

\newif\ifptitle
\newif\ifpnumber
\newcounter{para}

\ptitletrue  
\pnumbertrue  

\begin{document}
\doparttoc 
\faketableofcontents 


\title{Derandomized shallow shadows: Efficient Pauli learning with bounded-depth circuits}
\date{\today}
\author{Katherine Van Kirk}
\affiliation{Department of Physics, Harvard University, Cambridge, MA 02138, USA}

\author{Christian Kokail}
\affiliation{Department of Physics, Harvard University, Cambridge, MA 02138, USA}
\affiliation{ITAMP, Harvard-Smithsonian Center for Astrophysics, Cambridge, Massachusetts, 02138, USA}

\author{Jonathan Kunjummen}
\affiliation{Joint Quantum Institute, University of Maryland, College Park, Maryland 20742, USA}
\affiliation{Joint Center for Quantum Information and Computer Science, University of Maryland, College Park,
Maryland 20742, USA}
\affiliation{National Institute of Standards and Technology, Gaithersburg, Maryland 20899, USA}

\author{Hong-Ye Hu}
\affiliation{Department of Physics, Harvard University, Cambridge, MA 02138, USA}

\author{Yanting Teng}
\affiliation{Department of Physics, Harvard University, Cambridge, MA 02138, USA}
\affiliation{Institute of Physics, Ecole Polytechnique Fédéderale de Lausanne (EPFL), CH-1015, Lausanne, Switzerland}

\author{Madelyn Cain}
\affiliation{Department of Physics, Harvard University, Cambridge, MA 02138, USA}

\author{Jacob Taylor}
\affiliation{Joint Quantum Institute, University of Maryland, College Park, Maryland 20742, USA}
\affiliation{Joint Center for Quantum Information and Computer Science, University of Maryland, College Park,
Maryland 20742, USA}
\affiliation{National Institute of Standards and Technology, Gaithersburg, Maryland 20899, USA}

\author{Susanne F. Yelin}
\affiliation{Department of Physics, Harvard University, Cambridge, MA 02138, USA}

\author{Hannes Pichler}
\affiliation{Institute for Theoretical Physics, University of Innsbruck, Innsbruck, 6020, Austria}
\affiliation{
Institute for Quantum Optics and Quantum Information of the Austrian Academy of Sciences, Innsbruck, 6020, Austria}

\author{Mikhail Lukin}
\email{lukin@physics.harvard.edu}
\affiliation{Department of Physics, Harvard University, Cambridge, MA 02138, USA}

\begin{abstract}

Efficiently estimating large numbers of non-commuting observables is an important subroutine of many quantum science tasks. 
We present the derandomized shallow shadow (DSS) algorithm for efficiently learning a large set of non-commuting observables, using shallow circuits to rotate into measurement bases.
Exploiting tensor network techniques to ensure polynomial scaling of classical resources, our algorithm outputs a set of shallow measurement circuits that approximately minimizes the sample complexity of estimating a given set of Pauli strings.
We numerically demonstrate systematic improvement, in comparison with state-of-the-art techniques, for energy estimation of quantum chemistry benchmarks and verification of quantum many-body systems, and 
we observe DSS's performance consistently improves as one allows deeper measurement circuits.
These results indicate that in addition to being an efficient, low-depth, stand-alone algorithm, DSS can also benefit many larger quantum algorithms requiring estimation of multiple non-commuting observables.

\end{abstract}

\maketitle

Recent advances in quantum science have led to the development of programmable quantum devices \cite{jurcevic2021demonstration,ryan2022implementing,ebadi2021quantum}, opening new avenues for applications in quantum chemistry \cite{hempel2018quantum,kandala2017hardware,google2020hartree}, 
materials science \cite{cerezo2021variational,bharti2022noisy}, and 
quantum optimization \cite{huang2021power,huang2019near, patel2024variational}. 
A significant milestone in this field is the emergence of digital hardware, including early fault-tolerant quantum devices \cite{anderson_2021, Postler2022, ryananderson2022implementing, Sivak2023, bluvstein_2023, acharya2024quantumerrorcorrectionsurface}, capable of implementing complex circuits with built-in error detection mechanisms. The increasing capability of these systems to prepare highly-complex quantum states necessitates the development of efficient protocols that probe their properties. 
Specifically, efficient estimation, of a large number of non-commuting observables, represents an essential subroutine for many near-term algorithms such as variational quantum optimization \cite{mcclean2018barren,commeau2020variational,cirstoiu2020variational,cerezo2021variational}, while at same time being vital for verification of quantum devices \cite{PRXQuantum.2.010102}. 

For measuring a given set of observables $\{P_k \}_{k = 1}^M$, our task is to find a set of efficiently implementable unitary transformations $\{U_i \}_{i = 1}^N$ such that each $U_i$ diagonalizes a subset of the observables $\{P_k\}$, enabling their direct measurement. In this Letter, we focus on Pauli observables, which are relevant for a large variety of applications \cite{li2020hamiltonian, huang2023learning, kokail2021quantum,cao2019quantum,colless2018computation,hempel2018quantum,huang2021near,mcardle2020quantum,mcclean2018barren,commeau2020variational,cirstoiu2020variational,cerezo2021variational} 
and can be efficiently diagonalized using Clifford circuits. While one could directly measure each $P_k$ one-at-a-time, various recent algorithms group Pauli strings into a small number of commuting sets 
\cite{yen2023deterministic, veltheim2024multiset, dalfavero2023k, wu2023overlapped},
i.e. finding some small set of unitary circuits $U_i$ that collectively diagonalize all observables. However, the depth of these circuits generally scales linearly with the system size \cite{aaronson2004improved}, making them challenging to implement with high fidelity due to the large number of two-qubit gates and their associated error rates \cite{evered2023high, google2023suppressing, ionq}.

Many modern strategies, that have proven practical on near term devices, are based on randomized measurements \cite{elben2022randomized, huang2020predicting,bertoni2022shallow,ippoliti2023operator, van2022hardware,king2024triply,hu2023classical,Bu24,matchgate,PhysRevResearch.6.043118,Ippoliti2024classicalshadows,contrastive_shallow,Bringewatt2024randomized}. Within the framework of classical shadow tomography \cite{huang2020predicting}, a set of snapshots $\{ \hat{\rho}_i \}$ of the quantum state $\rho$ is produced by processing the data acquired from sampling the state after applying a unitary $U_i$ randomly-drawn from an ensemble $\mathcal{U}$. Classical shadows can utilize these snapshots to estimate the expectation values of $M$ observables using only $\mathcal{O}(\log M)$ measurements \cite{huang2020predicting, king2024triply, chen2024optimal}. 
Moreover, recent investigations, into classical shadows protocols utilizing shallow quantum circuits $U_i$ \cite{bertoni2022shallow, ippoliti2023operator, hu2023classical}, demonstrate that shallow circuits can provide significant improvements in efficiency over single-qubit methods. 

\begin{figure*}[t]
\centering
\includegraphics[scale = 0.43]{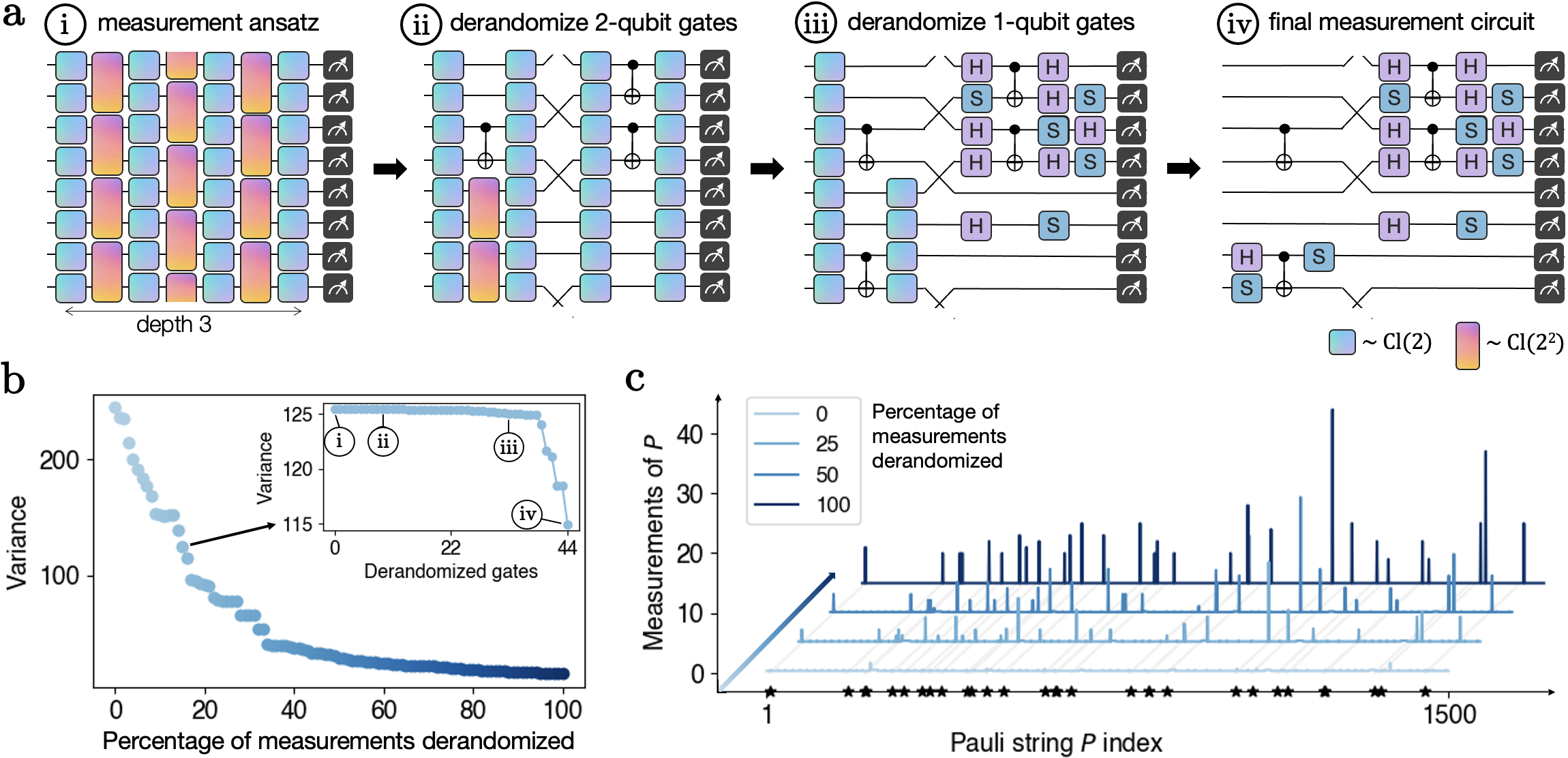}
\caption{\emph{Derandomized shallow shadows (DSS) algorithm.}
Given a measurement budget of $100$ measurements of depth $d=3$ and $30$ Pauli strings to learn, the classical DSS algorithm specifies the measurements we should make. 
(a) Derandomization of the $17$th measurement circuit:
each measurement circuit is initially $d$ layers of random Clifford gates (circuit i). The algorithm goes gate-by-gate, first fixing each random two-qubit gate (circuit ii), then fixing each single-qubit gate (circuit iii), and ending with a fully deterministic circuit (circuit iv) where all single qubit rotations are decomposed into S and H gates \cite{nielsen_chuang_2010}.  
(b) As this derandomization subroutine is performed on each of our $100$ measurement circuits, the variance of learning our $30$ Pauli strings decreases. 
The inset shows how the variance decreases with each gate that is fixed in the $17$th measurement. The greatest improvement comes when the final layer of gates is derandomized: suddenly we no longer have random gates and thus learn a fixed set of Paulis. 
(c) We show how many times, out of our 100 measurements, we probabilistically learn each of our 30 Paulis of interest (indices labeled with a star).  As we derandomize our circuit, our probability of learning each starred Pauli peaks.
}
\label{fig:Figure1}
\end{figure*}

While classical shadows can exploit short-range entanglement generated by shallow circuits, their inherent randomness causes them to probe the entire space of operators \cite{van2022hardware}, even though only a subset of observables is typically relevant. For many important applications, such as estimating the energy of many-body electronic Hamiltonians in quantum chemistry \cite{PhysRevResearch.4.033154,PhysRevX.8.011044,Cao} or Hamiltonian learning \cite{bairey2019learning, li2020hamiltonian}, the relevant Pauli strings are known in advance. In such cases, schemes tailored to specific observables would enhance efficiency and are thus highly desired. Although tailored approaches have been developed for schemes based on single-qubit rotations \cite{huang2021efficient,van2022hardware}, it remains an open challenge how to optimally leverage the short-range entanglement \cite{2024arXiv240217911H, schuster2024random,2023arXiv231211840D} produced by shallow circuits.

In this work
we propose a derandomization procedure for shallow shadows, demonstrating systematic improvements with increasing circuit depth across various Pauli estimation tasks. Our approach is based on optimizing shallow circuits to maximize the probability of learning a specific set of Pauli strings. We find effective, shallow circuits by tracking the Pauli probability distributions using tensor network techniques \cite{2024arXiv240217911H, bertoni2022shallow} that mimic classical Markovian processes.  
With the ablility to calculate these probabilities, we can assess a set of circuits $\{U_i\}_{i=1}^N$ with a cost function quantifying how often each of the Pauli strings is learned.  
We use this cost function to identify a good set of circuits. More 
specifically, our derandomization procedure begins with a view of the entire landscape of all possible $\{U_i\}_{i=1}^N$ circuit choices, and then iteratively hones in on regions that are effective for learning the specific set of Pauli strings. 
As discussed below, we find our efficient procedure outperforms all previous state-of-the-art bounded-depth learning strategies.

\vspace{1mm}
\textit{Derandomized Shallow Shadows Algorithm.}
Our derandomized shallow shadows (DSS) algorithm systematically determines what shallow circuits efficiently learn a chosen set of Pauli strings. 
Given a measurement budget $N$, it outputs shallow measurement circuits $\{U^\text{DSS}_i\}_{i=1}^N$. 
To illustrate the key idea,
consider the following example learning problem: estimate $30$ Pauli strings $\{P\}$ on $n=8$ qubits using $N=100$ measurements with depth at most $d=3$. 
The $30$ Pauli strings we consider are randomly chosen and listed in Appendix \ref{app:DSSalgorithm}. 
For each measurement $i \in \{1,...,N\}$, we start with a shallow shadows ensemble $\mathcal{U}_i$ \cite{bertoni2022shallow, ippoliti2023operator, hu2023classical}. 
These ensembles are defined by a circuit ansatz with $d=3$ layers of two-qubit gates, interleaved with single-qubit rotations, and each gate is uniformly sampled from the local Clifford group 
(see circuit (i) in Figure \ref{fig:Figure1}(a)).
We derandomize each measurement by  
replacing the ensemble $\mathcal{U}_i$ with a properly-chosen, fixed unitary $U_i^\text{DSS}$. 
To choose $U_i^\text{DSS}$ we sequentially fix each Clifford gate in measurement $i$’s ansatz such that it is no longer randomly sampled.
The two-qubit gates become \textsf{Identity}, \textsf{CNOT} or \textsf{SWAP}, and the single qubit rotations become one of the 6 independent rotations -- up to a phase \cite{sixindependentchoices} -- in the group of single-qubit Cliffords, $\text{Cl}(2)$. 
From these options, we choose the gate that minimizes our cost function, 
\begin{equation} \label{eqn:costfunction_maintext}
    \textnormal{\textsf{COST}}_\epsilon 
    \left(\{\mathcal{U}_i\}_{i=1}^N\right)
    = \sum_P w_P \prod_{i=1}^{N} \exp{\left[ \hspace{1mm} - \frac{\epsilon^2}{2} 
    p_i\left(P\right)
    \hspace{1mm} \right]}.
\end{equation}
Here, $\epsilon$ is a hyperparameter controlling the magnitude of our \textsf{COST} function landscape. 
The weights $w_P$ indicate relative importance of each Pauli $P$ ($w_P = 1$ for all $P$ in Figure \ref{fig:Figure1}), and $p_i(P)$ is the probability that the $i$th measurement circuit $U \sim \mathcal{U}_i$ diagonalizes $P$,
\begin{equation}
    p_i\left(P\right) 
    = \frac{1}{2^n} \Expect_{U\sim \mathcal{U}_i} \sum_{b \in \{0,1\}^n}  \bra{b} U P U^\dagger \ket{b}^2.
\end{equation}
Minimizing \textsf{COST} corresponds to maximizing the  probability each Pauli string is learned, weighted by their importance $w_P$.
Figure \ref{fig:Figure1}(a) shows snapshots of the $i=17$th measurement circuit’s derandomization; as expected, the final circuit (circuit (iv)) no longer contains randomly-sampled gates.  
We apply the derandomization procedure sequentially to each of our 100 measurement ensembles $\{\mathcal{U}_i\}_{i=1}^{100}$, ending with 100 delta distributions that determine the final $100$ measurement circuits $\{U^\text{DSS}_i\}_{i=1}^{100}$. 
Each measurement circuit corresponds to a single shot of our protocol, and $\{U^\text{DSS}_i\}_{i=1}^{100}$ could contain the same circuit multiple times. 
We then implement these measurement circuits on our quantum simulator and obtain bit strings $\{b_i\}_{i=1}^{100}$, where $b_i \in \{0,1\}^n$. 
Since our measurement circuits are Clifford rotations, we can use these results to efficiently estimate our Pauli strings with the empirical average 
\begin{equation}
\label{eqn:estimator_maintext}
    \hat{o}(P) = \frac{1}{N_P} \sum_i \bra{b_i} U^\text{DSS}_i P U^{\text{DSS}\dagger}_i \ket{b_i},
\end{equation}
which provides an estimator for the exact expectation value $\expval{P}$.
Here, $N_P$ is the number of DSS unitaries that diagonalize $P$.

\textit{Precision and Efficiency.}
Minimizing the \textsf{COST} function of \equref{eqn:costfunction_maintext} allows us to find effective measurement circuits $\{U_i^\text{DSS}\}_{i=1}^N$ for our learning problem. 
The \textsf{COST} function value bounds the confidence with which we estimate our $30$ Pauli strings of interest and, therefore, leads to precision guarantees in learning Pauli expectation values (see Appendix \ref{app:Performance_guarantees}).
This was first pioneered in Ref \cite{huang2021efficient} for local Pauli rotations, and here we generalize it to shallow circuits, in the form of the following: 
\begin{theorem}
    \textbf{(informal)}
    Given Pauli strings $\{P\}$, fix some desired precision $\epsilon \in (0,1)$.  
    Using the measurement circuits from DSS's final delta distributions $\{\mathcal{U}_i\}_{i=1}^N$,
    the estimates $\hat{o}(P)$ 
    all achieve precision
     \begin{equation}
          |\expval{P} - \hat{o}(P)| \leq \epsilon \hspace{5mm} \forall P \in \{P\}
     \end{equation}
     with probability (at least) $1-\delta$, where the final unweighted cost function $\textnormal{\textsf{COST}}_\epsilon
     \left(\{\mathcal{U}_i\}_{i=1}^N\right) \leq \delta/2$.

\end{theorem}

\noindent Crucially, the protocol's final \textsf{COST} value immediately provides a precision guarantee across our estimates. 
While this is recapitulated with rigorous analysis in Appendix~\ref{app:Performance_guarantees}, the guarantee  follows from the confidence’s relationship with variance.  
A central metric of any learning protocol,
the variance of an estimated observable represents how efficiently we learn it 
-- for example, a smaller variance indicates that, across the measurements made, the observable is often learned. 
Moreover, we can see the variance improve in our example learning problem: in Figure \ref{fig:Figure1}(b) the variance 
decreases as we derandomize each of our 100 measurements.
The variance plotted is the convex combination of the variance of our $30$ individual Pauli strings, and since we do not assume prior knowledge of the underlying state, each observable’s variance is maximized over all states.

Moreover, consider the $\textsf{COST}$ function landscape over all final $N$ measurement circuit choices $\{U^\text{DSS}_i\}_{i=1}^{N}$. 
Each ensemble $\mathcal{U}_i$ starting in a statistical mixture over depth-$d$ Clifford rotations (Figure \ref{fig:Figure1}(a)) allows the algorithm to ``see’’ this entire landscape. Then, as we derandomize each ensemble, our statistical mixture restricts to a smaller and smaller region of the landscape. This continues until all gates in all ensembles have been fixed, and we converge to one point in the landscape.  
See Figure \ref{fig:Figure1}(c) to affirm this intuition. 
When each ensemble begins as a statistical mixture over Clifford rotations, we can learn any Pauli string with low but nonzero probability (lightest blue line). 
Then, as we derandomize each measurement, we become probabilistically more likely to learn our $30$ Paulis of interest. 
Once all gates in all ensembles are derandomized, we then have $N$ fixed measurement circuits, which often learn our chosen Pauli strings (darkest blue line). 
In Appendix A we give examples where DSS finds the global minimum of the $\textsf{COST}$ landscape, and
in Appendix~\ref{app:Performance_guarantees} 
we formalize the improvement on shallow shadows:  
we show a slightly modified version of DSS will always perform at least as well as shallow shadows \cite{bertoni2022shallow, ippoliti2023operator, hu2023classical}.

Finally, in order to implement DSS at scale, we must be able to efficiently evaluate our \textsf{COST} function, which requires computing the probabilities $p_i(P)$. 
This proves nontrivial: one must track how our (sometimes probabilistic) circuits $U\sim \mathcal{U}_i$ transform $P$, which via direct simulation requires space exponential in the system size. 
However, we bypass this overhead and utilize resources only polynomial in the system size via tensor network techniques, which mimic random Markovian classical processes and track how our probabilistic Clifford rotations transform our Paulis. 
See Appendix~\ref{app:TensorNetwork_for_cost_function} for a pedagogical description of our technique, which builds on ideas introduced in 
Refs \cite{2024arXiv240217911H,bertoni2022shallow}.
We also prove this technique's correctness and efficiency, giving the following guarantee: 
\begin{theorem} \textbf{(informal)}
    Given a set $\{P\}$ of Pauli strings on $n$ qubits and a measurement budget of $N$ at most $\mathcal{O}(\textnormal{poly}\log(n))$ depth circuits, the DSS algorithm requires time $\mathcal{O}(|\{P\}| \times N \times \textnormal{poly}(n))$. 
\end{theorem} 
\noindent 
Crucially, when our measurements have bounded depth $d=\mathcal{O}(\textnormal{poly}\log(n))$, the DSS algorithm's time complexity scales polynomially with system size, enabling DSS to be used at scale.

\begin{figure}[t]
\centering
\includegraphics[scale = 0.40]{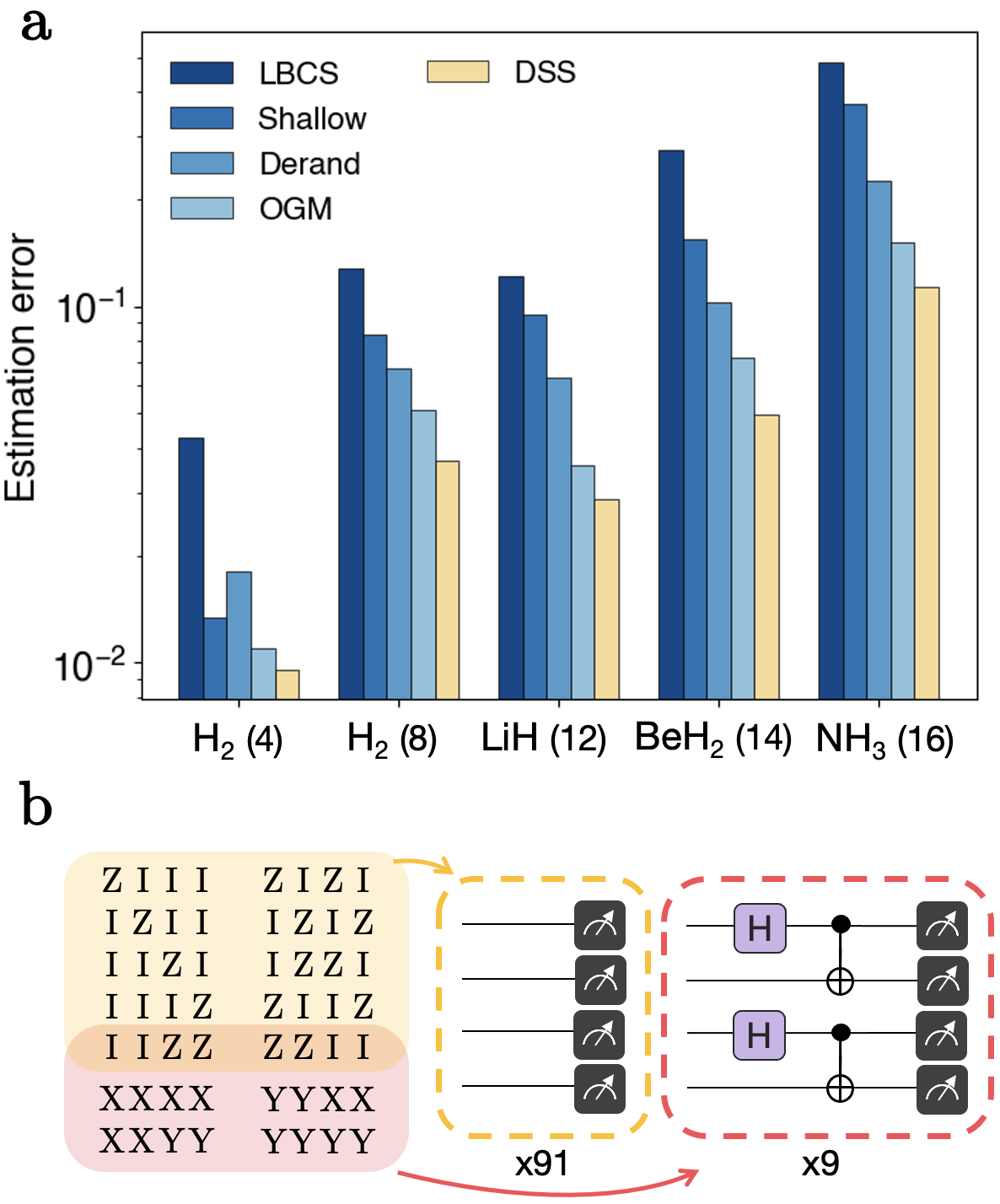}
\caption{\emph{Efficient ground state energy estimation under various quantum chemistry Hamiltonians.}  (a) With only one layer of two-qubit gates ($d=1$), DSS already outperforms previous state-of-the-art techniques for estimating the ground state energy of various molecules (parentheses indicate number of qubits). We plot the estimation error after 1000 measurements, and to avoid fluctuations due to non-trivial variance, the estimation errors reported are averaged over many simulations.  We compare to locally-biased classical shadows (“LBCS”) \cite{hadfield2022measurements}, Random Pauli derandomization (“Derand”) \cite{huang2021efficient}, shallow shadows (“Shallow”) \cite{bertoni2022shallow, ippoliti2023operator, hu2023classical}, and overlapped grouping measurement (“OGM”) \cite{wu2023overlapped}. 
Notice that we consider estimating the ground state energy of $H_2$ with two different representations: 4 and 8 qubits. This corresponds to how one chooses different sets of molecular orbitals when discretizing real space. 
(b) Measurement circuits ($d =1$) for 100 measurements on $H_2$ (4 qubits). We find that  groupings naturally emerge: our DSS algorithm groups Paulis which are simultaneously-diagonalizable under shallow circuits. For 100 measurements, DSS  suggests taking $91$ measurements in the $Z$ basis (yellow) and $9$ in the bell basis (pink). 
}
\label{fig:Figure2}
\end{figure}

\vspace{1mm}
\textit{Efficient energy estimation of quantum chemistry Hamiltonians.}
In quantum chemistry and material science, the electronic structure problem's goal is to estimate the energy of many-body electronic Hamiltonians, whose electronic structure is unknown.
These complicated Hamiltonians with $\mathcal{O}(K)$ fermionic orbitals usually contain $M = \mathcal{O}(K^4)$ terms \cite{PhysRevX.8.011044,Cao}. 
This can manifest as a large number of terms in practice, making direct measurement of each term time-consuming and inefficient \cite{PhysRevResearch.4.033154}. 
Moreover, many quantum algorithms, especially those tailored for near-term quantum devices, require estimating the Hamiltonian -- among other complicated observables -- as a frequent subroutine \cite{2023VQLS,PhysRevX.8.011021,cerezo2021variational,2019arXiv190713623G,adaptiveVQE,Cao,2024arXiv240803987R}. 
Therefore, if not done efficiently, this task can become a bottleneck for not only the electronic structure problem, but also many near-term algorithms. 
We apply DSS to this problem. 

In particular, we estimate the ground state energy of many quantum chemistry Hamiltonians. 
Multiple molecules \cite{charlesgitrepo} have been commonly used in the literature to benchmark different learning strategies \cite{hadfield2022measurements, huang2021efficient,wu2023overlapped}. 
While the ultimate goal is to address molecules with \textit{unknown} ground state energies, the small, classically-solvable Hamiltonians serve as effective benchmarks, allowing us to estimate our procedure's precision under a fixed budget. Therefore, we examine these molecules to compare DSS with previous state-of-the-art strategies.
First, we apply the Jordan-Wigner transformation \cite{JWtransform} to map the second-quantized electronic Hamiltonian to a qubit Hamiltonian $H = \sum_P c_P P$. We then input both the Pauli strings and their corresponding coefficients into the DSS algorithm. These coefficients are used to weight the \textsf{COST} function: $w_p = |c_P|$ in \equref{eqn:costfunction_maintext}. 
As illustrated in Figure \ref{fig:Figure2} (a), our DSS protocol demonstrates state-of-the-art performance in accurately estimating the energy of quantum chemistry Hamiltonians across all benchmarked molecules. With just a single layer of two-qubit gates in our measurement circuits, DSS achieves the smallest estimation error under a fixed measurement budget of $M = 1000$ measurements. A comprehensive description of the previous algorithms with bounded circuit depth can be found in Appendix \ref{app:QuantumChemistry}.

To build intuition on how the protocol assigns measurement bases, consider estimating the terms of $H_2$. 
In Figure \ref{fig:Figure2}(b), we list all the terms in $H_2$'s 4-qubit representation. 
With a measurement budget of $N=100$ depth-$1$ measurements, our DSS technique identifies two measurement settings: 91 measurements using the yellow circuit and 9 using the pink one. 
This distribution of measurements arises because the Pauli $Z$ strings have much larger coefficients than the strings containing $X$ and $Y$ (see Appendix \ref{app:QuantumChemistry}). 
Many previous strategies in the literature \cite{wu2023overlapped, verteletskyi2020measurement, yen2023deterministic, veltheim2024multiset} focus on identifying ``good groups'' of Pauli operators that locally mutually commute. While DSS is not explicitly a grouping strategy, this example illustrates its similarity to such approaches. 
The shaded regions indicate which Pauli strings are simultaneously measured by each circuit.
Moreover, for a budget of $N=100$ depth $d=1$ measurements, the $100$ measurement settings chosen by DSS turn out to be optimal (see Appendix \ref{app:QuantumChemistry}). 
Note that, while we empirically show this for the $H_2$ example, in general finding the optimal measurement strategy is NP-hard \cite{jena2019pauli}.

\vspace{3mm}
\textit{Quantum verification by estimating the energy variance.}
Most verification strategies such as process tomography \cite{PhysRevLett.130.150402}, cross-platform verification \cite{PhysRevLett.124.010504} or Hamiltonian learning \cite{Qi2019determininglocal, PhysRevLett.122.020504, eisert2020quantum} rely on estimating a set of non-commuting observables, whose number rapidly grows with increasing system size.
In this section, we explore combinations of  quantum verification \cite{PRXQuantum.2.010102} strategies with DSS, aiming to improve the sampling complexity with increasing system size. We will focus specifically on measuring the energy variance in ground states of quantum many-body systems. As shown in Ref.~\cite{kokail2019self}, the Hamiltonian variance $\text{var}[H] = \braket{H^2} - \braket{H}^2$ can serve as an algorithmic error bar to quantify the success of the state preparation, vanishing as the state approaches an eigenstate of $H$. Ref.~\cite{huang2020predicting, huang2021efficient} has demonstrated scaling improvement, over naive grouping strategies, by derandomizing classical shadows built from single-qubit measurements. Here we build on these developments, investigating the effect of incorporating finite correlations via DSS. We note that the variance bounds the spectral distance to the closest energy eigenstate $\ket{\psi_{\ell}}$: $|\braket{H} - \varepsilon_{\ell}|^2 \le \text{var}[H]$, and thus serves for verification when measured to a precision below the energy gap. This is especially challenging in systems with small gaps, such as near quantum phase transitions, demanding highly precise schemes.

\begin{figure}[t]
\centering
\includegraphics[scale = 0.234]{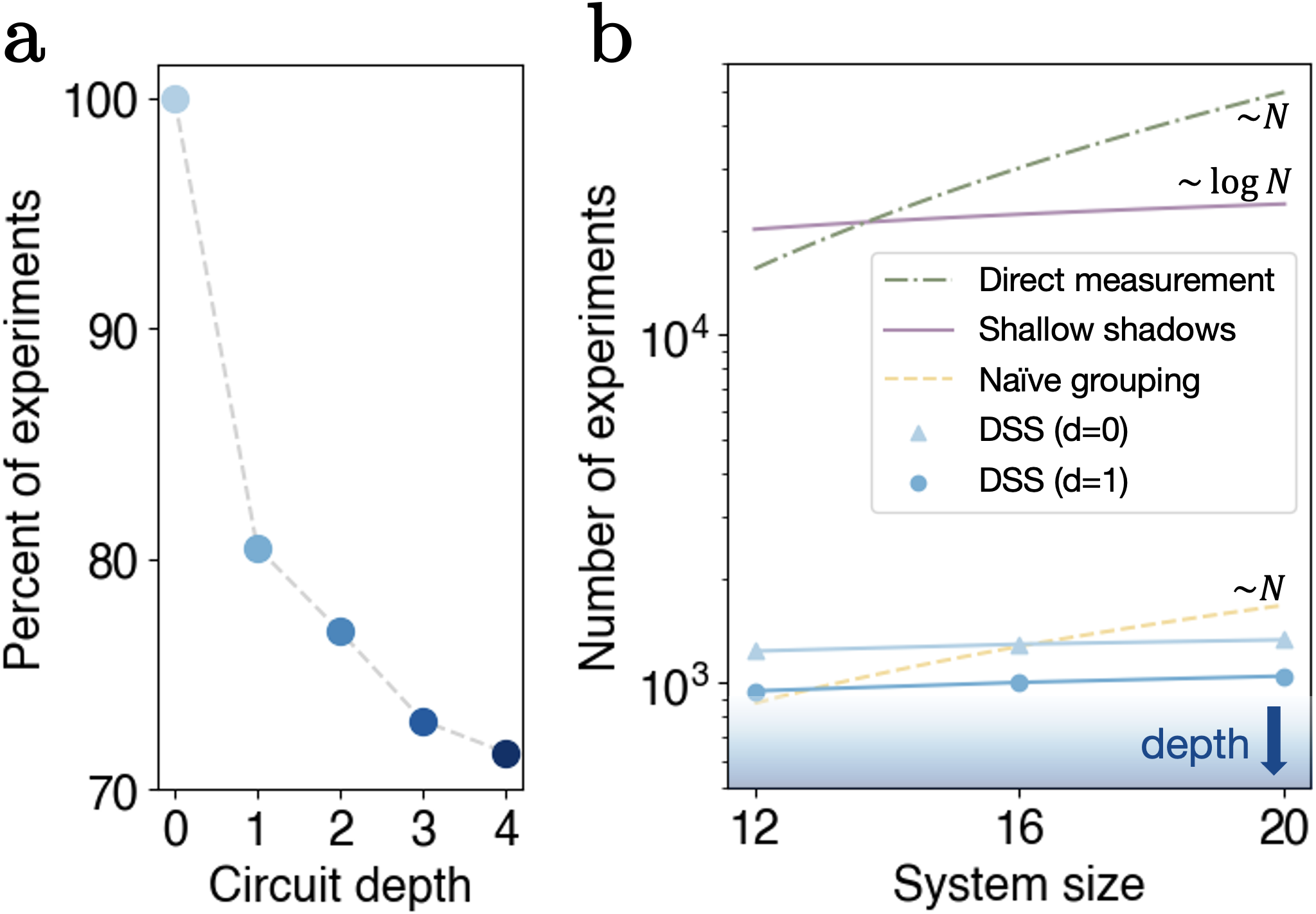}
\caption{\emph{Variational quantum simulation of the Hubbard model.} 
We consider estimating the variance of the Hubbard Hamiltonian at various system sizes, which requires estimating the terms of $H^2$.  In this figure we plot the number of measurements required to learn each term in $H^2$ at least 25 times. 
(a) Here, we consider the Hubbard model on 12 qubits. We normalize the number of experiments by the number required at $d=0$ (1236 experiments). Each point is therefore the percentage, of the number of measurements required at $d=0$, we need as we increase our measurement ansatz depth. 
As we increase the depth, we require fewer measurements, as expected, because we can simultaneously learn more of our Paulis in $H^2$.
(b) We also fix the depth and modulate the system size. We find that at large system sizes our DSS protocol is most efficient -- in particular, DSS with only one layer of two-qubit gates ($d=1$) already outperforms all other strategies at 13 qubits. 
We expect that as we increase the depth of our ansatz, the required number of experiments will continue to decrease -- see blue gradient.}
\label{fig:Figure3}
\end{figure}

We numerically test our approach on a 1D Hubbard chain of spinful fermions, described by the Hamiltonian
\begin{align} \label{eq:hubbard}
    H = -J \sum_{j = 1}^{L-1} \sum_{\sigma} \left( c_{j \sigma}^{\dagger} c_{j+1\, \sigma} + \text{H.c.} \right) + U \sum_{j=1}^L n_{j \uparrow} n_{j \downarrow}.
\end{align}
Here $c_{j \sigma}^{\dagger}$ ($c_{j \sigma}$) denote creation (annihilation) operators of fermions with spin $\sigma$ on lattice site $j$, and $n_{j \sigma} = c_{j \sigma}^{\dagger}c_{j \sigma}$ denotes the local particle density. The system is described by the nearest-neighbor tunneling amplitude $J$ and onsite interaction energy $U$ for two fermions occupying the same lattice site. Mapping the fermion chain to qubits by doubling the lattice and encoding the different spin components on even (odd) sites, the hopping term maps to a Pauli of weight 3:  $\left(\sigma_{j-1}^+ Z_j  \sigma_{j+1}^- + \text{H.c}\right) $. The squared Hamiltonian $H^2$ thus contains weight-6 Pauli strings and the number of Paulis scales quadratically with system size, i.e. for 100 sites $H^2$ consists of $2.4\cdot 10^5$ terms.

We start by studying the performance of DSS for estimating the energy-variance of a 12-qubit Hubbard chain by modulating the depth of the measurement circuit. Figure \ref{fig:Figure3}(a) shows the reduction in the number of required measurements as a function of circuit depth. As we increase the depth, we can simultaneously diagonalize larger sets of Pauli strings, leading to a ~30 \% reduction in the number of required experiments to estimate the Pauli strings up to a given precision.
On the contrary, shallow shadows has some optimal depth, after which recovering small weight observables becomes exponentially costly. 
Due to shallow shadows' randomized layers, the desired information gets scrambled across a larger system from which it has to be recovered probabilistically. 
Therefore, DSS's improvement with depth demonstrates that we utilize the extra layers only when they actually facilitate measuring more observables in parallel.

In Figure \ref{fig:Figure3}(b) we compare the performance of DSS to optimal-depth shallow shadows (see Appendix~\ref{app:HubbardModel} for details) \cite{bertoni2022shallow, ippoliti2023operator, hu2023classical,2024arXiv240217911H}, a na\"ive grouping strategy akin to what was done in Ref \cite{kokail2019self}, and the standard approach that directly measures all observables independently. Again we plot, for the methods considered, the number of experiments 
needed to measure the expectation value of all Pauli observables with an error equivalent to measuring each of these observables at least 25 times. We find that after N=12, our DSS algorithm at depth $d=1$ is already significantly more efficient than all other strategies (note: $d=0$ is equivalent to the Random Pauli derandomization strategy of Ref \cite{huang2021efficient}), and its performance should only improve as we allow higher-depth measurement circuits.

\textit{Discussion.}
Ideal for near-term learning applications, our bounded-depth DSS approach is both efficient and tailorable to the constraints of a wide variety of quantum simulation platforms.
It already outperforms previous state-of-the-art Pauli learning techniques such as OGM \cite{wu2023overlapped}, with just a single layer of two-qubit gates, and we observe that performance consistently improves with access to more layers.  
Moreover, DSS is not only valuable for near-term quantum processors, but also in the fault-tolerant regime. Our measurement circuits only require a Clifford gate set, which can be implemented with low overhead using transversal gates in many quantum error correction codes \cite{steane1996quantum, Dennis2002} such as the surface code and two-dimensional color codes~\cite{bombin2006topological}.
Such transversal gates can be naturally implemented in reconfigurable architectures such as neutral atom arrays~\cite{bluvstein_2023} and trapped ions~\cite{Postler2022, ryananderson2022implementing, menendez2023implementing, wang2023faulttolerant, dasilva2024demonstrationlogicalqubitsrepeated, anderson_teleportation_2024}.
Relevant in both the noisy intermediate-scale quantum (NISQ) era and beyond, DSS can be utilized for any application -- from materials science \cite{cerezo2021variational,bharti2022noisy} to quantum chemistry \cite{hempel2018quantum,kandala2017hardware,google2020hartree} to optimization \cite{huang2021power,huang2019near} -- that requires learning a set of Pauli observables as a subroutine.

Our work can be extended in a number of directions. First, one could design better optimization algorithms for the derandomization procedure. By replacing our greedy optimization with simulated annealing, we could both escape local minima and ensure guaranteed convergence to the optimal measurement strategy in the infinite iteration limit \cite{kirkpatrick1983optimization}. 
One could also consider improving the optimization using machine learning. With generative or language models \cite{wu2024optimization}, instead of fixing individual gates, we would define a probability distribution over all possible circuits \cite{van2022hardware} and refine this distribution to our learning problem. 
Each Clifford gate could be viewed as a ``word'' and each measurement circuit as a ``sentence.'' 
The language model could learn the optimal probability of each ``word'' and ``sentence'' by minimizing the cost function \cite{2024arXiv241016041H}. 
Second, one could modify the ansatz of our derandomization technique by interleaving ancillas among the ``data qubits’’ of our original ansatz. Entanglement of these ancillas with the data qubits -- and then subsequent measurement -- models deeper measurement circuits on just the data qubits \cite{iqbal2024topological}.  
To build intuition for this, consider that preparing a GHZ state on $n$ qubits traditionally requires a depth O($n$) circuit using only nearest-neighbor interactions \cite{nielsen_chuang_2010}. 
Yet, incorporating $n$ ancillas facilitates this in only constant depth \cite{watts2019exponential}. Applied to DSS this trick allows us to consider deeper measurement circuits, while bounding the classical overhead for derandomizing these circuits.

\vspace{3mm}
\textit{Acknowledgements} 
We would like to thank Robert Huang, Simon Evered, Yi-Zhuang You and Jorge Garcia Ponce for useful discussions, and we thank Victor Albert, Srilekha Gandhari, Michael Gullans, and Eric Benck for helpful comments on the manuscript. 
We acknowledge financial support from and the U.S. Department of Energy [DE-SC0021013 and DOE Quantum Systems Accelerator Center (Contract No.: DE-AC02-05CH11231)], the National Science Foundation (grant number PHY-2012023), the Center for Ultracold Atoms (an NSF Physics Frontiers Center), the DARPA ONISQ program (grant number W911NF2010021), and  the DARPA IMPAQT program (grant number HR0011-23-3-0030). 
KVK acknowledges support from the Fannie and John Hertz Foundation and the National Defense Science and Engineering Graduate (NDSEG) fellowship. 
CK acknowledges support from the NSF through a grant for ITAMP at Harvard University.
Funding for JK provided by the NSF-funded NQVL:QSTD: Pilot: Deep Learning on Programmable Quantum Computers. 
HYH acknowledges the support from Harvard Quantum Initiative. 
HP acknowledges funding from European Union’s Horizon Europe research and innovation program under Grant Agreement No.~101113690 (PASQuanS2.1), the ERC Starting grant QARA (Grant No.~101041435), the EU-QUANTERA project TNiSQ (N-6001), the Austrian Science Fund (FWF) (Grant No.~DOI 10.55776/COE1).

\newpage
\bibliography{references}

\begin{thebibliography}{10}
\bibitem{sixindependentchoices}
While the size of the single-qubit Clifford group is 24, we only need to
  consider 6 independent rotations in $\text{Cl}(2)$ because we do not care
  about the phase. Up to a phase (4 choices for phase: $+1$, $-1$, $+i$, and
  $-i$), there are only 6 independent rotations.
\end{thebibliography}
\bibliographystyle{ieeetr} 



\onecolumngrid  
\appendix

\begin{appendix}
\clearpage 

\vspace{2.0em}
\begin{center}
\textbf{\Large Derandomized Shallow Shadows Appendices}
\end{center}

\renewcommand{\appendixname}{APPENDIX}
\renewcommand{\thesubsection}{\MakeUppercase{\alph{section}}.\arabic{subsection}}
\renewcommand{\thesubsubsection}{\MakeUppercase{\alph{section}}.\arabic{subsection}.\alph{subsubsection}}
\makeatletter
\renewcommand{\p@subsection}{}
\renewcommand{\p@subsubsection}{}
\makeatother

\renewcommand{\figurename}{Supplementary Figure}
\setcounter{secnumdepth}{3}
\makeatletter
     \@addtoreset{figure}{section}
\makeatother

\vspace{3mm}
\begin{center}
    \textbf{Summary}
\end{center}

In this work we present a protocol for the following learning task: given a set of $n$-qubit Pauli strings $\{P\}$ to estimate and a budget of $N$ depth-$d$ measurements, what are the best measurements to make?  
Crucially, this novel learning scheme is designed for the regime where we want to implement non-trivial -- but still bounded-depth -- measurement circuits. 
For example, DSS allows one to find effective measurements for devices that can implement $d \sim O(\log(n))$ layers of multi-qubit gates but where a high-fidelity circuits of depth $d \sim O(n)$ are still intractable. 
Here we define a ``measurement circuit'' as the unitary that rotates our system into some desired measurement basis, and we assume these bounded-depth measurement circuits can have \textit{at most} $d$ layers of two-qubit gates.   

Our derandomized shallow shadows (DSS) protocol can be decomposed into two steps. First, given the specified maximum measurement circuit depth $d$ and set of Pauli strings $\{P\}$ whose expectation values we want to estimate, we determine what $N$ measurements to make. In the main text, we refer to this as our DSS ``algorithm'', and Figure 1 shows how the DSS algorithm determines what measurements to make for the task of estimating $30$ (randomly-chosen) Pauli strings.  Once this algorithm outputs what measurements we should make, the second step is to actually implement these measurements on the state we've prepared on our quantum computer. 
Therefore, this protocol has two computations -- one classical, one quantum -- where the classical one determines what measurements to implement and the quantum one implements the measurements.

In these appendices we first describe our DSS algorithm (appendix \ref{app:DSSalgorithm}). 
This first appendix recapitulates much of what was discussed in the main text but does so with full technical rigor. 
We then describe how we estimate the expectation values of our desired Pauli strings (appendix \ref{app:estimating_exp_vals_with_DSS}), derive rigorous performance guarantees for our learning protocol (appendix \ref{app:Performance_guarantees}), and discuss our tensor network approach to efficiently evaluating our algorithm's \textsf{COST} function (appendix \ref{app:TensorNetwork_for_cost_function}). 
In the remaining appendices, we provide background and further context on the two applications we explored in the main text. First, we describe our quantum chemistry application and how we benchmark our DSS protocol against other schemes, finding it already outperforms previous state-of-the-art at depth $d=1$ (appendix \ref{app:QuantumChemistry}). Then, we provide further details our second application: estimating the variance of the Fermi-Hubbard Hamiltonian (appendix \ref{app:HubbardModel}). The learning problem setup for this application requires a modification to the original DSS algorithm. We discuss this modification in detail and describe the other protocols with which we compare DSS's performance.

\vspace{5mm}

\begin{center}
    \textbf{Table of Contents}
\end{center}

\smallskip

\noindent \textbf{\ref{app:DSSalgorithm}.~~\hyperref[app:DSSalgorithm]{THE DSS ALGORITHM: DETERMINING THE MEASUREMENT CIRCUITS}} \dotfill\textbf{\pageref{app:DSSalgorithm}}
\medskip

\noindent \qquad \begin{minipage}{\dimexpr\textwidth-1.3cm}
  \hyperref[app:DSSalgorithm]{How the DSS algorithm works}
 $\bullet$ 
\hyperref[app:DSSalgorithm]{The DSS \textsf{COST} function}
 $\bullet$ 
\hyperref[app:DSSalgorithm]{Example from the main text: 30 randomly-chosen Pauli strings}
 \end{minipage}
\medskip

\noindent \textbf{\ref{app:estimating_exp_vals_with_DSS}.~~\hyperref[app:estimating_exp_vals_with_DSS]{ESTIMATING EXPECTATION VALUES WITH DSS MEASUREMENTS}} \dotfill\textbf{\pageref{app:estimating_exp_vals_with_DSS}}
\medskip

\noindent \textbf{\ref{app:Performance_guarantees}.~~\hyperref[app:Performance_guarantees]{DSS PERFORMANCE GUARANTEES}} \dotfill\textbf{\pageref{app:Performance_guarantees}}
\medskip

\noindent \qquad \begin{minipage}{\dimexpr\textwidth-1.3cm}
  \hyperref[app:Performance_guarantees]{Deriving our DSS \textsf{COST} function and its guarantees}
 $\bullet$ 
\hyperref[app:bounding_precision_DSS]{Bounding the precision of each Pauli string under the DSS protocol}
 $\bullet$ 
\hyperref[app:COSTrelationship_with_VARIANCE]{The \textsf{COST} function’s relationship with variance}
 $\bullet$ 
\hyperref[app:asgoodas_shallow_shadows]{DSS performance at least as good as equivalent-depth shallow shadows}
 \end{minipage}
\medskip

\noindent \textbf{\ref{app:TensorNetwork_for_cost_function}.~~\hyperref[app:TensorNetwork_for_cost_function]{TENSOR NETWORKS FOR EFFICIENT COST FUNCTION EVALUATION}} \dotfill\textbf{\pageref{app:TensorNetwork_for_cost_function}}
\medskip

\noindent \qquad \begin{minipage}{\dimexpr\textwidth-1.3cm}
 \hyperref[app:TensorNetwork_for_cost_function]{Tensor network technique for estimating Pauli weights}
 $\bullet$ 
 \hyperref[app:TensorNetwork_for_cost_function]{The DSS algorithm is efficient}
\end{minipage}
\medskip

\noindent \textbf{\ref{app:QuantumChemistry}.~~\hyperref[app:QuantumChemistry]{QUANTUM CHEMISTRY NUMERICAL SIMULATIONS}} \dotfill\textbf{\pageref{app:QuantumChemistry}}
\medskip

\noindent \qquad \begin{minipage}{\dimexpr\textwidth-1.3cm}
\hyperref[app:QuantumChemistry]{DSS for estimating ground state energy of quantum chemistry molecules}
 $\bullet$ 
 \hyperref[app:converging_est_error]{Repeated simulations to show converging estimation error}
 $\bullet$ 
 \hyperref[app:optimal_h2]{DSS optimal on $H_2$}
\end{minipage}
\medskip

\noindent \textbf{\ref{app:HubbardModel}.~~\hyperref[app:HubbardModel]{HUBBARD MODEL NUMERICAL SIMULATIONS}} \dotfill\textbf{\pageref{app:HubbardModel}}
\medskip

\noindent \qquad \begin{minipage}{\dimexpr\textwidth-1.3cm}
\hyperref[app:ModifiedDSS_25mmts_perPauli]{DSS for estimating energy variance during quantum verification}
 $\bullet$ 
 \hyperref[app:ShallowShadowsBackground]{Shallow shadows}
 $\bullet$ 
 \hyperref[app:NaieveGroupingBackground]{Naieve grouping scheme}
\end{minipage}
\medskip

\newpage

\section{\label{app:DSSalgorithm} THE DSS ALGORITHM: DETERMINING THE MEASUREMENT CIRCUITS}

In this appendix we present our Derandomized Shallow Shadows (DSS) classical algorithm, determining the $N$ depth-$d$ measurement circuits we should apply on our quantum simulator.
The measurements more efficiently estimate the set of Pauli strings $\{P\}$ defined on a system of $n$ qubits compared to existing protocols. 
We outline the detailed steps of our DSS algorithm in Algorithm \hyperref[alg:DSS]{ 1} and provide insights into why this approach results in more effective measurements. 
One of the key contributions of our work is to prove that that the classical algorithm for finding good measurement circuits can be performed efficiently. A naive evaluation of the DSS cost function would require exponential time and memory, making it impractical for large system sizes. However, we overcome this by leveraging tensor network techniques to reduce the computational overhead to resources only polynomial in the system size. This section focuses on the DSS algorithm, and Appendix \ref{app:TensorNetwork_for_cost_function} is dedicated to explaining our efficient tensor network approach for evaluating the cost function.

\vspace{3mm}
\begin{center}
    \textit{How the DSS algorithm works}
\end{center}

Given a measurement budget of $N$ depth-$d$ measurements for estimating the expectation value of 
some Pauli strings $\{P\}$, the DSS algorithm will output  $N$ circuits of at most depth $d$ to use for this learning problem.
As such, our DSS algorithm starts with $N$ ensembles, $\{\mathcal{U}_i\}_{i=1}^N$, each representing a measurement shot.
Each ensemble is a distribution over a circuit ansatz that contains $d$ layers of two-qubit gates interleaved with single-qubit rotations -- see the measurement ansatz circuit (i) in Figure \ref{fig:Figure1}(a). And each gate in this ansatz is sampled from the Clifford group: each two-qubit (single-qubit) gate is randomly sampled from $\text{Cl}(2^2)$ ($\text{Cl}(2)$). 
See Definition \ref{def:initial_ansatz}.
\begin{definition}
 [\textbf{DSS Initial Ansatz}]\label{def:initial_ansatz}
Given depth $d$ and system size $n$, each measurement circuit ensemble $\mathcal{U}_i$ in the DSS algorithm is initialized with the following ansatz:
\begin{itemize}
    \item $\mathcal{U}_i$ contains $d$ layers of two-qubit gates interleaved with $d+1$ layers of single-qubit rotations (see circuit i in Figure \ref{fig:Figure1}a). 
\item The final two-qubit gate layer always couples the $1$st and $2$nd qubits, the $3$rd and $4$th, and so on. Note that if $n$ is odd, we only have $\lfloor n/2 \rfloor$ two-qubits gates per layer, and we define the final layer such that the $n$th qubit has no gate acting on it. The second to last layer is offset by one, coupling the $2$nd and $3$rd qubits, the $4$th and $5$th qubits, and so on. If $n$ even, we assume periodic boundary conditions and couple the $n$th and $1$st qubits. 
\item Each gate in this ansatz is sampled from the Clifford group: each two-qubit (single-qubit) gate is randomly sampled from $\textnormal{Cl}(2^2)$ ($\textnormal{Cl}(2)$). 
\end{itemize}
\end{definition}
Since each starting $\mathcal{U}_i$ has all one- and two-qubit gates sampled from $\text{Cl}(2)$ and $\text{Cl}(2^2)$ respectively, it is akin to being randomly sampled from an ensemble of depth-$d$ Clifford rotations.
In fact, this ensemble of rotations is equivalent to that of a shallow shadows ensemble \cite{bertoni2022shallow, ippoliti2023operator, hu2023classical} containing $d$ layers of two-qubit gates. 
These ensembles traditionally do not contain layers of single qubit rotations because sampling two-qubit gates from $\text{Cl}(2^2)$ already encapsulates the randomness generated by the single-qubit rotations. 
However, in our case we include the single-qubit rotations because our two-qubit gate options during derandomization do not generate the full two-qubit Clifford group (we will comment more on this below). 
Since the initial measurement circuit ansatz is equivalent to shallow shadows, the variance before derandomization is the variance of the shallow shadows protocol.
Crucially, the variance in shallow shadows is computed under the assumption that each gate in the circuit is sampled uniformly from the local Clifford ensemble.

While we start with the randomized ansatzes discussed above, our derandomization protocol will output a set of $N$ fixed circuits, which we defnine as $\{U^{DSS}_i\}_i$. In other words, the probability density function associated with the $i$th measurement ensemble $\mathcal{U}_i$ will become a delta distribution. Sampling from the final $\mathcal{U}_i$ will return the same unitary $U^{DSS}_i$ every time with probability 1. 
This is what the derandomization achieves. 
After starting with $N$ measurement ensembles initialized in the ansatz of Definition \ref{def:initial_ansatz}, the derandomization transitions each ansatz's gates from statistical mixtures over Clifford rotations (randomized) to fixed rotations (deterministic). 
Our algorithm goes gate-by-gate through each measurement ansatz, fixing each gate such that it is no longer randomly sampled from the Clifford group. 
The two-qubit gates become either \textsf{Identity}, \textsf{CNOT} or \textsf{SWAP}, and the single qubit rotations become one of the 6 independent rotations in $\text{Cl}(2)$ (see table \ref{tab:GateValues}). A reader well-versed in the Clifford group might point out that $|\text{Cl}(2)| = 24$. While this is true, notice that the Clifford group unitaries contain global phases in $\{1,-1,i,-i\}$, which do not matter when making a measurement. As such we only really have 6 independent rotations to consider, and these rotations define how our Paulis transform. 
We can fully define each measurement $i$ ansatz $\mathcal{U}_i = \mathcal{U}_i(\boldsymbol{t}^{(i)},\boldsymbol{s}^{(i)})$ by vectors $\boldsymbol{t}^{(i)}$ and $\boldsymbol{s}^{(i)}$, which specify the two- and single-qubit gates, respectively. 
For example, each entry of $\boldsymbol{t}^{(i)}$ corresponds to a different two-qubit rotation in the $i$th measurement ansatz, and this entry takes on one of the following numbers: 0 ($\sim \text{Cl}(2^2)$), 1(\textsf{Identity}), 2 (\textsf{CNOT}), or 3(\textsf{SWAP}). 
Notice that the directionality of the \textsf{CNOT} does not matter: a \textsf{CNOT} can be decomposed into a \textsf{CZ} with Hadamards on the target qubit \cite{nielsen_chuang_2010}. Therefore, when we set the single qubit rotation layers, we can always effectively swap the directionality of the \textsf{CNOT}.   
In fact we label all our rotations with numbers -- see table \ref{tab:GateValues}. 
We begin with each vector as the all-zero vector, $\boldsymbol{t}^{(i)}=\Vec{0}$ and $\boldsymbol{s}^{(i)}=\Vec{0}$, representing that each gate is sampled from the local Clifford group.  Then we derandomize, finding our final set of measurement ansatzes $\mathcal{U}_i$ by updating the entries of these vectors to be non-zero.
See Algorithm \hyperref[alg:DSS]{1} for a formal description of how we derandomize. Since each measurement ansatz $\mathcal{U}_i$ is fully specified by its vectors $\boldsymbol{t}^{(i)}$ and $\boldsymbol{s}^{(i)}$, this task is reduced to finding a good set of vectors $\{ \boldsymbol{t}^{(i)}, \boldsymbol{s}^{(i)}\}_i$ minimizing the cost function, \textsf{COST}, which is related to the sample complexity and will be defined in Definition \ref{def:DSScost_appendix}.

\begin{table}
[t]
  \begin{center}
    \renewcommand{\arraystretch}{1.5}
    \begin{tabular}{p{1.5cm}||p{1.5cm}p{1.5cm}p{1.5cm}p{1.5cm}p{1.5cm}p{1.5cm}p{1.5cm}} 
       & \multicolumn{1}{c}{0} & \multicolumn{1}{c}{1} & \multicolumn{1}{c}{2} & \multicolumn{1}{c}{3} & \multicolumn{1}{c}{4} & \multicolumn{1}{c}{5} & \multicolumn{1}{c}{6}
      \\
      \hline \hline
      \multicolumn{1}{r||}{\textit{Two-qubit gate options} }
      & 
      \multicolumn{1}{c}{$\hspace{3mm}$$\sim $Cl$(2^2)$$\hspace{3mm}$}
      & 
      \multicolumn{1}{c}{$\hspace{3mm}$\textsf{Identity}$\hspace{3mm}$}
      & 
      \multicolumn{1}{c}{$\hspace{5mm}$\textsf{CNOT}$\hspace{5mm}$}
      & 
      \multicolumn{1}{c}{$\hspace{4mm}$\textsf{SWAP}$\hspace{4mm}$}
      & 
      \multicolumn{1}{c}{\cellcolor{gray!15}$\hspace{19mm}$}
      & 
      \multicolumn{1}{c}{\cellcolor{gray!15}$\hspace{19mm}$}
      & 
      \multicolumn{1}{c}{\cellcolor{gray!15}$\hspace{19mm}$}
      \\
      \multicolumn{1}{r||}{\textit{Single-qubit gate options}}
      & 
      \multicolumn{1}{c}{$\sim $Cl$(2)$}
      & 
      \multicolumn{1}{c}{\textsf{Identity}}
      & 
     \multicolumn{1}{c}{$X\leftrightarrow Z$}
      & 
     \multicolumn{1}{c}{$Y\leftrightarrow X$}
      & 
      \multicolumn{1}{c}{$Z\leftrightarrow Y$}
      & 
      \multicolumn{1}{c}{$X\rightarrow Z\rightarrow Y$}
      & 
     \multicolumn{1}{c}{$X\rightarrow Y\rightarrow Z$}
      \\
    \end{tabular}
    
    \vspace{1mm}

    \caption{\textit{Gate assignments before and after derandomization.} 
    We can fully define each measurement $i$ ansatz $\mathcal{U}_i$ by vectors $\boldsymbol{t}$ and $\boldsymbol{s}$, which specify the two- and single-qubit gates in the measurement circuit. This table depicts the rotations these gates can be -- here, we label our rotations with numbers (columns). 
    Notice that the vectors $\boldsymbol{t}$ and $\boldsymbol{s}$ are always initialized as the all-zero vector $\Vec{0}$ by definition. 
    The two-qubit gates become either \textsf{Identity}, \textsf{CNOT} or \textsf{SWAP}, and the single qubit rotations become one of the 6 independent rotations in $\text{Cl}(2)$. We identify these based on how they permute the Pauli group. }
    \label{tab:GateValues}
  \end{center}
\end{table}


\noindent\rule{17.5cm}{0.8pt}

\noindent
\textbf{Algorithm 1} Derandomized Shallow Shadows (DSS) 
\vspace{-2mm}

\noindent\rule{17.5cm}{0.4pt}
    \begin{algorithmic}[1]\label{alg:DSS}
    \Require Measurements $N$, number of qubits $n$, Pauli strings to estimate $\{P\}$, max measurement circuit depth $d$
    \Ensure Measurement ansatzes $\{\mathcal{U}_i (\boldsymbol{t}^{(i)}, \boldsymbol{s}^{(i)})\}_{i=1}^N$, which have at most $d$ layers of two-qubit gates. The vectors $\boldsymbol{t}^{(i)}, \boldsymbol{s}^{(i)}$ efficiently represent the $i$th circuit by specifying its two-qubit and single-qubit gates. 
    \vspace{4mm}
    \State {\bf Setup Ansatzes:} $\mathcal{U}_i = \mathcal{U}_i\left(\boldsymbol{t}^{(i)}, \boldsymbol{s}^{(i)}\right)$ as the depth $d$ random gate ansatz (i) of Figure \ref{fig:Figure1}(a) and Definition~\ref{def:initial_ansatz}. The vector $\boldsymbol{t}^{(i)}$ specifies the $d\times \lfloor n/2 \rfloor$ two-qubit gates, and the second vector $\boldsymbol{s}^{(i)}$ specifies the $(d+1)\times n$ single-qubit gates. 
    Each entry takes on an integer (initially all are $0$) defined in table \ref{tab:GateValues}, which corresponds to a fixed gate.
    \State 
    \State $\forall i,$ $\boldsymbol{t}^{(i)} = \Vec{0}$ and $\boldsymbol{s}^{(i)} = \Vec{0}$
    \For{$j=1$ to $N$} \Comment{loop over measurements}
        \For{$g=1$ to $d\times \lfloor n/2 \rfloor$} \Comment{loop over two-qubit gates}
        \For{$V=$ $1$ (\textsf{Identity}), $2$ (\textsf{SWAP}), $3$ (\textsf{CNOT})} \Comment{two-qubit gate assignment options}
        \State $f(V) =  \text{\textsf{COST}}\left(\{\boldsymbol{t}^{(i)}, \boldsymbol{s}^{(i)}\}_i \text{ |  $t^{(j)}_g = V$} \right)$
        \EndFor
        \State $t^{(j)}_g\leftarrow \text{argmin}_V f(V) $
        \EndFor

        \vspace{1mm}
        \For{$g=1$ to $(d+1)\times n$} \Comment{loop over single-qubit gates}
        \For{$W=$ $1$, $2$, $3$, $4$, $5$, $6$} \Comment{single-qubit gate assignment options}
        \State $f(W) =  \text{\textsf{COST}}\left(\{\boldsymbol{t}^{(i)}, \boldsymbol{s}^{(i)}\}_i \text{ |  $s^{(j)}_g = W$} \right)$
        \EndFor
        \State $s^{(j)}_g\leftarrow \text{argmin}_W f(W) $
        \EndFor
    \EndFor
    \State \vspace{2mm}output $\{\boldsymbol{t}^{(i)}, \boldsymbol{s}^{(i)}\}_i$
    \end{algorithmic}
\vspace{-2mm}
\noindent\rule{17.5cm}{0.4pt}

\vspace{3mm}

Once the vectors $\{ \boldsymbol{t}^{(i)}, \boldsymbol{s}^{(i)}\}_i$ are fully nonzero at the end of the algorithm, we obtain a fixed set of circuits $\{U^\textnormal{DSS}_i\}_i$ that can be directly implemented to measure the target Pauli strings. We find $U^\textnormal{DSS}_i$ from
\begin{equation}
    \{U^\textnormal{DSS}_i\}_i = \{\mathcal{U}(\boldsymbol{t}_\text{final}^{(i)}, \boldsymbol{s}_\text{final}^{(i)})\}_i,
\end{equation}
where each $(\boldsymbol{t}_\text{final}^{(i)}, \boldsymbol{s}_\text{final}^{(i)})$ is the final vector that specifies a deterministic circuit, i.e. $U_i^{DSS}$.
The DSS strategy is a greedy algorithm and, as such, is not guaranteed to give a global minimum. 
However, as is demonstrated by multiple examples in the main text, the strategy still outperforms previous state-of-the-art bounded depth techniques and sometimes finds the optimal set of bounded-depth measurements. We can consider a simple example. Imagine wanting to learn the $n=8$ qubit Pauli strings $IIIXIIIX$,  $IIIYIIIY$, and  $IIIZIIIZ$ and allowing circuits up to depth $d=3$. 
Since the Bell basis simultaneously diagonalizes $XX$, $YY$, and $ZZ$, the optimal measurement strategy is clearly to rotate into the Bell basis of the 1st and 5th sites. We put these Paulis into our algorithm and allow $N=100$ measurments.
The DSS algorithm finds this optimal solution, outputting $100$ copies of the Bell basis measurement circuit. See Supplementary Figure \ref{supfig:Bell_example}.
We believe the DSS algorithm does so well because derandomizing the two-qubit gates first allows our measurement circuits to explore a variety of geometric structures \cite{jonforthcoming}.
This ability to explore comes from the possible two-qubit gate choices during derandomization. 
In particular, the two-qubit \textsf{SWAP} gate allows the algorithm to modulate the geometric structure of each measurement circuit: after a few layers of \textsf{SWAP}s, two-qubit \textsf{CNOT} gates can couple non-local pairs of qubits. 
Moreover, in addition to \textsf{SWAP}, we only include \textsf{Identity} and \textsf{CNOT} as options because, when interleaved by single qubit Cliffords, these gates can generate the Clifford group  \cite{nielsen_chuang_2010}.

\begin{figure*}[h]
\centering
\includegraphics[scale = 0.43]{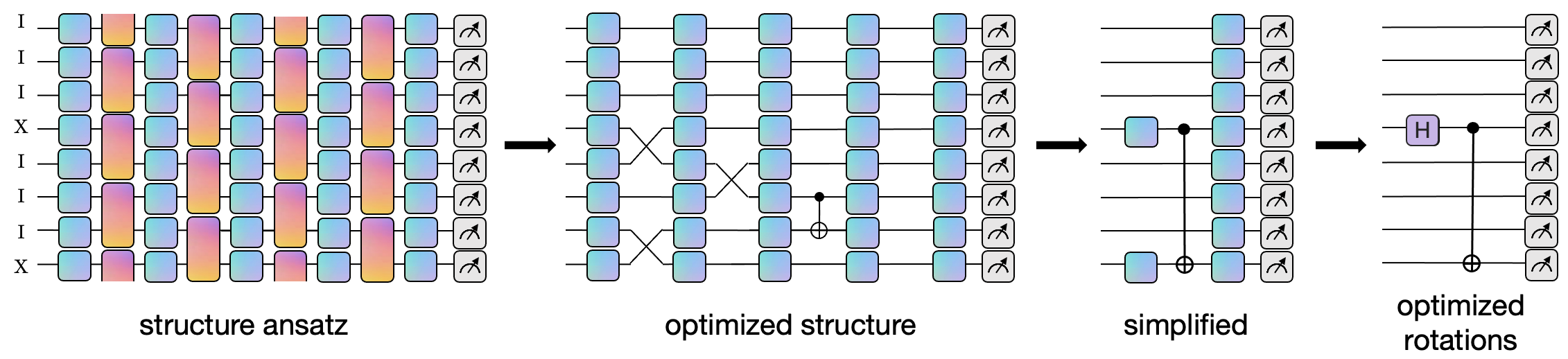}
\caption{
\emph{Bell basis example.}  
Imagine wanting to learn the $n=8$ qubit Pauli strings $IIIXIIIX$,  $IIIYIIIY$, and  $IIIZIIIZ$ and allowing $N=100$ measurement circuits up to depth $d=4$. The DSS algorithm outputs $100$ copies of the measurement circuit, which rotates into the Bell basis on the fourth and final qubits. The bell basis simultaneously diagonalizes $XX$, $YY$, and $ZZ$, and therefore DSS finds the optimal outcome. This figure depicts what a single measurement circuit looks like as it is being derandomized. First the two qubit gates are fixed and then the single qubit gates are fixed. It is impressive that even though we allow depth $d=3$, DSS can find the optimal, lower-depth solution.
}\label{supfig:Bell_example}
\end{figure*}

Finally, notice in Algorithm \hyperref[alg:DSS]{1} that the order in which we assign gates (i.e. in our $\boldsymbol{t}^{(i)}$ and $\boldsymbol{s}^{(i)}$ vectors) determines which gates are fixed first -- and thus can affect the outcome measurements of our DSS protocol. In general, we start with derandomizing the final layer of two qubit gates (final = directly before measurement) and end with derandomizing the first layer. And for the single qubit gates, we generally go the opposite direction: derandomizing the first layer of gates, then the second layer and so on, ending with the last layer of gates before measurement. All of this being said, we want to stress that the best ordering will depend on the underlying learning problem, and we recommend attempting different orderings.

\vspace{3mm}
\begin{center}
    \textit{The DSS \textsf{COST} function}
\end{center}

When our DSS protocol derandomizes each random gate, it chooses the most effective gate -- from the options listed in Table \ref{tab:GateValues} above -- for our learning problem.
In other words, each gate is chosen to minimize the confidence with which we estimate our Pauli strings of interest.
With a fixed measurement budget $N$, the confidence is directly related to the variance, and in Appendix \ref{app:Performance_guarantees} we will show that bounding confidence leads to precision guarantees on the desired observables. Therefore, we adopt the confidence as our cost function, as was first pioneered in Ref \cite{huang2021efficient}. 
The confidence with which we learn the Pauli string $P$ with precision $\epsilon$, using our learning protocol of $N$ measurements $\{\mathcal{U}_i\}_{i=1}^N$, is as follows:
\begin{equation}
\text{\textsf{conf}}_P\left(\{\mathcal{U}_i\}_{i=1}^N\right) = 2 \prod_{i=1}^{N} \exp{\left[ \hspace{1mm} -\frac{\epsilon^2}{2} p_i(P) \hspace{1mm} \right]},
\end{equation}
    where $p_i(P)$ is the probability that the $i$th measurement circuit diagonalizes $P$. We can express this in terms of our measurement ensembles $\{\mathcal{U}_i\}_{i=1}^N$. See the following definition. 
\begin{definition}[\textbf{Pauli weight for measurement ensembles $\mathcal{U}_i$}]\label{def:PauliWeightUm}
Given the $i$th measurement's ensemble $\mathcal{U}_i$, which is some statistical mixture over Clifford rotations, the Pauli weight of the $n$-qubit Pauli string $P$ is
\begin{equation}
     p_i(P) = \frac{1}{2^n} \Expect_{U\sim \mathcal{U}_i} \sum_{b}  \bra{b} U P U^\dagger \ket{b}^2.
\end{equation}    
The sum over $b$ is a sum over all computational basis states.
\end{definition}
It turns out that this probability is, formally, the Pauli weight \cite{Bu24} (the inverse of the shadow norm \cite{huang2020predicting, jonforthcoming, van2022hardware}), when using the statistical ensemble defined by the $i$th measurement $\mathcal{U}_i$ to perform classical shadow tomography. Again, see appendix subsection Appendix \ref{app:Performance_guarantees} for background on this connection.  
When a measurement circuit has been fully derandomized, it simultaneously diagonalizes some set of $2^n$ Paulis. In other words, it learns those Paulis and no others. Therefore, the probability $p_i(P)$ that the derandomized measurement circuit $U_i^{DSS}$ learns one of our Paulis $P$ of interest must be either 0 or 1. As we will see below, learning a Pauli string $P$ ($p_i(P) = 1$) will give a narrower confidence interval.

We define the cost function of the DSS protocol as the convex combination of the confidences of the desired Pauli strings $\{P\}_P$. 
We use weights $w_P$ to indicate relative importance among our Paulis ($w_P = 1$ for all $P$ in Figure \ref{fig:Figure1}).
With measurement circuits indexed $i = 1$ to $N$, our DSS cost function is defined below. 
\begin{definition} [\textbf{DSS Cost function}]\label{def:DSScost_appendix}
Given $N$ measurement ensembles $\{\mathcal{U}_i\}_{i=1}^N$, where $\mathcal{U}_i = \mathcal{U}_i(\boldsymbol{t}^{(i)}, \boldsymbol{s}^{(i)})$ and  Pauli strings $\{P\}$ we want to learn.  The DSS cost function takes the form 
\begin{eqnarray}
    \text{\textsf{COST}}\left(\{\mathcal{U}_i\}_{i=1}^N\right) 
    &=& \sum_P w_P \hspace{1mm} \text{\textsf{conf}}_P\left(\{\mathcal{U}_i\}_{i=1}^N\right) \\ 
    &=& \sum_P w_P \times 2 \prod_{i=1}^{N} \exp{\left[ \hspace{1mm} -\frac{\epsilon^2}{2} p_i(P) \hspace{1mm} \right]}. \label{eqn:DSScost_appendix} 
\end{eqnarray}
where $w_P$ are weights chosen for each Pauli string $P$, $\epsilon$ is a hyperparameter, and $p_i(P)$ is the Pauli weight of $P$ under the ensemble $\mathcal{U}_i$.
\end{definition}
Here, we consider the precision $\epsilon$ as a hyperparameter to tune.  We can build intuition for this hyperparameter using concentration inequalities -- for example, under a fixed number of measurements $N$, choosing a very small precision ($\epsilon$) results in poor confidence. Whereas, a large precision ($\epsilon$) will allow for high confidence. In other words, this $\epsilon$ hyperparameter controls the sensitivity of our cost function landscape to different gate choices. 
While here our \textsf{COST} function's input is the set of measurement ensembles $\{\mathcal{U}_i\}_i$,
when implementing DSS, we always use the measurement ensembles' vector representation $\{\boldsymbol{t}^{(i)},\boldsymbol{s}^{(i)}\}_i$ because it is considerably more efficient.

In order to evaluate \textsf{COST} for some set of measurement circuits, we need the probabilities $\{p_i(P)\}$ that the measurement circuits diagonalize each Pauli string $P$ of interest. 
Once we have these probabilities, we plug them into Equation~\eqref{eqn:DSScost_appendix} and evaluate. 
However, solving for these probabilities is nontrivial. In order to evaluate $p_i(P)$, we must track how the unitaries in our ensemble $\mathcal{U}_i$ transform $P$, and using direct simulation this requires space exponential in the system size.
In particular, for each $U \sim \mathcal{U}_i$ you would apply $U$: $P \rightarrow U P U^\dagger$ and see how often $U \rightarrow U P U^\dagger$ is diagonal; this evaluation is $\mathcal{O}(\text{exp}(n))$ because we are multiplying matrices of size $2^n \times 2^n$. Moreover, we would need to perform this expensive subroutine for every physical circuit in $\mathcal{U}_i$'s ensemble of possible circuits.

A key contribution of this work is that we are able to avoid this exponential time (and space) complexity. 
Using tensor network techniques, we bypass this exponential overhead and utilize resources only polynomial in the system size when $d$ is polylogarithmic. Our technique mimics random Markovian classical processes, and appendix \ref{app:TensorNetwork_for_cost_function} pedagogically describes how it efficiently calculates the cost function during derandomization.

\vspace{3mm}
\begin{center}
    \textit{Example from the main text: $30$ randomly-chosen Pauli strings}
\end{center}

To round out this discussion of DSS, let's return to the example in the main text. 
In Figure \ref{fig:Figure1}, we considered the problem of estimating $30$ Pauli strings on $8$-qubits, using $N=100$ measurements with depth $d=3$.  We randomly generate these $30$ Pauli strings 
\footnote{These Pauli strings are as follows: \newline $IXIXYXYY$, $XIYIIIZZ$, $XYXZXIXX$, $YZYXIIYZ$, $IYZZIZII$, $ZZYXYXZI$, $ZXIZIYXI$, $IIIIZZYY$, $XIZZXZXY$, $IZYZXYII$, $XZZYXXXX$, $ZIIIZIII$, $IZYXXZYY$, $XIXIXYYY$, $XIYYXZYX$, $IZIXXZIY$, $YYIXYYZI$, $IXIIXIII$, $IZYXYYIZ$, $XYXXIIIY$, $XYZIXZIZ$, $IYYZZXYX$, $XZYYXYIY$, $IZXXXIII$, $IYZYYIYZ$, $IXXXXIZZ$, $ZZYYZYIY$, $YYZXXXYY$, $XIIIXIII$, and $XIYXYYXZ$.}.
Figure \ref{fig:Figure1}(a) demonstrates snapshots of this procedure on the $17$th measurement ensemble, and we perform this derandomization procedure one-at-a-time on each of our $100$ measurement ensembles. Notice the final circuit (circuit iv) no longer contains any random gates.  
As such, for each observable $P$ in our set $\{P\}$, this (now deterministic) measurement circuit either diagonalzes the Pauli string $P$ or it does not because the circuit is Clifford. If it diagonalizes $P$, then measurement in the $Z$ basis will yield information about the expectation value $\expval{P}$. 
In the next appendix (Appendix \ref{app:estimating_exp_vals_with_DSS}), we discuss how to turn the quantum computer's measurement outputs (i.e. the measurements from the derandomized circuits) into estimates of these expectation values.

This example also highlights how DSS's algorithmic structure finds measurement circuits that are effective for the learning task. 
At the beginning each gate is randomly sampled from the Clifford group (Figure \ref{fig:Figure1}(a), circuit i), so the procedure starts in a mixture of depth-$d$ Clifford circuits. 
Each measurement circuit being in a statistical mixture over Clifford rotations allows the algorithm to ``see’’ the entire landscape of possible measurements. Then, as we derandomize each ansatz, we restrict our classical mixture to smaller and smaller regions of the landscape. This continues until each gate in each measurement has been fixed, and we have converged to one point in the landscape.  
Figure \ref{fig:Figure1}(c) supports this intuition. 
Since each ansatz begins as a statistical mixture over Clifford rotations, we can learn any Pauli string with low but nonzero probability (lightest blue line). 
Then, as we derandomize each measurement, we are probabilistically more likely to learn our $30$ Paulis of interest. 
At the end once all gates in all measurements are fixed, each measurement is fully deterministic, and we often learn our chosen Pauli strings (darkest blue line).

\newpage
\section{ESTIMATING EXPECTATION VALUES WITH THE DSS MEASUREMENT CIRCUITS}\label{app:estimating_exp_vals_with_DSS}

Now that we have described how the DSS algorithm determines which $N$ depth-$d$ measurements to make, we will discuss how to construct estimates of our Pauli strings $\{P\}$ of interest.  
Assume that we have already run our DSS algorithm and implemented the $N$ measurements $\{U_i^{DSS}\}_{i=1}^N$ it told us to make. These $U_i^{DSS}$ are the unitaries sampled with probability 1 (because the ensemble becomes fixed) from the ensemble $\mathcal{U}_i(\boldsymbol{t}^{(i)}_{\text{final}}, \boldsymbol{s}^{(i)}_{\text{final}})$, where $\boldsymbol{t}^{(i)}_{\text{final}}, \boldsymbol{s}^{(i)}_{\text{final}}$ are vectors output by the DSS algorithm. 
We can use the bit strings $\{b_i\}_{i=1}^N$ such that $b_i \in \{0,1\}^n$ output by the quantum computer to estimate our Paulis. 
Notice that since the output of the DSS algorithm is some set of deterministic Clifford circuits, every output circuit $U_i^{DSS}$ either diagonalizes or does not diagonalize $P$. 
In the derandomization work of Ref \cite{huang2021efficient}, the authors defined a ``hitting count'' $h(P)$ for each Pauli $P$. In analogy with their work, we define the hitting count (below), which - after the DSS algorithm has finished - is equivalent to the number of these $N$ circuits that diagonalize Pauli string $P$. 
\begin{definition} [\textbf{Hitting Count}]\label{def:hittingcount} 
Given the $N$ measurement ensembles $\mathcal{U}_i$, 
the hitting count $h(P)$ for a Pauli string $P$ is
    \begin{equation}
        h(P) = \sum_i p_i(P).
    \end{equation}
\end{definition}
Again, if the DSS algorithm has already run and $\{U_i^{DSS}\}_{i=1}^N$ are the $N$ unitaries sampled with probability 1 (because the ensemble becomes fixed) from the ensemble $\mathcal{U}_i(\boldsymbol{t}^{(i)}_{\text{final}}, \boldsymbol{s}^{(i)}_{\text{final}})$. Then the hitting count is the number of times we measure $P$. Therefore, we can use it to normalize our estimator. In particular, we estimate $\expval{P} \approx \hat{o}(P)$, where the estimator
\begin{equation}\label{eqn:empirical_average_estimator}
    \hat{o}(P) = \frac{1}{h(P)}  \sum_{i} \bra{b_{i}} U^{DSS}_{i} P U_{i}^{DSS\dagger} \ket{b_{i}}
\end{equation}
for our learning protocol is simply this empirical average. Each $\bra{b_{i}} U_{i} P U_{i}^\dagger \ket{b_{i}}$ can be computed efficiently because our circuits are Clifford rotations, and our observable is a Pauli. Moreover, notice that even though the DSS algorithm -- which identifies the Pauli strings -- is born from randomized measurements, this technique does not require any special post-processing. 

Since each measurement $i$ is deterministic, we no longer sample from a nontrivial distribution of Clifford circuits containing more than a single circuit. Instead we just sample $U_i^{DSS}$ with probability 1, and so $p_i(P) \in \{0,1\}$ serves as an indicator function for whether $U_i^{DSS}$ diagonalizes $P$.
This is the estimator expression we will use below as we derive performance guarantees. 
Notice, however, that if $h(P)=0$ for some $P$, then none of our $N$ measurements ever diagonalize and learn $\expval{P}$. As such, we just set $\hat{o}(P) = 0$ since we have no information on this expectation value.

The above procedure allows us to jointly estimate $|\{P\}|$ Pauli observables using $N$ measurements of depth at most $d$. The quality of our empirical average reconstruction is exponentially suppressed in $h(P)$, the number of times we hit each Pauli string $P$, and below (Appendix~\ref{app:Performance_guarantees}) we will derive the associated guarantees. 
Finally, in the procedure above, we assumed all of the post-processing could be done efficiently. For example, we commented that we already know which circuits diagonalize each Pauli $P$. 
This comes from the DSS algorithm: recall that we compute (lines 7, 11 in Algorithm \hyperref[alg:DSS]{1}) all $p_i(P)$ to calculate the cost function. At the end of the protocol, when each $U_i^{DSS}$ is determined, these probabilities become indicator functions for whether $U_i^{DSS}$ diagonalizes $P$. Therefore, we can just save these mid-algorithm calculations to use when constructing our estimators. That being said, this is also not entirely necessary -- since each $U_i^{DSS}$ is a Clifford circuit, it is also efficient to recalculate this in post-processing. Again, we defer the interested reader to our tensor network section in Appendix \ref{app:TensorNetwork_for_cost_function}.

\newpage
\section{DSS PERFORMANCE GUARANTEES} \label{app:Performance_guarantees}

In the previous sections, we discussed our classical DSS algorithm.
As discussed, the DSS algorithm's derandomization optimizes the measurement circuits to minimize the variance of specific Pauli strings.
We then outlined how to implement these circuits and use the outcomes to estimate expectation values for the Pauli strings of interest. In this appendix, we provide formal performance guarantees on these estimates. 
The main result (Theorem \ref{thm:theorem_performance_guarantee}) shows how our final cost function value, at the end of our DSS algorithm when all gates have been derandomized, can be used to bound the precision of our estimates. 
In particular, in order to use this theorem to obtain a numerical upper bound on the precision with which we estimate every $\expval{P} \in \{P\}$,  one will need to solve a transcendental equation.
We bring the reader to this point and then cite a variety of numerical methods one can use to perform this evaluation.  
Moreover, we will also show how to (formally) tighten this upper bound on precision for each observable $\expval{P}$, by using additional information from the output of the DSS algorithm. 
This is perhaps both a simpler and more practical strategy than solving the aforementioned transcendental equation.
By considering each Pauli observable individually, we can use its variance to solve for a better bound -- and note that our bounds are state-agnostic because we maximize our variance over states. We will also discuss our derandomization \textsf{COST} function's relationship with variance.

Finally, we end with some technical remarks. 
In particular, we show that a relaxed version of our DSS protocol is guaranteed to perform at least as well as shallow shadows with the same depth circuits. 
We also numerically observe 
that depth $d$ DSS consistently performs at least as well lower depth ($< d$) versions. 
The freedom to set deep gates to the identity is a key feature of our algorithm. 
Recovering observables of small weight becomes exponentially costly in depth with fully randomized shallow shadows because the desired information gets scrambled across a larger system from which it has to be recovered probabilistically. 
The fact that our protocol shows improvement with depth indicates that we use the freedom of increased depth only when it actually facilitates measuring observables in parallel.

\vspace{5mm}

\begin{center}
    \textit{Deriving our DSS }\textsf{COST}\textit{ function and its guarantees}
\end{center}

Here we show
that the final value of our cost function bounds the precision across all our Pauli string estimates (Theorem \ref{thm:theorem_performance_guarantee}).
The results we derive here have similar spirit as the results of Ref~\cite{huang2021efficient}; however, we prove everything for measurement circuits with nontrivial - albeit shallow - entanglement. 
Moreover, we also take the results a step further by (1) showing how to obtain tigher precision-guarantees on our estimators via the final hitting counts (Definition~\ref{def:hittingcount}) of our DSS measurement circuits and (2) specifying the confidence cost function's relationship with our protocol's variance. 
As the variance depends on the underlying state, we always consider the maximum variance over states such that this quantity is agnostic of the underlying state prepared on our quantum simulator.
We suggest that a reader, who is simply interested in applications of this scheme, skip over Lemma \ref{lem:lemma_concentration_ineq} and Theorem \ref{thm:theorem_performance_guarantee} below;  start reading again (after Theorem \ref{thm:theorem_performance_guarantee}) where we discuss how to further tighten our estimation error guarantees.

In order to prove Theorem \ref{thm:theorem_performance_guarantee}, we must first choose a concentration inequality for our scheme, which bounds the precision and confidence for our learning task in terms of the measurements $\{U^{DSS}_i\}_i$ we make in our protocol. 
We choose Hoeffding's inequality.
It turns out that, to apply this inequality, the only information we need about DSS's depth-$d$ measurements will be their hitting counts (see Definition~\ref{def:hittingcount}). 

\begin{lemma} [\textbf{Confidence Bound}]\label{lem:lemma_concentration_ineq}
     Suppose we want to estimate Pauli observables $\{P\}_P$ using $N$ depth-$d$ measurement circuits $\{U_i^{DSS}\}_{i=1}^N$. Fix some desired precision $\epsilon \in (0,1)$.
     Then the associated empirical averages $\hat{o}(P)$ (see \equref{eqn:empirical_average_estimator}) constructed from the measurement outcomes all obey
     \begin{equation}
          |\expval{P} - \hat{o}(P)| \leq \epsilon \hspace{5mm} \forall P \in \{P\}
     \end{equation}
     with probability (at least) $1-\delta \in (0,1)$ where $\delta$ takes the form
     \begin{equation}
        2 \sum_P
    \textnormal{exp}\left( \frac{- \epsilon^2}{2} h(P) \right) \leq \delta.
     \end{equation}
\end{lemma}
\begin{proof}
Consider each observable $P \in \{P\}_P$ of interest. Each $P$ is diagonalized by our measurement circuits $\{U^{DSS}_i\}_i$ a number of times equal to $h(P)$, and therefore, we have measured $\expval{P}$ with $h(P)$ shots. We can use a concentration inequality to bound the precision of our empirical average estimator $\hat{o}(P)$ defined in \equref{eqn:empirical_average_estimator} ( see Appendix~\ref{app:estimating_exp_vals_with_DSS}). 
Hoeffding's inequality assumes this estimator is an empirical average and is expressed in terms of the number of shots contributing to this estimator (here: $h(P)$). 
It assumes each shot contributing to the sum is an independent random variable in $\{-1,1\}$ -- just as in \equref{eqn:empirical_average_estimator}.  
As such, Hoeffding's inequality upper bounds the probability that our estimator is within some precision $\epsilon$ and takes the form
\begin{equation}
\label{eqn:hoeffding_each_P}
    \text{Pr}\big[\hspace{1mm} |\expval{P} - \hat{o}(P)| \geq \epsilon \hspace{1mm} \big] 
    \leq 
    2 \hspace{1mm} \text{exp}\left( \frac{- \epsilon^2}{2} h(P) \right).
\end{equation}
This inequality can be defined for every Pauli observable $P$ we want to learn. And therefore, we end up with $|\{P\}_P|$ such inequalities.  As was done in Ref~\cite{huang2021efficient}, we can apply Boole’s inequality (the union bound): remember that for events $A$ and $B$ and $C$ we have $\text{Pr}(A \cup B \cup C) \leq \text{Pr}(A) + \text{Pr}(B) + \text{Pr}(C)$.
Boole’s inequality also extends to many events, and therefore, we can use it to consider the probability of any our Pauli estimators giving precision larger than $\epsilon$. Notice that this is the antithesis of what we want: we ideally want \textit{all} our estimates to be within precision $\epsilon$. We apply Boole’s inequality below.
The right hand side will be a sum over the individual probabilities, each of which we can upper bound using its \equref{eqn:hoeffding_each_P} Hoeffding inequality. See the first line $\rightarrow$ second line below. 
\begin{eqnarray}
    \text{Pr}\left[ \hspace{1mm}\bigcup_P \hspace{1mm} \big( |\expval{P} - \hat{o}(P)| \geq \epsilon \hspace{1mm} \big) \right]
    &\leq&
    \sum_P
    \text{Pr}\big[\hspace{1mm} |\expval{P} - \hat{o}(P)| \geq \epsilon \hspace{1mm} \big], \\
    &\leq&
    2 \sum_P
    \text{exp}\left( \frac{- \epsilon^2}{2} h(P) \right).
\end{eqnarray}
We have upperbounded the probability that \textit{any} our estimates have precision greater than $\epsilon$. Note we can also write the left hand side in terms of the Pauli $P$ with the most imprecise estimate -- the one with the largest $|\expval{P} - \hat{o}(P)| $, as follows:
\begin{equation}\label{eqn:final_eqn_concentration_lemma}
    \text{Pr}\left[\hspace{1mm}\max_P |\expval{P} - \hat{o}(P)| \geq \epsilon \hspace{1mm} \right] 
    \leq
    2 \sum_P
    \text{exp}\left( \frac{- \epsilon^2}{2} h(P) \right).
\end{equation}
At this point we are done. We define our protocol's confidence as the probability that all our estimates have precision $|\expval{P} - \hat{o}(P)| \leq \epsilon$. Therefore, if we upper bound the right hand side of \ref{eqn:final_eqn_concentration_lemma} with $\delta$, then our confidence must be at least $1-\delta$. 
\end{proof}

We can use this lemma to derive our confidence cost function's associated performance guarantees, given that our cost function achieves a certain value. Recall that we defined the upper bound (right side) of \equref{eqn:hoeffding_each_P}: $\text{\textsf{conf}}_P = 2 \hspace{1mm} \text{exp}\left( \frac{- \epsilon^2}{2} h(P) \right)$. As the reader will see, we chose our cost function to be an upper bound on the confidence across all observables. Due to this setup, whatever final value our cost function takes will inform us about the quality of our estimates. 
\setcounter{theorem}{0}
\begin{theorem} [\textbf{Performance guarantee from \textsf{COST} function}]\label{thm:theorem_performance_guarantee}
Given Pauli strings $\{P\}$ we want to learn, fix some desired precision $\epsilon \in (0,1)$.  With the DSS Algorithm \hyperref[alg:DSS]{1}'s hyperparameter set to this precision $\epsilon$, the DSS Algorithm outputs $N$ measurement circuits $\{U_i^{DSS}\}_{i=1}^N$. The final unweighted ($w_P =1$) cost function $\textnormal{\textsf{COST}}\left(\{U_i^{DSS}\}_{i=1}^N\right)$ achieved with these circuits gives the following guarantee: 
the empirical averages $\hat{o}(P)$ (see \equref{eqn:empirical_average_estimator}) constructed from the associated measurement outcomes all obey
     \begin{equation}
          |\expval{P} - \hat{o}(P)| \leq \epsilon \hspace{5mm} \forall P \in \{P\}
     \end{equation}
     with probability (at least) $1-\textnormal{\textsf{COST}}\left(\{U_i^{DSS}\}_{i=1}^N\right)$.

\end{theorem}
\begin{proof}
The \textit{unweighted} DSS cost function is as follows,
where $p_i(P)$ is the Pauli weight of $P$ under the final measurement ensemble $\mathcal{U}_i$. 
\begin{eqnarray}
    \text{\textsf{COST}}\left(\{U_i^{DSS}\}_{i=1}^N\right) 
    &=& 2 \sum_P \prod_{i=1}^{N} \exp{\left[ \hspace{1mm} -\frac{\epsilon^2}{2} p_i(P) \hspace{1mm} \right]} \\
    &=& 2 \sum_P \exp{\left[ \hspace{1mm} -\frac{\epsilon^2}{2} \sum_{i=1}^{N} p_i(P) \hspace{1mm} \right]} 
\end{eqnarray} 
At the end of the DSS algorithm, each ensemble $\mathcal{U}_i$ has been derandomized into a final circuit $U_i^{DSS}$, which therefore has all fixed gates. As such, it either does or does not diagonalize (and thus measure) the Pauli string $P$ -- recall $p_i(P) \in \{0,1\}$ is an indicator function since $U_i^{DSS}$ is a Clifford rotation. If $U_i^{DSS}$ diagonalizes $P$, then $p_i(P) = 1$, and if not then $p_i(P) = 0$. 
Therefore, the sum $\sum_i p_i(P)$ simply counts how many times we measure $P$ and, by definition \ref{def:hittingcount}, is equal to $P$'s hitting count $h(P)$. We have
\begin{equation} \label{eqn:final_thm_equation}
    \text{\textsf{COST}}\left(\{U_i^{DSS}\}_{i=1}^N\right) 
    = 2 \sum_P \exp{\left[ \hspace{1mm} -\frac{\epsilon^2}{2} h(P) \hspace{1mm} \right]},
\end{equation}
which is equivalent to the confidence upper bound derived in Lemma \ref{lem:lemma_concentration_ineq} (see \equref{eqn:final_eqn_concentration_lemma}). Therefore, after we run our DSS algorithm and obtain the $N$ measurement circuits $\{U_i^{DSS}\}_{i=1}^N$ that achieve some cost function value $\text{\textsf{COST}}\left(\{U_i^{DSS}\}_{i=1}^N\right)$, this final cost function value upper bounds the probability that any one of our estimates is beyond precision $\epsilon$. The claim follows.
\end{proof}

Therefore, when our DSS algorithm's derandomization minimizes the cost function, we are really increasing the probability with which we achieve the guarantee in the Theorem above. 
This result is crucial: we directly relate minimizing on the cost function landscape to improving the probability that we achieve some precision for our estimates. 
Or equivalently, with some fixed probability, we can achieve an upper bound on precision across \textit{all} Pauli strings we care about.
From this setup and given some desired confidence $\delta$, one could also solve for this precision $\epsilon$ across all estimates. 
However, as it stands the final \equref{eqn:final_thm_equation} is transcendental and thus does not have an analytic solution for arbitrary values of $h(P)$ and hence $\delta$. That being said, in practice one can solve it -- using numerical methods such as Newton's method or a numerical solver to approximate $\epsilon$.

Finally, in some cases one might want to learn a complex observable by decomposing it into Pauli strings, estimating the Pauli expectation values, and then recombining the results. In fact, in one of our applications, we do exactly this. We estimate the energy of various quantum chemistry Hamiltonians. If we have a Hamiltonian $H = \sum_P c_P P$, then the corollary below immediately follows via the triangle inequality to give a guarantee on the accuracy of our energy estimate. Notice that here we utilize an unweighted \textsf{COST} function.

\begin{corollary} [\textbf{Estimating complex observables}]
Given some observable $H$, we can decompose it into Pauli strings, $H = \sum_P c_P P$, and therefore, we want to learn Pauli strings $\{P \hspace{1mm} |\hspace{1mm} c_P \neq 0\}$.  Fix some desired precision $\epsilon \in (0,1)$.  With the DSS Algorithm \hyperref[alg:DSS]{1}'s hyperparameter set to this precision $\epsilon$, the DSS Algorithm outputs $N$ derandomized measurement ensembles $\{\mathcal{U}_i\}_{i=1}^N$. The final unweighted ($w_P =1$) cost function $\textnormal{\textsf{COST}}\left(\{\mathcal{U}_i\}_{i=1}^N\right)$ achieved with these ensembles gives the following guarantee: 
if we estimate $\hat{o}(H) = \sum_P c_P \hat{o}(P)$, where
$\hat{o}(P)$ is constructed via \equref{eqn:empirical_average_estimator}, then
    \begin{equation}
        | \expval{H} - \hat{o}(H) | 
        \leq  \epsilon \sum_P  | c_P| 
    \end{equation}
    with probability (at least) $1-\textnormal{\textsf{COST}}\left(\{\mathcal{U}_i\}_{i=1}^N\right)$.
\end{corollary}

\vspace{5mm}
\noindent
\begin{center}
    \textit{Bounding the precision of each Pauli string under the DSS protocol}\label{app:bounding_precision_DSS}
\end{center}

We can extend and tighten this result by considering each Pauli observable individually. The only information we will need from the DSS protocol is the hitting count $h(P)$ for each of our $P$ of interest.  This information must be kept track of in our DSS algorithm implementation anyway and can be returned with the measurement circuits (or solved for again post derandomization).
For a single Pauli string $P$ expectation value estimate, $\textnormal{\textsf{conf}}_P\left(\{U^{DSS}_i\}_{i=1}^N\right)$ upper bounds the probability that our protocol \textit{does not} estimate $\expval{P}$ within precision $\epsilon$. 
The probability, therefore, that the protocol learns $\expval{P}$ to within precision $\epsilon$ is $1-\textnormal{\textsf{conf}}_P$.
This quantity $\textnormal{\textsf{conf}}_P$ is defined by the application of Hoeffding's inequality to each individual $P$ in \equref{eqn:hoeffding_each_P} above. Formally, we state this relationship below. 
\begin{corollary}[\textbf{Confidence of learning Pauli $P$}]\label{cor:confidence_of_P}
The confidence with which we learn the Pauli string $P$ with precision $\epsilon$, using our learning protocol of $N$ final measurements $\{U^{DSS}_i\}_{i=1}^N$, is as follows:
\begin{eqnarray}
\textnormal{Pr}\big[\hspace{1mm} |\expval{P} - \hat{o}(P)| \geq \epsilon \hspace{1mm} \big] 
&\leq& 2 \exp{\left[ \hspace{1mm} -\frac{\epsilon^2}{2} h(P) \hspace{1mm} \right]} \\
&=& \textnormal{\textsf{conf}}_P\left(\{U^{DSS}_i\}_{i=1}^N\right) ,
\label{eqn:conf_corrollary_def}
\end{eqnarray}
where $h(P)$ is the hitting count defined by the measurements $\{U^{DSS}_i\}_{i=1}^N$.
\end{corollary}
\begin{proof}
    This immediately follows from Lemma \ref{lem:lemma_concentration_ineq} and is equivalent to the definition of confidence we presented earlier in section \ref{app:DSSalgorithm}, where we substituted $h(P) = \sum_i p_i(P)$.
\end{proof}
We can use this to individually bound the precision of each observable. Consider a given Pauli string $P$ with hitting count $h(P)$, defined by measurements $\{U^{DSS}_i\}_i$, and set $\delta \in (0,1)$ to be
\begin{eqnarray}
    1- \textnormal{\textsf{conf}}_P\left(\{U^{DSS}_i\}_{i=1}^N\right) 
    &\geq& 1 - 
    2 \exp{\left[ \hspace{1mm} -\frac{\epsilon^2}{2} h(P) \hspace{1mm} \right]}\\
    &\geq&
    1- \delta.
    \label{eqn:individ_P_1minusdelta}
\end{eqnarray}
This comes from \equref{eqn:conf_corrollary_def}. With probability at least $1-\delta$ we can estimate $\expval{P}$ with our empirical average $\hat{o}(P)$ with precision 
    \begin{eqnarray}
        |\expval{P} - \hat{o}(P)| 
        &\leq& \sqrt{\frac{2 \log(2/\delta)}{h(P)}}.
    \end{eqnarray}
We obtain this bound on precision by solving for $\epsilon$ in  \ref{eqn:individ_P_1minusdelta}. 
As a result, the right side of this final expression is an upper bound on the precision $\epsilon$.  We can see how an upper bound on confidence can translate to a bound on precision.  
Finally, 
while we discussed above how to use the final cost function output to upper bound \textit{all} our Pauli string estimation errors, this bound across all Pauli strings is not very tight. Here, since we can consider each Pauli string one-at-a-time, we expect this bound on precision to be tighter than what follows from Theorem \ref{thm:theorem_performance_guarantee}. 
The estimation error bound will not be brought down by other observables with low hitting counts.

\vspace{5mm}
\noindent
\begin{center}
    \textit{The }\textsf{COST}\textit{ function's relationship with variance} \label{app:COSTrelationship_with_VARIANCE}
\end{center}

We will next 
deepen the reader's intuition for our DSS scheme by explicitly deriving our cost function’s relationship with the variance.
A ubiquitous metric for learning protocols, the variance is present - albeit slightly hidden - in our cost function in the form of the hitting count. 
The hitting count entered our Lemma \ref{lem:lemma_concentration_ineq} derivation due to our use of Hoeffding's inequality. 
Hoeffding's inequality can be derived using Hoeffding’s lemma and Chernoff bounds.  While one can go down this path (i.e. trace back to Markov's inequality) to uncover the dependence on variance, we will present a slightly different treatment.

Using the language of classical shadows, we can show how the variance is related to hitting count.
In order to use the language of classical shadows, we will use our $N$ output measurements from the DSS algorithm to define a new ensemble. 
In classical shadows there exists some ensemble $\mathcal{U}$ defining a set of unitaries $\{U\}$  and some associated probability density function $\mu(U)$. Here we will define our new ensemble, which we will call $\mathcal{U}_\text{DSS}$, to be over all measurement circuits $\{U^{DSS}_i\}_{i=1}^N$, and notice that we may measure with the same measurement circuit multiple times (for example, it's possible that $U_2 = U_{17} = U_{35}...$). 
Therefore, the associated probability density function (PDF) takes the form
    \begin{equation}\label{eqn:pdf_DSS}
       \mu(U_i^{DSS}) =  \frac{1}{N}\sum_j \mathds{1}_{U^{DSS}_i = U^{DSS}_j} \hspace{5mm} \textnormal{ where  }  \mathds{1}_{U^{DSS}_i = U^{DSS}_j} = 
           \begin{cases} 
            1 & U^{DSS}_i = U^{DSS}_j \\
             0 & U^{DSS}_i \neq U^{DSS}_j 
             \end{cases}
    \end{equation}
This PDF represents how often, across our $N$ measurements, we measure with the unitary $U_i^{DSS}$. 

We can imagine sampling from this ensemble in a classical shadows protocol -- doing this allows us to easily define our protocol's variance using the machinery of classical shadows. 
A central metric of any learning protocol, the variance of an estimated observable represents how efficiently we learn it. 
A smaller variance indicates that, across the shots taken in a protocol, the observable is often learned, and therefore our estimate is precise.
In classical shadows one often considers a state-agnostic upper bound on variance -- one often maximizes the variance over states in order to obtain a state-independent quantity.
In classical shadows one also drops the second term in the variance (recall: $\text{Var}(\expval{O}) = \expval{O^2} - \expval{O}^2$), and this is called the shadow norm ~\cite{huang2020predicting}:
    \begin{equation}
        \max_\rho \text{Var}_\rho(P) \leq \|P\|^2_\text{shadow}. 
    \end{equation}
For the reader new to the shadow norm, it can be a nice exercise to look up the full definition in Ref~\cite{huang2020predicting} and show that this quantity obeys the properties of a norm. In the proposition below, we will use this relation to upper bound the maximum variance over states with a function dependent on the hitting count. This will allow us to gain intuition on how minimizing our DSS \textsf{COST} function reduces the variances of our desired Pauli strings.

\begin{proposition} \label{prop:variance_and_hittingcount}
Consider the ensemble $\mathcal{U}_\text{DSS}$ defined on the set of $N$ measurement circuits returned from the DSS algorithm (Algorithm \hyperref[alg:DSS]{1}) and associated probability density function $\mu$ defined in equation ~\ref{eqn:pdf_DSS}. Then the max variance over states can be upper bounded by
    \begin{equation}
        \max_\rho \textnormal{Var}_\rho(\expval{P})  
        \leq 
        \frac{N}{h(P)},
    \end{equation}
where $h(P)$ is the hitting count (Definition~\ref{def:hittingcount}) of how many times the DSS protocol measures $P$.
\end{proposition}
\begin{proof}
For the measurement primitive \cite{huang2020predicting} $\mathcal{U}_\text{DSS}$, we can define a classical shadows protocol that randomly samples unitaries from $\mathcal{U}_\text{DSS}$, which is by definition a subset of the Clifford group. 
As such, by Lemma 1 of Ref.~\cite{bertoni2022shallow}, the eigenoperators of the corresponding classical shadows measurement channel, $\mathcal{M}_\text{DSS}$, are Pauli strings. 
Therefore, we can express the action of the measurement channel as $\mathcal{M}_\text{DSS}(P) = \lambda_P P$, where $\lambda_P$ denotes the eigenvalue corresponding to the Pauli string eigenoperator $P$. As was shown in \cite{bertoni2022shallow,van2022hardware}, the eigenvalue of a classical shadows measurement channel takes the form 
\begin{eqnarray}
    \lambda_P 
        &=& \frac{1}{2^n} \Expect_{U\sim\mathcal{U}_\text{DSS}} \sum_{b}  \bra{b} U P U^\dagger \ket{b}^2 \\
        &=& \frac{1}{2^n} \sum_i  \frac{1}{N} \sum_{b}  \bra{b} U^{DSS}_i P U_i^{DSS\dagger} \ket{b}^2  \label{eqn:eval_lambdaeqn_DSS}
\end{eqnarray}
Notice that the Pauli string $P$ here is not normalized with respect to the Hilbert Schmidt inner product, and therefore, we compensate for this with the $1/2^n$ prefactor out front.
One can also easily derive this relation using the equation $\tr(\mathcal{M}_\text{DSS}(P) P) = \lambda_P \tr(P^2)$ and substituting in the definition of the measurement channel from Ref~\cite{huang2020predicting}.
We move from the first to the second line using the definition of  $\mathcal{U}_\text{DSS}$.

It turns out that this classical shadows measurement channel eigenvalue can also be related to the shadow norm and, therefore, the maximum variance over states. 
Take the shadow norm \cite{huang2020predicting} of the Pauli string $P$ and substitute in our eigenoperator equation $\mathcal{M}_\text{DSS}^{-1}(P) = \frac{1}{\lambda_P} $;  assume $\mathcal{M}_\text{DSS}^{-1}$ is the (Moore-Penrose) pseudoinverse if $\mathcal{M}_\text{DSS}$ has a nontrivial nullspace. 
Following the calculation through (see Refs~\cite{bertoni2022shallow,ippoliti2023operator}), one can show that $\|P\|^2_\textnormal{shadow} = 1/\lambda(P)$. In other words, our maximum variance over states can be upper bounded by the inverse of the classical shadows measurement channel eigenvalue:
    \begin{equation}
        \max_\rho \text{Var}_\rho(P)  
                \leq \|P\|^2_\textnormal{shadow} 
                = \frac{1}{\lambda_P}
    \end{equation}

At this point, it suffices to derive the Pauli string eigenvalue $\lambda_P$'s 
relationship to the hitting count $h(P)$ of that same Pauli string.
Let us return to the formal definition (Definition \ref{def:PauliWeightUm}) of the Pauli weight for measurement ensemble $\mathcal{U}_i$: $p_i(P) = \frac{1}{2^n} \Expect_{U\sim \mathcal{U}_i} \sum_{b}  \bra{b} U P U^\dagger \ket{b}^2$. 
However, since the DSS algorithm has terminated and returned a set of derandomized measurement circuits $\{U_i^{DSS}\}_i$, we find
    \begin{equation}\label{eq:prob_m}
            p_i(P) = 
            \frac{1}{2^n} \sum_{b}  \bra{b} U^{DSS}_i P U_i^{DSS\dagger} \ket{b}^2,
    \end{equation}
where $p_i(P) \in \{0,1\}$ depending on whether $U^{DSS}_i$ diagonalizes $P$. 
We can now substitute this into the eigenvalue \equref{eqn:eval_lambdaeqn_DSS}, and we find
        \begin{equation}
            \lambda_P = \frac{1}{N} \sum_i p_i(P) = \frac{1}{N} h(P).
        \end{equation}
The claim follows. 
\end{proof}

This proposition upper bounds the maximum variance over states with a function of the hitting count, allowing us to gain intuition on how minimizing our DSS \textsf{COST} function affects the variances of our desired Pauli strings $P$. 
In particular, consider the DSS cost function (below). Minimizing this cost function corresponds to choosing circuits that have large hitting counts $h(P)$ across our Paulis of interest. 
    \begin{equation}
        \text{\textsf{COST}}\left(\{\mathcal{U}_i\}_{i=1}^N\right) 
    = 2 \sum_P \exp{\left[ \hspace{1mm} -\frac{\epsilon^2}{2} h(P) \hspace{1mm} \right]}
    \end{equation}
In terms of the variance, the maximum variance over states, for Pauli string $P$, has an upper bound $\sim 1/h(P)$. Therefore, finding measurement circuits that yield a large $h(P)$ on our desired Pauli strings corresponds to small variances for the same $P$. 
This checks with our intuition: intuitively, we want to minimize the variance with which we learn $\expval{P}$ for the Pauli strings $P$ we care about.

We have learned that, in the DSS algorithm, when we minimize our cost function (\equref{eqn:DSScost_appendix}), we find circuits that reduce our upper bound on variance.
As a result, one could ask why we use the confidence as our cost function rather than some collective, convex combination of the upper bounds on the maximum variance over states.
And indeed, for an individual Pauli $P$, choosing either of these two quantities for the cost function is actually equivalent: the quantities' landscapes retain the same locations for maxima and minima.
However,  this relationship is only guaranteed when considering a single Pauli string. 
When learning many Pauli strings, the convex combination of confidences (our \textsf{COST}) is nice because it directly connects to a performance guarantee across \textit{all} observables we want to estimate. 
If we were to use a convex combination of variances as our cost function, this would not improve a guarantee across all Paulis we want to learn. For example, this convex combination of variances cost function would go down if we slightly reduced all the variances, and it would also go down is we substantially reduced one Pauli's variance but kept the other variances the same. 
In other words, the convex combination of variances does not exhibit a preference between reducing all variances versus just bringing down one. 
In our setup for learning many Pauli strings, we want to bring down all variances simultaneously, and therefore, we opt for the confidence.

\vspace{5mm}
\noindent
\begin{center}
    \textit{DSS performance at least as good as equivalent-depth shallow shadows} \label{app:asgoodas_shallow_shadows}
\end{center}
We will conclude this appendix by showing how to make a minor modification to DSS such that it performs at least as well as shallow shadows.
While we numerically observe in our applications that DSS outperforms shallow shadows, we also present a formal proposition (below).
Before we proceed notice that our depth definition is different from that of the shallow shadows literature \cite{hu2023classical,ippoliti2023operator,bertoni2022shallow}. 
Our ``depth'' $d = d_\text{DSS}$ is the number of two-qubit gate layers. See Figure \ref{fig:Figure1}. 
For example, $d_\text{DSS}=1$ represents one layer of two-qubit gates. In the shallow shadows literature, depth $d_\text{shallow}=1$ is \textit{two} layers of two-qubit gates, where the second layer is always offset by 1 qubit from the first layer. Our layers are also always offset by 1 -- again see Figure 1; we just count each layer independently.  We will redefine shallow shadows here in order to be able to make one-to-one comparisons between DSS and shallow shadow performance.

\begin{definition}
 [\textbf{Shallow Shadows}]\label{def:shallow_shadows}
    Given depth $d$ and system size $n$, $N$-shot shallow shadows utilizes an ensemble ansatz  $\mathcal{U}^\text{(shallow)}$ containing $d$ layers of two-qubit random Clifford gates, interleaved with $d+1$ layers of single-qubit random Clifford gates. 
    This ensemble ansatz is the same as the DSS Initial Ansatz of Definition~\ref{def:initial_ansatz} (same depth $d$, system size $n$). 
    The shallow shadows measurement protocol randomly samples $N$ circuits from this ensemble -- implementing them before measuring in the $Z$ basis. 
\end{definition}
Equivalently, we can also say that shallow shadows samples one circuit $U \sim \mathcal{U}_i$ per measurement ensemble in $\{\mathcal{U}_i\}_{i=1}^N$ where, for all $i$, $\mathcal{U}_i = \mathcal{U}^\textnormal{(shallow)}$. Finally notice that while we include single-qubit random Clifford gates, these can equivalently be absorbed by the two-qubit random Clifford rotations, and thus this definition is only different from the original works \cite{ippoliti2023operator, bertoni2022shallow, hu2023classical} in how we define depth.

We now proceed to show that our (slightly-modified) DSS will do at least as well as the equivalent-depth version of shallow shadows.
The slight modification we make, which does not affect the protocol's classical or quantum complexity, is to the gate options during derandomization.
In our main-text and Appendix \ref{app:DSSalgorithm} definitions of DSS, we allow our one- and two-qubit gates to derandomize to options 1-6 and 1-3 in Table \ref{tab:GateValues}, respectively. For example, a two-qubit gates can become (1) \textsf{Identity}, (2) \textsf{CNOT}, and (3) \textsf{SWAP}. 
However, we will now expand the set of options: we will allow our one- and two-qubit gates to derandomize to options 0-6 and 0-3 in Table \ref{tab:GateValues}. In other words, each gate has the option to remain randomly sampled from the local Clifford group (option 0 in Table \ref{tab:GateValues}). 
As such, the derandomization procedure will not always end with a delta distribution for all $N$ measurements -- each measurement circuit could be sampled from a nontrivial ensemble. 
A priori it may seem that this modification makes evaluating the \textsf{COST} function more intensive. But again, as we will see in Appendix \ref{app:TensorNetwork_for_cost_function}, this crucially does not change the asymptotics of the DSS algorithm's classical complexity.

\begin{proposition}\label{prop:better_than_shallow_shadows}
    Given the Pauli strings $\{P\}$ we want to learn, a budget of $N$ measurements of at most depth $d$, and desired precision $\epsilon$.
    The DSS Algorithm, with the modification that each gate in Algorithm \hyperref[alg:DSS]{1} has the option (option 0 in Table \ref{tab:GateValues}) to remain randomly sampled from the local Clifford group, outputs measurement ensembles $\{\mathcal{U}_i\}_{i=1}^N$. The value of the associated \textsf{COST} function will be at most the \textsf{COST} of shallow shadows, as defined in Definition~\ref{def:shallow_shadows}:
    \begin{equation}
        \textnormal{\textsf{COST}}
        \big(\{\mathcal{U}^\textnormal{(DSS)}_i\}_{i=1}^N\big) 
        \leq 
        \textnormal{\textsf{COST}}
        \big(\{U^\textnormal{(shallow)}\}_{i=1}^N\big)
    \end{equation}
\end{proposition}

\begin{proof}
Since the DSS algorithm 
employs a greedy strategy, it will only replace a gate in one of its ensembles when that choice of gate reduces the \textsf{COST} function value. 
Therefore at the end of Algorithm \hyperref[alg:DSS]{1}, the \textsf{COST} function will be at most the same as it was at the start of the algorithm, which by Definition \ref{def:shallow_shadows} is equivalent to \textsf{COST}($\{\mathcal{U}^\textnormal{(shallow)}\}_{i=1}^N$).
\end{proof}

This guarantee has implications on the performance of the DSS algorithm. Following Proposition~\ref{prop:better_than_shallow_shadows}, notice that the DSS protocol will also give a better lower bound on probability 
    \begin{equation}
        1 - \textnormal{\textsf{COST}}
        \big(\{\mathcal{U}^\textnormal{(DSS)}_i\}_{i=1}^N\big) 
        \geq 
        1 - \textnormal{\textsf{COST}}
        \big(\{\mathcal{U}^\textnormal{(shallow)}\}_{i=1}^N\big)
    \end{equation}
that all observables are learned to at least precision $\epsilon$.  This follows immediately from the proposition and the \textsf{COST} functions relationship bound on confidence (see Theorem~\ref{thm:theorem_performance_guarantee}). By theorem \ref{thm:theorem_performance_guarantee}, we can learn all Pauli strings $P$ to at least precision $\epsilon$ with probability at least $1 - \textsf{COST}\big(\{\mathcal{U}^\textnormal{(DSS)}_i\}_{i=1}^N\big)$.

Since our numerics consistently show that DSS outperforms shallow shadows (e.g. Figures \ref{fig:Figure2} and \ref{fig:Figure3}),
one could ask: was it actually necessary to modify the DSS algorithm to show that it would outperform equivalent-depth shallow shadows? 
While we expect our original (unmodified) DSS to always outperform equivalent-depth shallow shadows, we will not give a rigorous guarantee in this work. 
Instead, we will outline some intuition -- which we hope will be helpful scaffolding for showing this in future work.
In particular, consider the $N$ unitaries $\{U^{DSS}_i\}_{i=1}^N$ as defining a PDF --  as was done earlier in this appendix, see \equref{eqn:pdf_DSS}. 
As we derandomize each circuit, making it some deterministic unitary, we are concentrating our measurement budget on some set of Pauli strings. 
The hitting counts for some Pauli strings will go up while others go down. Recall that $h(P) = \sum_i p_i(P)$, and Pauli weight $p_i(P)$, which can be a fraction before derandomization, will change as the circuit is derandomized. 
Therefore, showing that DSS -- as originally formulated -- outperforms shallow shadows reduces to guaranteeing that our choice of unitaries concentrates hits on the Pauli strings we care about. 
Finally, notice that showing this for DSS of arbitrary depth immediately provides the equivalent guarantee for the Random Pauli setup of Ref~\cite{huang2021efficient}, as it is the depth $d=0$ version of DSS.

\newpage
\section{\label{app:TensorNetwork_for_cost_function} TENSOR NETWORK TECHNIQUE FOR EFFICIENT COST FUNCTION EVALUATION}

This appendix discusses how we use tensor network techniques to efficiently calculate our cost function during the derandomization procedure. Parts of this method were used in Ref \cite{2024arXiv240217911H}'s analysis of results, and these ideas rely on formalisms first introduced in Ref~\cite{bertoni2022shallow}. 
We first describe our technique for efficiently evaluating the cost function, and then we quantify its classical complexity. This leads us to a guarantee that the DSS algorithm can be performed efficiently. 
At first glance evaluating our cost function 
\begin{equation}
    \text{\textsf{COST}}\left(\{\mathcal{U}_i\}_{i=1}^N\right) 
    = 2 \sum_P w_P \prod_{i=1}^{N} \exp{\left[ \hspace{1mm} -\frac{\epsilon^2}{2} p_i(P) \hspace{1mm} \right]}
\end{equation} 
can seem prohibitive because we have to calculate the Pauli weights $p_i(P)$. 
To evaluate the Pauli weights, we must track how our (sometimes probabilistic) circuits transform $P$. Recall that our Pauli weight takes the form
\begin{equation}
     p_i(P) = \frac{1}{2^n} \Expect_{U \sim \mathcal{U}_i} \sum_{b}  \bra{b} U P U^\dagger \ket{b}^2,
\end{equation}  
as defined in Definition \ref{def:PauliWeightUm}. Na\"ively, this requires exponential resources: for example, using direct simulation $P \rightarrow UPU^\dagger$ we need space exponential in the system size. However, using our tensor network techniques, we can bypass this overhead and utilize resources only polynomial in the system size.
We map the Pauli weight computation onto a Markov chain, which can be represented with a tensor network. As such, our technique mimics random Markovian classical processes.  
Below we will pedagogically describe and evaluate our technique's complexity in order to guarantee that it efficiently calculates the cost function during derandomization. 
We show that our DSS algorithm is polynomial time when considering measurement circuit depth at most polylog($n$).

\vspace{5mm}
\noindent
\begin{center}
    \textit{Tensor network technique for estimating Pauli weights} 
\end{center}

Evaluating the DSS \textsf{COST} function reduces to evaluating the Pauli weights $p_i(P)$.
Looking at the $\mathcal{U}_i$ measurement ansatz, we can construct a tensor network for calculating $p_i(P)$ for some input Pauli string $P$.
We will show how to do this for a simple example and then discuss how to extend to higher depth and more complicated circuits. 
At a high level, the idea is to represent Pauli operators as vectors in a 
Hilbert space. We use a variant of the \textit{Pauli transfer matrix} (PTM) formalism, a tool commonly used in simulations of quantum dynamics and open quantum systems \cite{chow2012universal, aharonov2023polynomial, cerezo2023does, rall2019simulation}. 
As alluded to above, our technique mimics classical Markovian processes, and after we describe the technique below, we will prove our efficiency claims. 

To explain the Markov chain interpretation, we repackage the Pauli weight as
\begin{align}\label{eq:Markov_interpretation}
     p_i(P) = \frac{1}{2^n} \Expect_{U \sim \mathcal{U}_i} \sum_{b}  \bra{b} U P U^\dagger \ket{b}^2 = \frac{1}{4^n} \Expect_{U \sim \mathcal{U}_i} \sum_{z \in \mathcal{Z} } \tr[  U^{\otimes 2} P^{\otimes 2} U^{\dag \otimes 2}  \, z^{\otimes 2}]
\end{align}
where $\mathcal{Z}$ is the set of $n$-qubit Pauli strings made up of only $\mathbb{1}$s and $Z$s. Since the distribution $\mathcal{U}_i$ only includes Clifford gates, the operator $U^{\otimes 2} P^{\otimes 2} U^{\dag \otimes 2}$ is always a single Pauli string of the form $Q^{\otimes 2}$, with $U P U^{\dag} = \pm Q$. Therefore, we can express $\Expect_{U \sim \mathcal{U}_i}   U^{\otimes 2} P^{\otimes 2} U^{\dag \otimes 2} $ as a sum of doubled Pauli strings,
\begin{align}\label{eq:pauli_basis_expansion}
    \Expect_{U \sim \mathcal{U}_i}   U^{\otimes 2} P^{\otimes 2} U^{\dag \otimes 2} =\sum_{Q \in\{\mathbb{1},X,Y,Z\}^{\otimes n} } c_Q \, Q^{\otimes 2},
\end{align}
with the additional properties
\begin{align}
     \quad c_Q \geq 0 \quad \text{and} \quad \sum_Q c_Q = 1, \quad \forall Q.
\end{align}
The sum to unity implies we can interpret the coefficients $c_Q$ as a probability distribution. 
Looking at the last term in \equref{eq:Markov_interpretation}, we see that only coefficients $c_Q$ for strings made entirely out of $\mathbb{1}$s and $Z$s will contribute to $p_i(P)$, and in fact $p_i(P)$ is exactly the probability that $P^{\otimes 2} $ gets taken to some doubled $\mathbb{1},Z$-string.

\addtocounter{figure}{+1}
\begin{figure}[ht]
\centering
\includegraphics[scale=0.4]{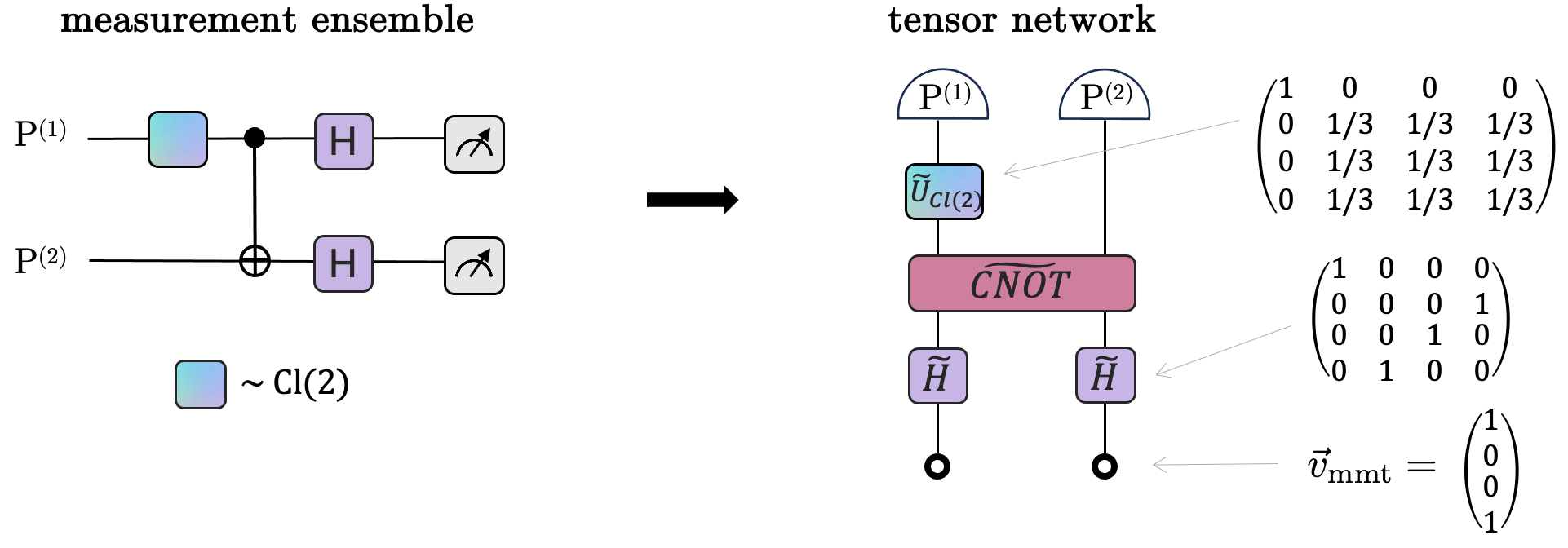}
\caption{
\emph{Example circuit and its associated tensor network.}  
Consider the measurement ensemble $\mathcal{U}$ defined by this example circuit. Using tensor network techniques, we can solve for the probability $p(P)$ that this measurement ensemble measures the Pauli string $P = P^{(1)} \otimes P^{(2)}$.
Looking at the circuit, we can immediately notice that the circuit diagram (left) maps to a network of tensors (right). Each gate, measurement, and input Pauli string becomes a tensor, as defined in the main text. We give examples of what some of the tensors look like on the right (see \equref{eqn:TN_hadamard_tensor}, \equref{eqn:TN_twirl_single_qubit}, and \equref{eq:v_measurement}). This setup allows us to directly evaluate $p(P)$ for this example ensemble.
}
\label{fig:TN_example_circuit}
\end{figure}

Our tensor network algorithm then is equivalent to updating the distribution over doubled Pauli strings through successive layers of gates. Concretely, to every initial Pauli, we associate a probability vector that encodes the initial distribution. To every (possibly twirling) gate, we associate a stochastic matrix that updates the distribution, and to every measurement site we associate a vector that picks out the probabilities of $\mathbb{1},Z$ strings. 
We discuss the form of these tensors below, but for more details on their derivations, we refer the reader to \cite{bertoni2022shallow,2024arXiv240217911H}. 
Once the tensors for each element are assigned, the index contraction scheme is given directly by the circuit diagram, as seen in the correspondence between the two parts of Supplementary Figure \ref{fig:TN_example_circuit}. After we list the forms of the various tensors, we will discuss this Supplementary Figure \ref{fig:TN_example_circuit} example in detail (see \equref{eqn:TN_example_exact_eval}).

\vspace{2mm}\noindent
\textit{1. Initial Pauli string} - Because we work in a basis which simply doubles a Pauli string over two copies, our encoding of the initial Pauli looks just like a single-copy encoding. We represent our initial Pauli with a tensor product of vectors encoding the single-site Paulis -- i.e. for single-site Pauli $P_s$ at site $s$, the corresponding vector $\vec{v}^{(s)}$ to be plugged into the tensor contraction looks like 
\begin{equation}\label{eq:Pauli_basis}
    P_s^{\otimes 2} \to \vec{v}^{(s)}, \quad \vec{v}^{(s)} = \begin{cases}
        \hspace{1mm}\begin{pmatrix}
            1&0&0&0
        \end{pmatrix}^\text{T}, & \vspace{2mm} \hspace{4mm}P_{s}^{\otimes 2} = \mathbb{1}^{\otimes 2}\\
        \hspace{1mm}\begin{pmatrix}
            0&1&0&0
        \end{pmatrix}^\text{T}, & \vspace{2mm}\hspace{4mm}P_{s}^{\otimes 2} =X^{\otimes 2}\\
        \hspace{1mm}\begin{pmatrix}
            0&0&1&0
        \end{pmatrix}^\text{T}, & \vspace{2mm}\hspace{4mm}P_{s}^{\otimes 2} =Y^{\otimes 2}\\
        \hspace{1mm}\begin{pmatrix}
            0&0&0&1
        \end{pmatrix}^\text{T}, & \vspace{2mm}\hspace{4mm}P_{s}^{\otimes 2} =Z^{\otimes 2}\\
    \end{cases}.
\end{equation}
In fact, because our basis is the same size as the basis for single-copy Paulis and the only effect of the doubled copy is to remove the need to track signs, we will often suppress the ``$\otimes 2$'' exponent in what follows. We make statements, like ``the Hadamard preserves $Y$,'' which are only true at the single-copy level if signs are disregarded. Note that although it is inefficient to represent explicitly, the $4^n$ component vector $\vec{v}_P = \bigotimes_s \vec{v}^{(s)}$ is 
our initial (delta) distribution over Paulis.

\vspace{1mm}\noindent
\textit{2. Gates - }We can now represent gates, whether fixed or averaged over the Clifford group, with stochastic matrices determined by how the gate transforms elements of our doubled Pauli basis. When a gate is fixed, the corresponding stochastic matrix is simply a permutation matrix, as a single Clifford gate permutes the Pauli group. For example, the Hadamard gate preserves Identity and $Y$ but swaps $X$ and $Z$ and therefore acts on our basis as the permutation
\begin{align}\label{eqn:TN_hadamard_tensor}
    \widetilde{H} = \begin{pmatrix}
    1 & 0 & 0 & 0 \\
    0 & 0 & 0 & 1\\
    0 & 0 & 1 & 0 \\
    0 & 1 & 0 & 0    
\end{pmatrix},
\end{align}
where again the first column/row represents $\mathbb{1}$, the second represents $X$, the third $Y$, and the fourth $Z$. Here we use the tilde notation $\widetilde{H}$ to represent the gate's stochastic matrix description, or equivalently the action of $H^{\otimes 2} \cdot \hspace{0.3 mm} (H^\dag)^{\otimes 2} $ on the space of doubled Paulis.
The matrix $\widetilde{H}$  is nearly the same as the standard PTM representation of Hadamard \cite{chow2012universal, aharonov2023polynomial, cerezo2023does, rall2019simulation}. The only difference is all signs in the single-copy transfer matrix are dropped to give the two-copy transfer matrix. This follows from the fact that, for example, $H Y H^\dag= - Y$ but $H^{\otimes 2} \,Y^{\otimes 2} H^{\dag \otimes 2} = +  Y^{\otimes 2}$. All 6 permutations of $\{X,Y,Z\}$ can be implemented with single-qubit Cliffords, giving us the (single-qubit) rotations corresponding to entries 1-6 of Table~\ref{tab:GateValues}.

The two-qubit case is similar. A fixed two-qubit gate $U$ will be represented by a $16\times 16$ stochastic (in fact, permutation) matrix whose entry at row $P'$ and column $P$ is 1 when $P'  = \pm {U} P {U}^\dagger$  and 0 otherwise. In general, the matrix element $\widetilde{U}_{P',P}$ for any fixed Clifford $U$ on $q$ qubits is 
\begin{equation}\label{eqn:TN_rule_for_U_construction}
 \widetilde{U}_{P',P} = 
 \frac{1}{2^q}
 |\Tr[P' U P U^\dagger]|.
\end{equation}
This takes care of representing our fixed two-qubit gates \textsf{Identity}, \textsf{SWAP}, and \textsf{CNOT}.

Alternatively, to represent a gate, which is twirled over the Clifford group, we simply average the transfer matrices for each Clifford.  In the single-qubit case, this gives 
\begin{align}\label{eqn:TN_twirl_single_qubit}
    \widetilde{U}_{\text{Cl}(2)}  \to \begin{pmatrix}
    1 & 0 & 0 & 0 \\
    0 & 1/3 & 1/3 & 1/3\\
    0 & 1/3 & 1/3 & 1/3 \\
    0 & 1/3 & 1/3 & 1/3   
\end{pmatrix},
\end{align}
which agrees with the intuitive notion that an average over single-qubit Cliffords should treat all Paulis -- except the identity -- on equal footing. Identity always goes to identity, and this transformation will take any nontrivial Pauli (i.e. $\{X,Y,Z\}$) to any of the others with probability 1/3.
Again, the two-qubit case is similar. 
When we randomly sample from the two-qubit Clifford group $\text{Cl}(2^2)$, we equally permute all 15 non-identity two-qubit Pauli strings. As such, $\widetilde{U}_{\text{Cl}(2^2)}$ looks like \equref{eqn:TN_twirl_single_qubit}, but now we ``twirl'' among 15 Pauli strings rather than 3. Therefore, the matrix will be block diagonal with a $15 \times 15$ block of $\frac{1}{15}$s rather than the $3 \times 3$ block of $\frac{1}{3}$s.

\vspace{1mm}\noindent
\textit{3. Measurement -} We have described how to represent every fixed or twirling gate with its corresponding transfer matrix. Once the probability distribution over Paulis has been updated through all the gates, we simply need to pick out the probabilities associated with strings made entirely out of $\mathbb{1}$s and $Z$s. This can be done using the ``measurement vector'' $\vec{v}_\text{mmt}$, where contraction at every measurement site with
\begin{equation}\label{eq:v_measurement}
    \vec{v}_\text{mmt} = 
    \begin{pmatrix}
        1 \\ 0 
        \\ 0 \\ 1    
    \end{pmatrix}
\end{equation}
precisely picks out the $\mathbb{1}$ and $Z$ components.

\vspace{5mm}
Putting this all together, we can solve for the Pauli weight of the example measurement ensemble posed 
in Supplementary Figure \ref{fig:TN_example_circuit}. This ensemble has one random gate (the upper left gate), and the rest are fixed. Looking at Supplementary Figure \ref{fig:TN_example_circuit}, given our tensor constructions above, we can immediately notice that the circuit diagram (left) maps to a network of tensors (right). Each gate, measurement, and input Pauli string becomes a tensor, as defined above. 
If one were to directly evaluate $p(P)$ for our example, utilizing these specifications, the calculation would look like 
    \begin{equation}
    \label{eqn:TN_example_exact_eval}
        p(P) = \left(\otimes_{s=1}^n \vec{v}_\text{mmt}^\text{T} \right)
        \hspace{1mm} (\widetilde{H} \otimes \widetilde{H}) 
        \hspace{1mm} \widetilde{\textsf{CNOT}} 
        \hspace{1mm} (\widetilde{U}_{\text{Cl}(2)} \otimes \mathbb{1}_4)
        \hspace{1mm} \vec{v}_P,
    \end{equation}
where $\widetilde{H}$ is defined as in \equref{eqn:TN_hadamard_tensor}, $\widetilde{U}_{\text{Cl}(2)}$ is defined by \equref{eqn:TN_twirl_single_qubit}, \textsf{CNOT} is defined by \equref{eqn:TN_rule_for_U_construction}, $\mathbb{1}_4$ is the $4 \times 4$ identity matrix, and the initial Pauli is given by $\vec{v}_P = \bigotimes_{s = 1,2} \vec{v}^{(s)}$.
Notice that we used our circuit to guide how we connect the tensors.

While we have explicitly written out this expression in \equref{eqn:TN_example_exact_eval} and could evaluate it exactly, representing the full $4^n$ size space is exponentially costly. 
Instead, one can use a tensor network package, which can more efficiently perform this calculation when the depth $d$ is bounded.
Specifying the tensors and their connectivity, one can efficiently contraction these tensors using established algorithms in tensor networks~\cite{Gray2021hyperoptimized,Gray2018,fishman2022itensor}.
Note that many packages can both determine a good contraction order and also re-represent some tensors if they have more efficient descriptions during the calculation~\cite{Gray2021hyperoptimized,Gray2018,roberts2019tensornetworklibraryphysicsmachine}. 

At this point we have set up the tensors, discussed how to write down all tensors used in the DSS algorithm, and defined the network to calculate $p(P)$.
While all of this has an eye towards implementation and practicality, this technique is special as it allows us to efficiently\textemdash i.e. in polynomial time in the system size\textemdash evaluate $p(P)$ for bounded-depth measurement ansatzes. The next section will discuss this efficiency, and before we delve into this discussion, below we note a trick one can play during derandomization, using the ``signature basis.''

\vspace{3mm} 
\begin{center}
    \textit{Aside: Improving efficiency with the signature basis }
\end{center}

We have explained how to carry out tensor network contractions in the Pauli basis given by \equref{eq:Pauli_basis}. We could use this basis for all contractions carried out by our algorithm, but in fact when all single-qubit gates are twirled (as is always true during the structure derandomization step of Figure~\ref{fig:Figure1}(a.ii)), we can work in the more compact ``signature'' basis first analyzed in Ref.~\cite{bertoni2022shallow}. This allows us to reduce the bond dimension for structure derandomization from $4^{d-1}$ to $2^{d-1}$. 

As noted earlier, twirling over single-qubit Cliffords washes out the distinction between $X$, $Y$, and $Z$ as, e.g.,
\begin{align}
    \mathbb{E}_{U\sim \text{Cl}(2)}\,[U^{\otimes 2} Z^{\otimes 2} (U^\dag)^{\otimes 2}] = \frac{X^{\otimes 2}+Y^{\otimes 2}+Z^{\otimes 2}}{3},
\end{align}
but twirling of course preserves the distinction between trivial ($\{\mathbb{1}\}$) and nontrivial ($\{X,Y,Z\}$) Paulis. Since for structure derandomization every two-qubit gate is sandwiched between layers of single-qubit twirls, the distinction between $X$, $Y$, and $Z$, is continually washed out, and we only need to track how subsequent layers of two-qubit gates probabilistically change the support of incoming Pauli strings. 

We now list the tensor substitution rules for the initial Pauli string, the two-qubit gates, and the selection on $\mathbb{1},Z$ strings, referring the reader to Ref. \cite{bertoni2022shallow} for more formal discussion of the signature basis.
We encode the support of our initial doubled Pauli as
\begin{align}
P_i^{\otimes 2} \to
\begin{cases}
\begin{pmatrix}
    1 & 0 
\end{pmatrix}^\text{T}, & P_i^{\otimes2} = \mathbb{1}^{\otimes2} \vspace{0.1 in} \\ 
\begin{pmatrix}
    0 & 1 
\end{pmatrix}^\text{T}, & P_i^{\otimes} \in \{X^{\otimes2},Y^{\otimes2},Z^{\otimes2}\}
\end{cases}.
\end{align}
Given a 2-dimensional basis for each site, we can represent each option for a two-qubit gate with a $4 \times 4 $ matrix, and we list these explicitly before providing intuition:  
\begin{align} 
\textsf{Identity} \to \begin{pmatrix}
    1 & 0 & 0 & 0 \\
    0 & 1 & 0 & 0 \\
    0 & 0 & 1 & 0 \\
    0 & 0 & 0 & 1 \\
\end{pmatrix},&\qquad
\textsf{SWAP} \to \begin{pmatrix}
    1 & 0  & 0 & 0 \\
    0 & 0  & 1 & 0 \\
    0 & 1  & 0 & 0 \\
    0 & 0  & 0 & 1
\end{pmatrix}, \vspace{1 in} \\ \label{eq:CNOT_sig}
\textsf{CNOT} \to \begin{pmatrix}
    1 & 0  & 0 & 0 \\
    0 & 1/3  & 0 & 2/9 \\
    0 & 0  & 1/3 & 2/9 \\
    0 & 2/3  & 2/3 & 5/9
\end{pmatrix},&\qquad
\text{Cl}(2^2) \to \begin{pmatrix}
    1 & 0  & 0 & 0 \\
    0 & 1/5  & 1/5 & 1/5 \\
    0 & 1/5  & 1/5 & 1/5 \\
    0 & 3/5  & 3/5 & 3/5
\end{pmatrix}.
\end{align}
Recall that in the Pauli basis, an element of a gate's stochastic matrix gave the probability that a particular input Pauli be taken to a particular output Pauli. In the signature basis, a gate's stochastic matrix gives the probability that a Pauli with a particular input \textit{support} be taken to a Pauli with a particular output \textit{support}. For example, consider the action of $\textsf{CNOT}$ on all Paulis supported only on the second qubit: 
\begin{align}
    \textsf{CNOT}\, I \!\otimes\!  X  \, \textsf{CNOT} = I \!\otimes\!  X,\\
    \textsf{CNOT} \, I \!\otimes\!  Y \, \textsf{CNOT} = Z \!\otimes\!  Y,\\
    \textsf{CNOT} \,I \!\otimes\!  Z \,\textsf{CNOT} = Z \!\otimes\!  Z.
\end{align}
If we take of each of these inputs to be equally likely (because the \textsf{CNOT} is preceded by single-qubit twirls), the output will be supported only on the second qubit with probability $1/3$ and supported on both qubits with probability $2/3$. This gives exactly the second column of the $\textsf{CNOT}$ matrix in \equref{eq:CNOT_sig}. In general, for any fixed or twirled gate the signature basis stochastic matrix can be derived from the Pauli basis stochastic matrix for that gate. A given signature basis matrix element is simply the sum of all matrix elements of the Pauli basis matrix consistent with the given input and output supports, divided by the number of allowed input Paulis. 
Finally, at every measurement site we contract against 
\begin{align} 
\vec{v}_{\rm{mmt,sig}} = \begin{pmatrix}
    1 \\ 1/3
\end{pmatrix}
\end{align}
because if the evolved Pauli is trivial at the measurement site this is always consistent with having an $\mathbb{1},Z$ string, but if it is nontrivial and therefore equally likely to be $X,$ $Y,$ or $Z$ so the probability of measuring $Z$ drops to 1/3.

Contraction for calculation in the signature basis proceeds according to the same circuit diagrams as in the Pauli basis. However, there is no need to explicitly represent single-qubit twirls, as their effect has already been taken into account by switching to the signature basis.

\vspace{5mm}
\noindent
\begin{center}
    \textit{The DSS algorithm is efficient} 
\end{center}

In this section we first provide intuition for why our tensor network calculation can be performed in polynomial time for depth $d = O(\textnormal{poly}\log(n))$ measurement ansatzes. We then provide mathematical background on this result, showing both correctness and efficiency and thus providing a formal performance guarantee.

\begin{theorem}
    Suppose we are given a set $\{P\}_P$ of Pauli strings on $n$ qubits and a measurement budget of $N$ at most polylog-depth circuits, $d = O(\textnormal{poly}\log(n))$. If the number of observables $|\{P\}_P|$ and number of measurements $N$ are not superpolynomial, we can perform our DSS algorithm in time $\textnormal{poly}(n)$.
\end{theorem}
The proof follows by construction. We propose an algorithm, which formalizes the tensor network technique in the previous section's example and achieves time complexity $O(\textnormal{poly}(n))$ when depth $d = O(\textnormal{poly}\log(n))$.
The cost function must be evaluated for all options of all gates, $O(nd)$, in all measurement ensembles, $O(N)$. Therefore, we require $O(N n d)$ calls to \textsf{COST},
\begin{equation}
    \text{\textsf{COST}}\left(\{\mathcal{U}_i\}_{i=1}^N\right) 
    = 2 \sum_P w_P \prod_{i=1}^{N} \exp{\left[ \hspace{1mm} -\frac{\epsilon^2}{2} p_i(P) \hspace{1mm} \right]},
\end{equation}
and since we only are updating one gate (in one measurement ensemble) at a time, we only need to calculate one set of Pauli weights.
The set we need to calculate is the $j$th set, $\{p_j(P)\}_P$, where $j$ corresponds to the measurement $j$ we are currently derandomizing. Therefore, each of our calls to \textsf{COST} will only need to calculate $|\{P\}_P|$ Pauli weights $\{p_j(P)\}_P$, where  $|\{P\}_P|$  is the size of the set of Pauli strings we want to learn.

To determine the total time complexity of the DSS algorithm, we need to determine the time complexity $t_{PW}$ of calculating a single Pauli weight $p_i(P)$.
The DSS algorithm's time complexity is then $O(N n d \times |\{P\}| \times t_{PW})$. As we will see below, the complexity of $t_{PW}$ grows exponentially in the depth $d$ of our ansatz, and therefore, $t_{PW} = O(\textnormal{poly}(n))$ as long as $d$ is $\textnormal{polylog}(n)$. The theorem then follows assuming both $N$ and $|\{P\}|$ are at most polynomial. 

The remainder of this appendix is dedicated to showing $t_{PW} = O(\textnormal{poly}(n))$ holds when utilizing our tensor network technique. 
We proceed in two steps: correctness and efficiency. 
To establish correctness of the weight calculation, we need to demonstrate that the general formula for $p_i(P)$ (\equref{eq:prob_m}) can be expressed in terms of contractions between the tensors we use in our numerics.
A key property both here and in establishing the efficiency of our algorithm is that whenever we evaluate a Pauli weight $p_i(P)$ for a measurement ensemble $\mathcal{U}_i$, every gate is either completely fixed or averaged over some distribution independent of all other gates in the circuit. As such, expectation values over the ensemble can be factorized into the composition of super-operator-valued expectations of each gate, and these gate-level superoperators can be explicitly calculated.
To establish efficiency, we need to provide a schedule for tensor contraction which ensures we never have to manipulate objects of exponentially large dimension. This is easiest to explain diagramatically, as we show in Supplementary Figure \ref{fig:contraction}. From the contraction strategy, the simulation can be carried out with a cost exponential only in the circuit depth $d$ or equivalently polynomial in $n$ provided the depth is restricted to $O(\textnormal{poly}\log(n))$ or smaller. In practice we are able to make the base of the exponential-in-depth scaling better than the naive estimate by working entirely within reduced subspaces of operator space, which are left invariant by our circuit ensembles.

\vspace{3mm}
\noindent
\textbf{Correctness.} 
First, we derive the tensors used in our network from \equref{eq:prob_m}. We will show how building up our ensemble of circuits from gates, which are independently distributed, leads to an efficient calculation of measurement probabilities using tensor networks. 
As a reminder, our goal is to calculate
\begin{align}
     p_i(P) = \frac{1}{2^n} \Expect_{U \sim \mathcal{U}_i} \sum_{b}  \bra{b} U P U^\dagger \ket{b}^2 = \frac{1}{4^n} \Expect_{U \sim \mathcal{U}_i} \sum_{z \in \mathcal{Z} } \tr[  U^{\otimes 2} P^{\otimes 2} U^{\dag \otimes 2}  \, z^{\otimes 2}]
\end{align}
 where $\mathcal{Z}$ is the set of $n$-qubit Pauli strings made up of only $I$s and $Z$s. It is important to distinguish between the tensor product over sites $\bigotimes_{i=1}^n$ and the tensor product between copies $\cdot^{\otimes 2}$. 
 The tensor product between copies comes from the fact that we are calculating a second moment property, i.e. a property not of $U$ but of $U^{\otimes 2}$. In what follows, a tensor product in an exponent  (e.g. $U^{\otimes 2}$) will always refer to the tensor product between the two copies, whereas a tensor product standing on its own (e.g. $P_1\otimes P_2$ or $\bigotimes_{i=1}^n P_i$) will refer to a tensor product between different physical sites. We can then repackage our formula for $p_i$ as
\begin{align} \label{eq:nearly_TN}
     p_i(P) =  \Expect_{U \sim \mathcal{U}_i}  \tr\left[  U^{\otimes 2} P^{\otimes 2} U^{\dag \otimes 2}  \, \bigotimes_{i=1}^n \frac{\mathbb{1}_i^{\otimes 2}+ Z_i^{\otimes 2}}{4}\right]. 
\end{align}
From the form of \equref{eq:nearly_TN} we begin to see why a tensor network approach may be useful. 
We have an initial operator $P^{\otimes 2}$ which is a tensor product over sites. We evolve it through our ensemble of circuits $\mathbb{E}_U U^{\otimes 2} \cdot U^{\dag \otimes 2}$ and then take the trace against an operator $\bigotimes_i (\mathbb{1}_i^{\otimes 2} + Z_i^{\otimes 2})/4$ which factorizes over \textit{sites} (though not over the two copies).
If $\mathbb{E}_U U^{\otimes } \cdot U^{\dag \otimes 2}$ does not entangle distant sites, we can simulate this evolution for large system sizes with tensor networks. 
\begin{figure}
    \centering
    \includegraphics[width=\linewidth]{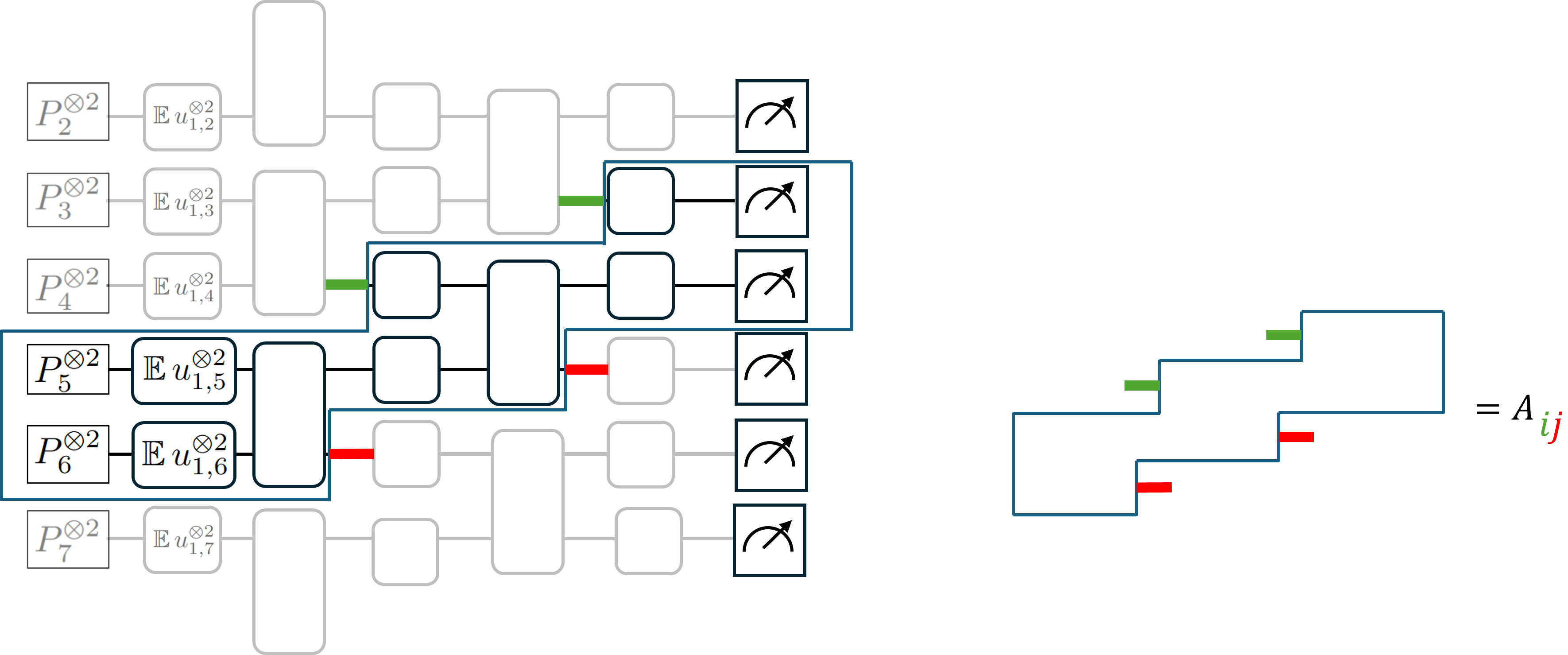}
    \caption{Diagram showing the contraction scheme for our tensor network approach to calculating Eqs. (\ref{eq:nearly_TN}), (\ref{eq:gate_decomp}). Once the internal contractions (contractions which lie entirely within the blue outline) are carried out, the final probability corresponds to the trace of the product of $n/2$ matrices with dimension exponential in depth $d$. The initial sites and gate locations stand for two-copy Paulis and superoperators respectively, as detailed in the text, and each measurement site stands for a trace against the operator $(\mathbb{1}^{\otimes 2}+Z^{\otimes 2})/4$.}
    \label{fig:contraction}
\end{figure}
To this end, it is useful to write out $U$ as a product of individual layers of single-qubit or two-qubit gates. Then averaging over $U$ can be expressed in terms of averaging over the gates, whose distributions are always independent from each other in our protocol. Thus we write
\begin{align}
    U = U^{(s)}_{d+1}&\prod_{l=1}^d U^{(t)}_l U^{(s)}_l\\
    U^{(t)}_l &= \bigotimes_i^{n/2} u^{(t)}_{l,i}\\ 
    U^{(s)}_l &= \bigotimes_i^{n} u^{(s)}_{l,i} 
\end{align}
where $U^{(t)}_l$ ($U^{(s)}_l$) means the $l$th layer of two-qubit (single-qubit) gates making up $U$, and  $u^{(t)}_{l,i}$, ($u^{(s)}_{l,i}$) is the $i$th two-qubit (single-qubit) gate in $U_l$. Recall from Appendix \ref{app:DSSalgorithm} that the two qubit gates are always offset by one in subsquent layers. 

Now we consider again the superoperator which comes from averaging over $U$, $\Expect_{U \sim \mathcal{U}_i}  U^{\otimes 2} \cdot U^{\dag \otimes 2}$. Our distribution over $U$ always comes from a product distribution over gates, i.e. for every two-qubit (single-qubit) gate we have a distribution $\mathfrak{u}^{(t)}_{l,i}$ ($\mathfrak{u}^{(s)}_{l,i}$), and these determine the distribution over circuits according to
\begin{align}
    \Expect_{U \sim \mathcal{U}_i}  U^{\otimes 2} \cdot U^{\dag \otimes 2} = & \bigg(\bigotimes_i^{n} \Expect_{u^{(s)}_{d+1,i} \sim \mathfrak{u}^{(s)}_{d+1,i}} u^{(s) \otimes 2}_{d+1,i} \cdot u^{(s) \dag\otimes 2}_{d+1,i} \bigg) \times\nonumber\\
    \hspace{1in}&\prod_l^d \bigg(\bigotimes_i^{n/2} \Expect_{u^{(t)}_{l,i} \sim \mathfrak{u}^{(t)}_{l,i}} u^{(t)\otimes 2}_{l,i} \cdot u^{(t) \dag \otimes 2}_{l,i} \bigg)\bigg(\bigotimes_i^{n} \Expect_{u^{(s)}_{l,i} \sim \mathfrak{u}^{(s)}_{l,i}} u^{(s) \otimes 2}_{l,i} \cdot u^{(s)  \dag \otimes 2}_{l,i} \bigg). \label{eq:gate_decomp}
\end{align}
Here, multiplication corresponds to superoperator composition.
At this point, 
the calculation reduces to evolving $P^{\otimes 2}$ through a composition of local superoperators and then calculating the overlap of this evolved ``state'' with a (product) distribution over $I^{\otimes 2},Z^{\otimes 2}$-strings. By performing internal contractions as shown in Supplementary Figure \ref{fig:contraction}, this calculation can be reduced to taking the trace of a product of matrices, whose dimension scales exponentially in depth. 
While one could naïvely simulate evolution through these channels using matrix product state of bond dimension $16^{d-1}$, in practice the unitary ensembles' structure allows us to work in more compact bases (described above). This results in matrices of bond dimension $4^{d-1}$ when some single-qubit gates are fixed and bond dimension $2^{d-1}$ when all single-qubit gates are twirled \cite{bertoni2022shallow,2024arXiv240217911H}.

\vspace{3mm}
\noindent
\textbf{Efficiency.} 
Once we represent the action of (possibly randomized) gates with tensors of fixed size, we need a schedule for carrying out the tensor contraction.
This can be done efficiently following the strategy in Supplementary Figure \ref{fig:contraction}, see also Ref.~\cite{bertoni2022shallow}. We first pair up the qubits and then divide the circuit into slices connected to each pair. This can be done by slicing the circuit into staircase shapes (again, see Supplementary Figure \ref{fig:contraction}), where each two-qubit gate in the staircase is shifted up one site from the preceding two-qubit gate. 
Without loss of generality, assume an even number of qubits (we can always add one if $n$ is odd, and this won't affect the scaling).  Since we have periodic boundary conditions and assume an even number of qubits, this slicing strategy is invariant under two-qubit translation. As shown in the figure, slicing the circuit and performing contractions internal to a given slice still leaves uncontracted indices which connect different slices. There are $d-1$ of these above a given slice and $d-1$ of them below. We group these $d-1$ indices above and these $d-1$ indices below into one single row index and column index, respectively. Then, the remaining contraction is equivalent to multiplying the matrices corresponding to each slice and taking the trace. These matrices are $ 4^{d-1} \times 4^{d-1}$ for the case where some single-qubit gates are fixed and $2^{d-1} \times 2^{d-1} $ when all are twirled, thanks to the reduced bases we employ. With the depth restricted to at most $O(\textnormal{poly}\log(n))$, the calculation requires multiplying $O(n)$ matrices of dimension $\textnormal{poly}(n)$, and is therefore efficient.

\newpage

\section{\label{app:QuantumChemistry} QUANTUM CHEMISTRY NUMERICAL SIMULATIONS}

In this appendix, we discuss our quantum chemistry application: using DSS we estimate the ground state energy of various molecules. 
This is a common application for learning protocols that estimate Pauli strings, and as such it has become a ubiquitous benchmark:  authors of new learning protocols often compare to previous state-of-the-art techniques by performing their protocol on a common set of ``test'' molecules. 
We gave these test molecules to DSS, and as we saw in the main text, at depth $d=1$ DSS already beats all previous protocols using bounded depth measurements. 
We compared all schemes' estimation errors when using the same, fixed number of measurements ($N =1000$), and DSS's estimates of the ground state energies were the most precise. 
We can determine DSS's estimation error on these molecules because they are small -- we can compute the true ground state energy exactly and, therefore, compute the difference between DSS's estimate and the true ground state energy.

In this work we compare depth $d=1$ DSS to previous state of the art techniques, such as locally-biased classical shadows (“LCBS”) \cite{hadfield2022measurements}, Random Pauli derandomization  \cite{huang2021efficient}, shallow shadows (“Shallow”) \cite{bertoni2022shallow, ippoliti2023operator, hu2023classical}, and overlap grouping measurement (“OGM”) \cite{wu2023overlapped}. 
We will give a high-level, one-sentence description of each, but we encourage interested readers to reference the original works. 
First, LBCS is classical shadows but with a new estimator, which is locally optimized using knowledge of the observable we want to learn (here, the hamiltonian $H$) and the best classical description of the state. 
Next, random Pauli derandomization is effectively just the depth-$0$ version of DSS.  
We also compare to shallow shadows, which have a single layer of entangling gates (in our language, depth $d = 1$). We choose to compare with depth 1 shallow shadows in order to allow the same entanglement generation as in DSS.
Finally, we also compare to OGM, which reduces the learning problem to a graph coloring problem \cite{veltheim2024multiset} and only allows single-qubit gates in their measurement circuits. 

For the remainder of this appendix, we provide background on the numerical simulations performed in Figure~\ref{fig:Figure2}, comparing depth $1$ DSS to previous learning strategies.  
First, we discuss how we use DSS to estimate the ground state energy of quantum chemistry molecules, and we show a plot of how the estimation error decreases with more measurements. 
In order to make this plot and obtain a good estimate of the estimation error we had to perform many simulations because the variance is nontrivial. 
Therefore, in the subsequent subsection we discuss how we determine the necessary number of simulations. In the results presented in the main text, we always compare estimation error averaged over at least 500 simulations. 
Finally, in the remainder of this appendix,
we discuss an example regime in which DSS gives optimal results. This is $H_2$ on 4 qubits and was featured in Figure~\ref{fig:Figure2}(b).

\vspace{5mm}
\noindent
\begin{center}
    \textit{Using DSS to estimate the ground state energy of quantum chemistry molecules} 
\end{center}

We estimate ground state energy of various quantum chemistry molecules by decomposing their Hamiltonians into Pauli strings, estimating the Pauli string expectation values, and recombining the estimates. We will write this out in detail below. 
First, we decompose the qubitized Hamiltonian into Pauli strings:
    \begin{equation}
        H = \sum_P c_P P.
    \end{equation}
The Hamiltonians we consider are transformed into qubit Hamiltonians by the Jordan Wigner transformation, and as such the Pauli strings with nontrivial coeffcients often exhibit a highly non-local structure. Although this information is not necessary for (nor used by) the DSS algorithm, it's helpful background for the reader.
Next, assuming we have prepared the ground state on our quantum simulator, we estimate the ground state energy
    \begin{equation} \label{eqn:decompose_expval_H_intoPaulis}
        \expval{H} = \sum_P c_P \expval{P}.
    \end{equation}
by estimating the expectation values $\hat{o}_N(P) \approx \expval{P}$ of the relevant Pauli strings ($P$ such that $c_P \neq 0$). Then we can feed those estimates into \equref{eqn:decompose_expval_H_intoPaulis} and estimate $\expval{H}$. 
We run the DSS algorithm ($N=1000$, $d=1$) on the Pauli strings $P$ such that $c_P \neq 0$, and once we measure with the output measurement circuits, we can construct our ground state estimate $\hat{o}(H) \approx \expval{H}$. It takes the form
    \begin{equation}
        \hat{o}_N(H) = \sum_P c_P \hat{o}_N(P),
    \end{equation}
where each Pauli observable estimate $\hat{o}_N(P)$ is constructed using the procedure outlined in appendix \ref{app:estimating_exp_vals_with_DSS}. Our entire protocol for estimating the ground state energy has estimation error $|\expval{H} - \hat{o}_N(H)|$.

When we run the DSS algorithm for this task, we employ a slightly modified cost function compared to the one used when proving our guarantees in Appendix \ref{app:Performance_guarantees}. 
In that appendix our learning problem was to estimate our desired Pauli strings to the best precision possible, with the given measurement budget $N$. 
However, for this application, we care about learning some Pauli strings more than others -- in particular, we want to be more precise when estimating Paulis the contribute more to the sum in \equref{eqn:decompose_expval_H_intoPaulis}. 
Our Pauli strings relative importance is determined by the strength of the coefficients $c_P$. 
Therefore, we use the Pauli's coefficients to weight the \textsf{COST} function.
Here, setting the weights to be $w_P = |c_P|$, our \textsf{COST} function takes the form...
\begin{equation}
    \textsf{COST}_\textnormal{weighted} \left(\{\mathcal{U}_i\}_{i=1}^N\right) 
    = 2 \sum_P |c_P| \prod_{i=1}^{N} \exp{\left[ \hspace{1mm} -\frac{\epsilon^2}{2} p_i(P) \hspace{1mm} \right]}.
\end{equation}
Weighting
the cost function allows us to signpost to our algorithm the relative importance of our Pauli strings.  
Using this weighted cost function is crucial for our quantum chemistry application -- consider a Hamiltonian with a large set of coefficients on some Paulis and then a long tail of small coefficients.
A cost function that has more emphasis on the large-$c_P$ Pauli estimates will do a better job at estimating the corresponding ground state energy. 
Finally, while here we just use the absolute value of $c_P$, we emphasize that there are many ways one can fix the weights. For example, one could also consider using the square of the coefficients -- this would amplify the dependence on the relative strength of the coefficients. 
In general, the best strategy will depend on the application.

\addtocounter{figure}{+3}
\begin{figure*}[h]
\centering
\includegraphics[scale = 0.5]{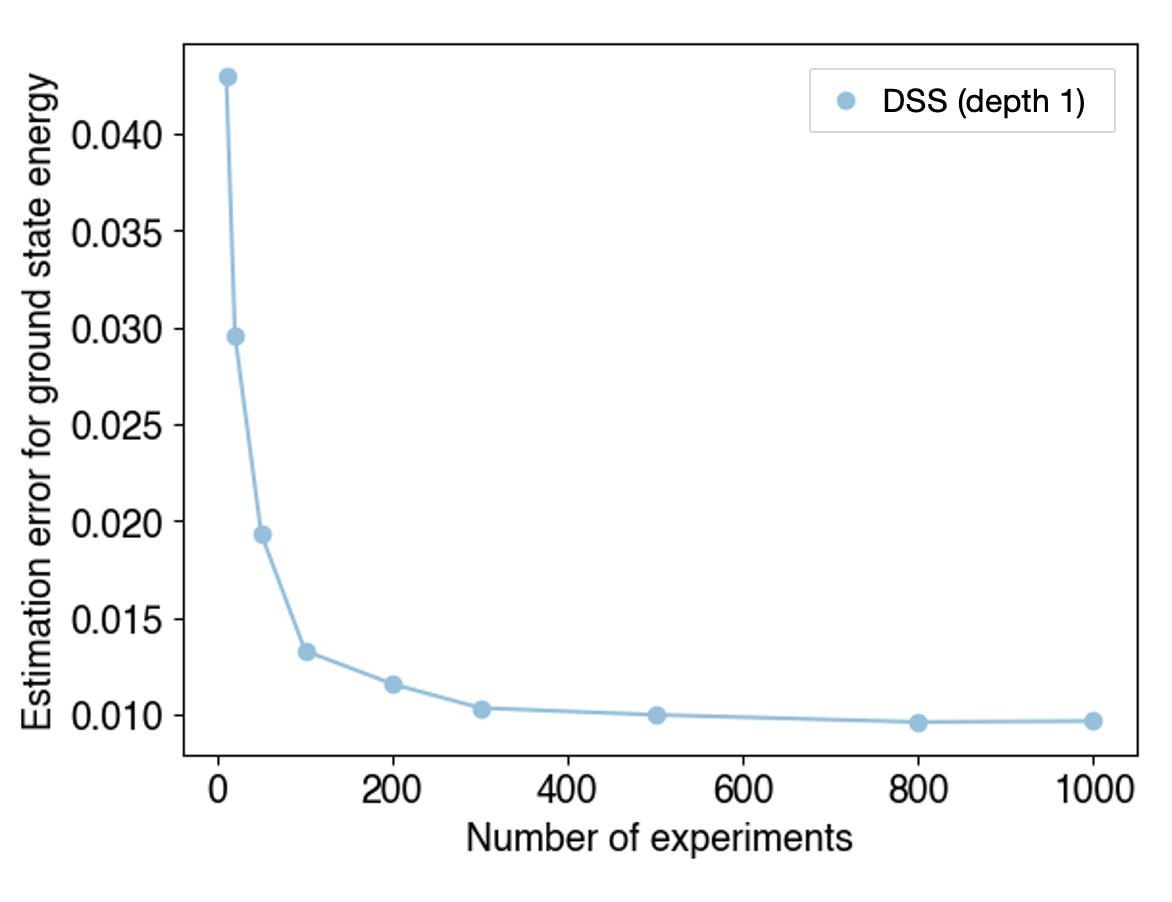}
\caption{
\emph{Estimation error decreases with more measurements.}  
We display the estimation error of depth-1 DSS, when estimating the ground state energy of the molecule $H_2$ on 4-qubits. As predicted by concentration inequalities (Appendix \ref{app:Performance_guarantees}), we find the estimation error decreases as we take more measurements. Each data point is average over 1000 simulations.
}
\label{fig:H2_estimation_error}
\end{figure*}

While the $N=1000$ results across all our test molecules can be found in the main text, let's consider an example molecule: $H_2$. 
Depending on how accurately we want to represent this molecule upon qubitization, we can represent it with more or less qubits. 
Here, consider the 4-qubit representation.
Using depth $d = 1$ DSS we can calculate the ground state energy estimation error using a variety of measurement budgets. See Supplementary Figure~\ref{fig:H2_estimation_error}. 
As expected, the estimation error decreases as a function of number of measurements.
 In this figure, for each measurement budget $N$ (i.e. number of experiments), we simulate many sets of measurements using the DSS output circuits $\{U_i^{DSS}\}_{i=1}^N$. 
This is because the variance of our protocol is nonzero. Simulating the $\{U_i^{DSS}\}_{i=1}^N$ measurement outputs gives some estimate for the ground state energy, but each of these measurements will not always give the same output. Therefore, with a finite measurement budget $N$, the quality of our energy estimate will fluctuate from simulation to simulation. 
As such, in order to get a good sense of the estimation error (how much our estimates fluctuate around the true value on average), we need to perform many simulations of taking the $\{U_i^{DSS}\}_{i=1}^N$ measurements. 
The next section discusses this and determines how many simulations we need to run to get our estimation error estimate to converge.

\vspace{5mm}
\noindent
\begin{center}
    \textit{Repeated simulations to show converging estimation error} 
\end{center}\label{app:converging_est_error}

In this section we discuss how many times we need to simulate the measurement procedure, using the set of circuits $\{U_i^{DSS}\}_{i=1}^N$ found by our DSS algorithm, in order to report a fair estimation error.
When estimating the ground state energy of our quantum chemistry Hamiltonians, we use $N = 1000$ measurements. DSS determines 1000 good measurement circuits for this task, reporting these bounded-depth measurement circuits at the output of the classical step (the DSS algorithm of appendix \ref{app:DSSalgorithm}). 
Upon simulating measurement with these $N = 1000$ measurements, our measurement outcomes will yield some estimate $\hat{o}(H)$ of the ground state energy $\expval{H}$. The estimation error is the difference between this estimate and the true energy, which we know because these molecules are small enough for full classical simulation. 
However, because this procedure has some non-trivial variance, sometimes our simulations may -- by chance -- give a very precise energy estimate. Other times, we may be far off from the true value.
In other words, when estimating the ground state energy using the measurement outputs of our DSS-given circuits, the estimation error will fluctuate due to our protocol's nontrivial variance. 
Therefore, we need to simulate taking measurements with our measurement circuits $\{U_i^{DSS}\}^{1000}_{i=1}$ many times in order to get a sense of what the \textit{average estimation error} looks like. 
The average estimation error is what we report in Figure \ref{fig:Figure2}.

\begin{figure*}[h]
\centering
\includegraphics[scale = 0.5]{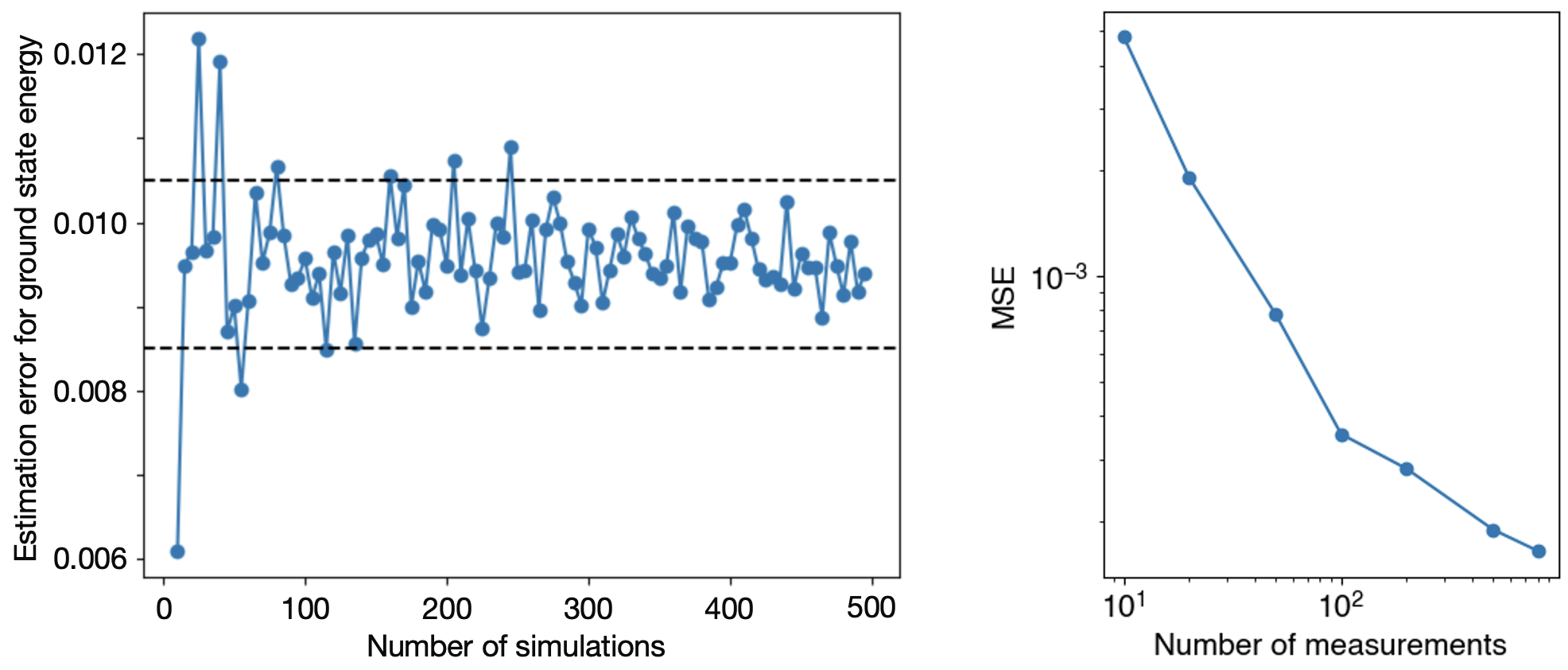}
\caption{
\emph{Characterizing the error of our ground state energy estimates of $H_2$ on 4 qubits.} 
(Left)
As we increase the number of simulations, we find that the estimation error (defined in \equref{eq:est_error}) converges towards an average value of $\sim 0.0096$ using the $N=1000$ DSS-specified circuits. 
The black dashed lines represent values that are $\pm 0.001$ away from our reported average estimation error $0.0096$.
(Right)  We can also perform a sanity check by plotting the mean squared error (MSE)  as a function of the number of measurements. 
As expected, MSE decreases as we ask our DSS algorithm to specify more measurement circuits for the estimation task. Here, we calculate the MSE using $S = 100$ simulations.
}
\label{fig:H2_simulations_converge}
\end{figure*}

Now let's determine how many simulations are needed to report a fair estimation error for DSS in Figure \ref{fig:Figure2}. 
In particular, for $S$ simulations of measurement circuits $\{U_i^{DSS}\}_{i=1}^N$, the average estimation error takes the form
    \begin{equation}\label{eq:est_error}
        \frac{1}{S} \sum_{s=1}^S \left| \expval{H} -\hat{o}_N^{(s)}(H) \right|.
    \end{equation}
Here $\hat{o}_N^{(s)}(H)$ is the the $s$th energy estimate, created from the $s$th simulation of measurements $\{U_i^{DSS}\}_{i=1}^N$.
For the molecules in Figure \ref{fig:Figure2}, we use $S = 500$ simulations when the number of qubits $n < 16$ and $S = 1000$ simulations otherwise. 
These numbers were chosen to guarantee estimation error convergence. 
For example, consider the molecule $H_2$ on 4 qubits, for which we report DSS has  
an average estimation error of 0.0096 under $S = 500$ simulations.
See Supplementary Figure \ref{fig:H2_simulations_converge}.
Notice that for $S \in [100,500]$ simulations the average estimation error oscillates in a small window (+/-0.001, represented by black dashed lines) around $0.0096$. 
In the data represented in main text Figure~\ref{fig:Figure2}, we plot the average estimation error with $S$ large enough that the the fluctuations have become small like this. This is how we decided to use $S = 500$ simulations when the number of qubits $n < 16$ and $S = 1000$ simulations otherwise.  
Finally, before we move onto the next section, we would be remiss to not mention that we considered an alternative way to quantify fluctuations in the estimation error: the mean squared error, defined as
    \begin{equation}\label{eq:mse}
        MSE = \frac{1}{S} \sum_{s=1}^S \left( \expval{H} -\hat{o}_N^{(s)}(H) \right)^2.
    \end{equation}
In Supplementary Figure \ref{fig:H2_simulations_converge}, we plot the mean squared error to look at how the average error trends, as a function of the number of measurements. These data points are averaged over $S = 100$ simulations.

\vspace{5mm}
\noindent
\begin{center}
    \textit{An example were DSS is optimal: 100 measurements on $H_2$} 
\end{center}\label{app:optimal_h2}

Here we discuss an example regime in which DSS gives optimal results. Consider $H_2$ on 4-qubits. For a budget of $N = 100$ depth $d = 1$ measurements, the 100 measurement settings chosen by DSS turn out to be optimal.
As shown in Figure \ref{fig:Figure2}(b), the DSS algorithm prescribes taking $91$ measurements in the all $Z$-basis and $9$ measurements in the double-Bell basis. For the latter, we rotate into the Bell basis on the first two qubits and on the last two qubits.
The first basis learns all strings with Identities and $Z$s, while the second learns Pauli strings with $XX$, $YY$, $ZZ$, and $II$ on either the first or last pairs.
Thus, these two measurement circuits together allow us to probe all Pauli strings, and the breakdown into $91$ versus $9$ measurements arises from the relative strength of the coefficients of the given Pauli strings. For the decomposition of our $H_2$ Hamiltonian, see below for the Pauli strings (and their corresponding coefficients). 
\begin{align*}
&ZIII \hspace{1mm} (0.172183) 
&\hspace{5mm}&  ZIZI \hspace{1mm} (0.168927)
&\hspace{5mm}&  YYXX \hspace{1mm} (0.045232)
\\
&IZII \hspace{1mm}  (-0.225753)   
&\hspace{5mm}&  ZIIZ \hspace{1mm} (0.166145)
&\hspace{5mm}&  YYYY \hspace{1mm} (0.045232)
\\
&IIZI \hspace{1mm}  (-0.172183)   
&\hspace{5mm}&  IZIZ \hspace{1mm} (0.174643)
&\hspace{5mm}&  XXXX \hspace{1mm} (0.045232)
\\
&IIIZ \hspace{1mm} (-0.225753)  
&\hspace{5mm}&   IZZI \hspace{1mm} (0.166145)
&\hspace{5mm}&   XXYY \hspace{1mm} (0.045232)
\\
&  
&\hspace{5mm}&   ZZII \hspace{1mm} (0.120912)
&\hspace{5mm}&   
\\
& 
&\hspace{5mm}&   IIZZ \hspace{1mm} (0.120912)
&\hspace{5mm}&   
\end{align*}

These measurement settings achieve a global minimum of the \textsf{COST} function landscape, when only allowing one layer $d=1$ of entangling gates. We will build some intuition for why.  
For the $d=1$ measurement ansatz, we have the option to add an entangling gate (\textsf{CNOT}) between qubits 1 and 2 and an entangling gate (\textsf{CNOT}) between qubits 3 and 4. Note that, since there's only one layer of entangling gates in this ansatz, a \textsf{SWAP} gate would be trivial and thus equivalent to  \textsf{Identity}.
We can now utilize our nonzero depth to simultaneously learn as many strings as possible. 
All strings commute with each other \textit{except} strings $ZIII$, $IZII$, $IIZI$, and $IIIZ$ (first column) do not commute with $XXXX$, $XXYY$, $YYXX$, and $YYYY$ (last column). The remaining strings (middle column) commute with both the first column and the last column.  
For the set in the last column, we can use two \textsf{CNOT} gates to rotate into a simultaneous eigenbasis. One option for combining single-qubit rotations with the \textsf{CNOT} is the double-Bell basis, which crucially also diagonalizes other Pauli strings we care about ($IIZZ$ and $ZZII$, from the middle column). Notice that because the two-qubit gate is on the first two (and, separately, the last two) qubits, we could never use our depth $1$ ansatz to simultaneously diagonalize $XX$, $YY$, and $ZI$ on the first two qubits, for example.
To find a basis for the remaining Pauli strings -- outside of those simultaneously diagonalized by the double-Bell basis -- we can simply notice that all remaining strings have $I$s and $Z$s. 
As such these two bases together can learn all Pauli strings in $H_2$.

Now that we have argued for the measurement bases, let's revisit the number of times we measure in each basis. Recall we have $N = 100$ measurements in our budget. 
Since the $XXXX$, $XXYY$, $YYXX$, and $YYYY$ strings have substantially smaller coefficients than the terms with $I$s and $Z$s, we only have to probe the double-Bell basis when the precision on the terms with $I$s and $Z$s are small enough such that the $X$ and $Y$ terms become relevant for our \textsf{COST} function.  
For a budget of $100$ measurements, 
we numerically investigated the optimal decomposition into these two measurements settings. We computed the $\textsf{COST}$ function associated with all combinations of these measurement bases, and we set the $\textsf{COST}$ hyperparameters to be the same as the ones we gave DSS. 
We found that 91 all $Z$ basis measurements and $9$ double-Bell measurements are a global minimum for the \textsf{COST} function. 


\newpage
\section{\label{app:HubbardModel}HUBBARD MODEL NUMERICAL SIMULATIONS}

This appendix discusses the simulations of Figure \ref{fig:Figure3}. 
Imagine the setting where we are trying to variationally prepare, for example, the ground state of a given Hamiltonian $H$. Here, we consider the Hamiltonian $H$ for the Fermi-Hubbard chain. 
We map the Hamiltonian Eq.~(\ref{eq:hubbard}) to a qubit lattice via a Jordan-Wigner transformation \cite{Sachdev_1999}, where a single site in the fermion lattice is mapped to a unit cell of two qubits. Our qubitized Hamiltonian then takes the form
\begin{equation}
\label{eqn:FERMIHUBBARDhamiltonian}
    H = -\frac{J}{2} \sum_{i=2}^{N - 1} \left( X_{i-1} Z_i X_{i + 1}+ Y_{i-1} Z_iY_{i+1}\right) + \frac{U}{4} \sum_{i=1}^{N/2} \left( \mathbb{1} - Z_{2i-1} \right) \left( \mathbb{1} - Z_{2i} \right) 
\end{equation}
where $N = 2L$ (see Eq.~\ref{eq:hubbard}). 
To determine whether we are in an eigenstate of this Hamiltonian, we want to estimate the variance.
The variance bounds the spectral distance to the closest energy eigenstate $\ket{\psi_{\ell}}$: $|\braket{H} - \varepsilon_{\ell}|^2 \le \text{var}[H]$ and thus serves for verification, when measured to a precision below the energy gap.
In order to estimate the variance, which takes the form $\expval{H^2} - \expval{H}^2$, we decompose it into Pauli strings as was done in Appendix \ref{app:QuantumChemistry} with the quantum chemistry Hamiltonians. Here,
    \begin{equation}
        \text{var}[H] = \expval{H^2} - \expval{H}^2
        = 
        \sum_P c_P \expval{P}.
    \end{equation}
Using \equref{eqn:FERMIHUBBARDhamiltonian}, we can determine the coefficients $c_P$ and the associated Pauli strings $P$ that we want to estimate. 
Once we have this set of Paulis $\{P\}$ we want to learn, we plot the number of measurements needed to learn each of these Pauli strings at least 25 times. 
As a result, for this application we do not use weights in our \textsf{COST} function, and so we only need to know which Pauli strings we care about -- i.e. those with $c_P \neq 0$.
Moreover, since one of the Pauli strings in our (qubitized) Hubbard $H$ is the all-Identity string, we only need to decompose $\expval{H^2}$ to obtain the set of Paulis we want to learn. The decomposition of $\expval{H^2}$ into strings will already contain all Pauli strings in  $\expval{H}$.

Since our learning problem has changed, we must also modify our DSS algorithm.
Recall that, given some measurement budget $N$, Algorithm \hyperref[alg:DSS]{1} (greedily) finds the best $N$ measurements. 
However, here we want to learn all of our Paulis $P$ at least $25$ times. In other words, our DSS algorithm should keep running until for all $P$, the hitting count $h(P) \geq 25$. See Appendix \ref{app:DSSalgorithm}, Definition \ref{def:hittingcount} for a formal definition of the hitting count. 
We want to know how many measurements $N$ we must make until this condition is satisfied. As such,  this requires a modification to the DSS algorithm as stated in Algorithm \hyperref[alg:DSS]{1}. 
In this appendix, we first discuss how to modify our DSS algorithm to accommodate this new learning problem, and while we chose the number 25, this number is simply an input. We set $N_O = 25$, where $N_O$ is the number of measurements \textit{per observable}.

After this appendix subsection, we next describe all methods plotted in Figure \ref{fig:Figure3}. 
In particular, we provide detailed background on the shallow shadows strategy and the naieve grouping strategy -- and why they scale as they do. 
Note that we also plotted the direct measurement strategy. 
However, this is trivial: this strategy simply measures each Pauli 25 times, requiring a number of measurements linear in the number of Paulis we want to learn (here, we need $N = 25 * |\{P\}_P|$ measurements).

\vspace{5mm}
\noindent
\begin{center}\label{app:ModifiedDSS_25mmts_perPauli}
    \textit{DSS Setup: 25 measurements per Pauli} 
\end{center}

In this subsection we discuss how to modify our DSS Algorithm, as stated in Algorithm \hyperref[alg:DSS]{1}, to accommodate the following setup: we want to measure each Pauli string in our set of interest $\{P\}$ at least $N_O$ times. And in particular, for Figure \ref{fig:Figure3}, we want to determine how many measurements are needed for this task. 
For protocols that effectively utilize each measurement, this will take fewer total shots than protocols that are ineffective -- for example, directly measuring each Pauli string requires $N = N_O * |\{P\}_P|$ measurements. However, this is the most inefficient strategy. 
To determine how many (and which) measurements to make for DSS, we have to modify DSS so it continues finding the next best measurement until we have learned each Pauli string at least $N_O$ times. 
We do so with the following two algorithmic updates...
    \begin{enumerate}
        \item We replace the for loop (line 4) by a while statement: while $\exists P \in \{P\}$ such that $h(P) < N_O$.
        \item  We also update our \textsf{COST} function: Once we have achieved $N_O$ measurements for some Pauli string $P$, we no longer care about learning $P$ and should remove it from our set of interest. As a result, we no longer include it in our \textsf{COST} function -- we only include the Pauli strings which we still need to learn.  
    \end{enumerate}

The first update to the DSS Algorithm is a straightforward change to line 4, but modifying the cost function requires more finesse. 
We update our \textsf{COST} function (see the original \textsf{COST} in Definition \ref{def:DSScost_appendix}) to only consider the Paulis of interest that also have not already been measured $N_O$ times, and therefore, our function becomes 
\begin{equation}
    \textsf{COST}_{N_O} \left(\{\mathcal{U}_i\}_{i}\right) = 2 \sum_{\substack{P \textnormal{ s.t. }\\ h(P) < N_O}} 
    \prod_{i=1}^{N_O * |\{P\}|} \exp{\left[ \hspace{1mm} -\frac{\epsilon^2}{2} p_i(P) \hspace{1mm} \right]}.
\end{equation}  
Here, again $\epsilon$ is a hyperparameter and $p_i(P)$ is the Pauli weight of $P$ under the ensemble defined by measurement ensemble $\mathcal{U}_i$. 
Moreover, notice the other change to this function -- we have replaced $N$ in the product to be $N_O * |\{P\}|$. This is because we do not actually know how many measurements are needed to learn all our strings $N_O$ times, and $N_O * |\{P\}|$ measurements is an upper bound because this is the number of direct measurements. 
One could also imagine this number to be another hyperparameter that can be tuned -- for example, if we expect the number of DSS measurements to be close to but below some other, tighter value, we can replace $N_O * |\{P\}|$ with this value.


\vspace{5mm}
\noindent
\begin{center}\label{app:ShallowShadowsBackground}
    \textit{Background on Shallow Shadows}
\end{center}

Randomized measurements, or classical shadows, enable the prediction of many properties of arbitrary quantum states using few measurements. 
While random single-qubit measurements are experimentally friendly and suitable for learning low-weight Pauli observables, they perform poorly for nonlocal observables \cite{huang2020predicting}. 
It has been shown that the shadow norm for predicting nonlocal Pauli observables scales as $3^k$, where $k$ is the number of non-trivial Pauli operators in the Pauli string \cite{huang2020predicting}. 
It can be understood using simple probabilistic arguments: for random single-qubit Clifford gates, it has an equal probability to rotate the measurement basis from $Z$ to $X$, $Y$, or $Z$. Therefore, for each qubit, we measure the $X,Y, \text{ and } Z$ basis with equal probability. And for a Pauli string operator with $k$ non-trivial Pauli operators, the probability of directly measuring it is $1/3^k$. Since Pauli operators have bounded norm, the shadow norm \cite{huang2020predicting} (maximum variance over states) is inversely proportional to this probability, i.e. $3^k$. Since Clifford circuits in general will map Pauli string operators to Pauli string operators, we can generalize this probabilistic intuition to other random Clifford circuit measurements. 

If one directly measures each qubit in the $Z$ basis, then all Pauli strings consisting solely of $I$ and $Z$ operators will be measured with probability one, while all others will have zero probability. There are $2^N$ total such operators consisting solely of $I$ and $Z$. If one applies a Clifford circuit before the $Z$ basis measurement, this unitary will transfer the $2^N$ operators only containing $I$ and $Z$ to a general set of Pauli strings that mutually commute with each other. This set is also called a \emph{stabilizer group}. Therefore, one can measure $2^N$ mutually commuting Pauli strings simultaneously with probability one with a certain Clifford circuit transformation. And for the other $4^N-2^N$ Pauli strings, this probability is zero, as they not measured.
Now imagine multiple Clifford circuits could be applied with some probability before the $Z$ basis measurement, the probability of being measured associated with each Pauli string will become a number between $0$ and $1$, i.e. $p(P)\in [0,1]$. Notice this is different than only being either $0$ or $1$ for a single Clifford transformation. And the shadow norm becomes $\norm{P}_{\text{shadow}}\propto 1/p(P)$. 

In addition to random Pauli and Clifford measurements, classical shadows generated by shallow brickwork circuits, \emph{shallow shadows}, have recently garnered significant attention. 
As one scales the circuit depth $d$ from 0 to $\infty$, the shallow shadow can extrapolate between the random Pauli measurement limit and the random Clifford measurement limit.
However, they are popular due to the intermediate regime between these two extremes: at the shallow circuit region, when circuit depth is logarithmic $d = \mathcal{O}(\log k)$, they perform almost optimally for learning all $k-$local random Pauli operators \cite{hu2023classical,ippoliti2023operator,bertoni2022shallow}. 
At this depth the shadow norm scaling is upper bounded by $\mathcal{O}(2.28^k)$, which is much better than the scaling $\mathcal{O}(3^k)$ of the random Pauli measurements \cite{ippoliti2023operator}.

One can understand this phenomenon with the following physical picture: consider the operator size under the shallow unitary transformation.
In this picture there are two competing forces: (1) information scrambling and (2) cancellation of Pauli operators. Let's first discuss information scrambling. When $d=0$, we randomly measure the single-qubit $X,Y,Z$ basis. However, when there is a brickwork circuit, the size of the single-qubit $X,Y,Z$ operator will grow because of the information scrambling. This light cone is linear in terms of circuit depth, and so the probability of directly measuring a particular Pauli operator $\hat{P}$ will decay. However, there is also a chance that there are some operators in the middle will be mapped to identity operators -- i.e. not all qubits within the light cone have non-trivial Pauli operators. This is the second competing force -- the ``cancellation of Pauli operators.'' This cancellation effect will increase the probability of directly measuring $k-$local Pauli operators. Interestingly, in \cite{ippoliti2023operator}, the authors show in the shallow circuit region, the cancellation effect will dominate such that the shadow norm will decay first before increasing for larger circuit depth. In addition, it has been recently shown that the lower bound for learning all $k-$local Pauli operators is $\Omega(2^k)$ \cite{2024arXiv240419105C}.

Finally, in Figure \ref{fig:Figure3}(b) we compare the performance of DSS to optimal-depth shallow shadows, where the optimal depth is a single layer of entangling gates. We know this is the optimal depth for the task of measuring each Pauli of interest $25$ times because we simulated shallow shadows with up to four layers of entangling gates. We found that beyond one layer of entangling gates, the shots -- that shallow shadows required for the task -- started increasing.

\vspace{5mm}
\noindent
\begin{center}\label{app:NaieveGroupingBackground}
    \textit{Naive grouping strategy for measuring the energy variance of the Hubbard model}
\end{center}

In this section, we provide additional information on the naive Pauli grouping scheme for measuring the energy variance of the Fermi-Hubbard chain as presented in Fig.~\ref{fig:Figure3} (b). The naive grouping approach involves organizing the Pauli strings into sets that can be simultaneously diagonalized using single-qubit rotations. As we will show below, within this scheme, the number of sets scales linearly with increasing number of particles.
We start by considering the composition of the transformed, qubitized Hamiltonian (from \equref{eqn:FERMIHUBBARDhamiltonian} above), which reads
\begin{align}
\begin{split}
    H &= -\frac{J}{2} \sum_{i=2}^{N - 1} \left( X_{i-1} Z_i X_{i + 1}+ Y_{i-1} Z_iY_{i+1}\right) + \frac{U}{4} \sum_{i=1}^{N/2} \left( \mathbb{1} - Z_{2i-1} \right) \left( \mathbb{1} - Z_{2i} \right) \\ &= H_J + H_U.
    \end{split}
\end{align}
Evaluating the energy variance involves estimation of expectation values of Pauli strings contained in the squared Hamiltonian $H^2 = H_J^2 + H_U^2 + \{ H_J, H_U\}$. While naively this would require a number of measurement bases that scales quadratically with the system size, the structure of the problem allows for more efficient grouping schemes. For example, the operator $H_U^2$ is diagonal in the Z-basis, enabling its evaluation through sampling in a single basis.

\addtocounter{figure}{+5}
\begin{figure*}[h]
\centering
\includegraphics[scale = 0.3]{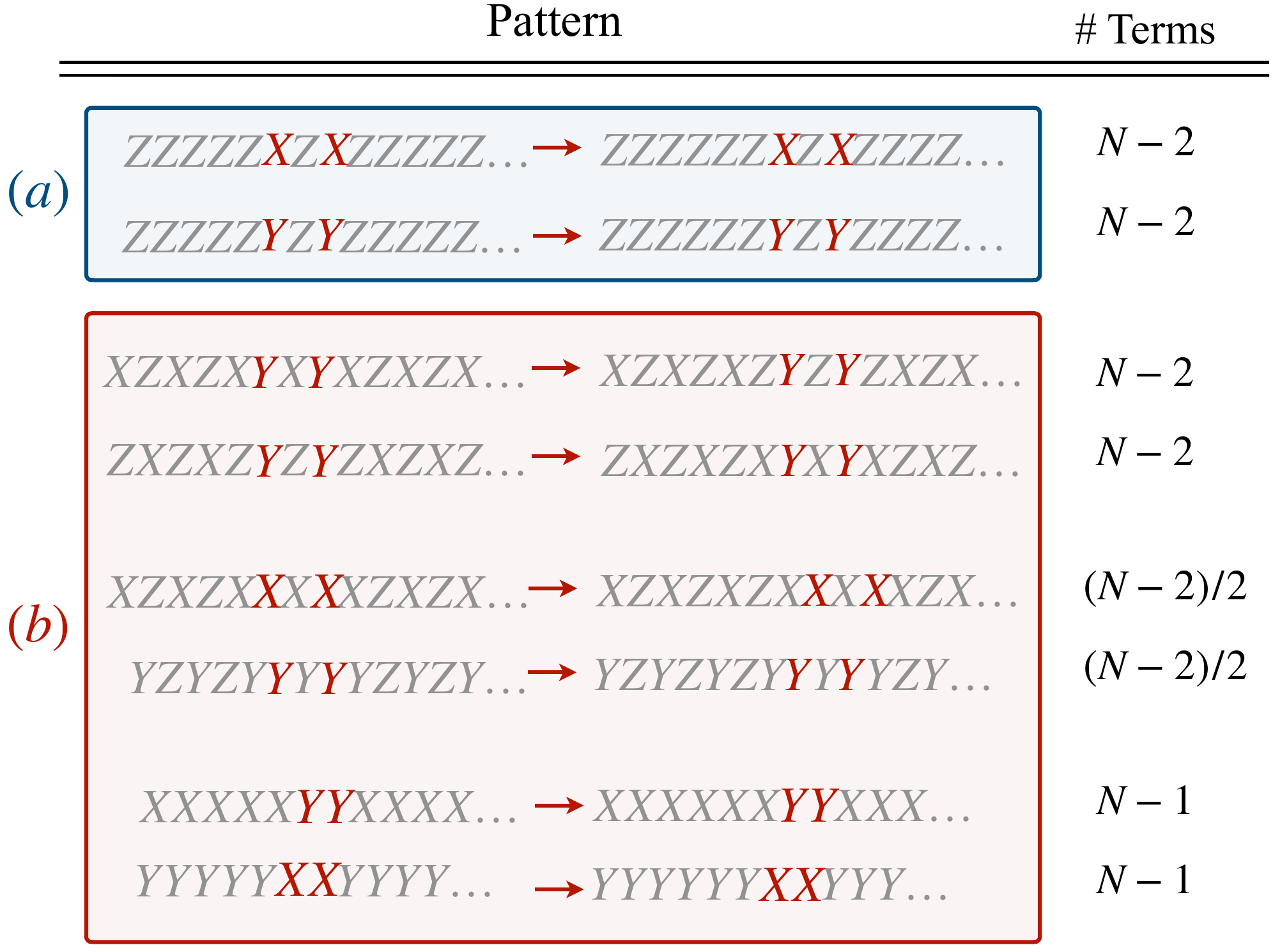}
\caption{
\emph{Naive grouping scheme for measuring the Pauli strings of the squared Hubbard Hamiltonian} (a) The $2(N-2)$ bases for measuring all terms of the anticommutator $\{H_J, H_U\}$ (see main text). The next-nearest-neighbor XX and YY operators are swept through the chain, resulting in $2(N-1)$ bases. (b) The required bases for measuring the square of the kinetic part of the Hubbard Hamiltonian: $H_J^2$.
}
\label{fig:naivegrouping}
\end{figure*}

We now focus on measuring expectation values of Pauli strings contained in the anticommutator $\{ H_J, H_U \}$. This operator is composed of correlation functions between tunneling terms $\left( X_{i-1} Z_i X_{i + 1}+ Y_{i-1} Z_iY_{i+1}\right)$ and doublon-densities $\left( \mathbb{1} - Z_{2i-1} \right) \left( \mathbb{1} - Z_{2i} \right)$, with their number scaling quadratically with system size. However, by transforming into the eigenbasis of
XZX (or YZY) for a single bond and measuring the remaining qubits in the Z-basis, we can measure a linear number of these correlation functions in parallel. As indicated in Fig.~\ref{fig:naivegrouping} (a), we can measure all correlation functions appearing in $\{ H_J, H_U \}$, by sweeping XZX (YZY) bases through the chain, resulting in a number $2(N-2)$ measurement bases.    

A similar scheme can be applied for measuring all correlation functions contained in $H_J^2$, which includes Pauli strings of weight up to 6. The patterns for the different basis are summarized in Fig.~\ref{fig:naivegrouping} (b). In conclusion, a total number of $7N - 11$ bases are required to measure the variance of the Hubbard Hamiltonian. We note that the prefactor of the linear scaling may be reduced further by removing additional redundancies in the patterns depicted in Supplementary Figure ~\ref{fig:naivegrouping}.

\end{appendix}

\end{document}